\documentclass[11pt]{article}
\pdfoutput=1

\usepackage[utf8]{inputenc}
\usepackage{amsmath, amsfonts, amssymb, amsthm, mathtools}
\usepackage{slashed}
\usepackage{nicefrac}
\usepackage{mathrsfs}

\usepackage{multirow}
\usepackage{comment}
\usepackage{graphicx, import}
\usepackage{psfrag}
\usepackage[usenames,dvipsnames,svgnames,table]{xcolor}
\usepackage{enumerate}
\usepackage{arydshln}
\usepackage{soul}
\usepackage{array}
\usepackage{caption}
\usepackage{subcaption}
\usepackage{colortbl}
\usepackage{float}
\usepackage{mdframed}
\usepackage{setspace}
\usepackage{pifont}
\usepackage{newunicodechar}

\usepackage{tikz}
\usepackage{tikz-cd}
\usetikzlibrary{shapes.geometric}

\usepackage[most]{tcolorbox}
\usepackage{newfloat}
\DeclareFloatingEnvironment{boxes}

\newcommand{\boxref}[1]{\hyperref[{#1}]{Box~\ref*{#1}}}

\newcommand{\beq}{\begin{equation}}  \newcommand{\eeq}{\end{equation}}
\newcommand{\bal}{\begin{aligned}}   \newcommand{\eal}{\end{aligned}}
\newcommand{\be}{\begin{equation}}
\newcommand{\ee}{\end{equation}}

\def\bea{\begin{eqnarray}}
\def\eea{\end{eqnarray}}

\newcommand{\cT}{\mathcal{T}}

\newcommand{\dd}{\mathrm{d}}

\def\simleq{\; \raise0.3ex\hbox{$<$\kern-0.75em
      \raise-1.1ex\hbox{$\sim$}}\; }
\def\simgeq{\; \raise0.3ex\hbox{$>$\kern-0.75em
      \raise-1.1ex\hbox{$\sim$}}\; }

\numberwithin{equation}{section}

\definecolor{newBLUE}{RGB}{0,102,255}
\definecolor{newGREEN}{RGB}{67,168,0}
\definecolor{dark-gray}{gray}{0.20}
\definecolor{gray}{gray}{0.30}
\definecolor{light-gray}{gray}{0.80}
\definecolor{dark-red}{rgb}{0.7,0,0}
\definecolor{dark-green}{rgb}{0.1,0.4,0}
\definecolor{dark-blue}{rgb}{0.3,0.3,0.7}
\definecolor{light-blue}{rgb}{0.8,0.8,1}
\definecolor{swamp}{RGB}{240, 199, 197}
\definecolor{landscape}{RGB}{180, 250, 199}
\definecolor{undecided}{RGB}{252, 252, 197}

\usepackage{jheppub}
\usepackage{cleveref}

\hypersetup{
	colorlinks=true,
	linkcolor=dark-blue,
	citecolor=dark-red,
	urlcolor=dark-green,
	linktoc=page
}

\theoremstyle{remark}

\crefname{appendix}{Appendix}{Appendices}
\theoremstyle{fact}
\newtheorem{fct}[equation]{Fact}

\title{\centering EFT (String) Tower Building}

\author{Alessandra Grieco$^{1,2}$, Ignacio Ruiz$^{3}$}
\affiliation{$^{1}$Instituto de F\'{i}sica Te\'{o}rica IFT-UAM/CSIC, C/ Nicol\'{a}s Cabrera 13-15, Campus de Cantoblanco, 28049 Madrid, Spain}
\affiliation{$^2$Departamento de F\'{i}sica Te\'{o}rica, Universidad Aut\'{o}noma de Madrid, Cantoblanco, 28049 Madrid, Spain}

\author{and Irene Valenzuela$^{1,2,3}$}
\affiliation{$^3$CERN, Theoretical Physics Department, 1211 Meyrin, Switzerland}

\emailAdd{alessandra.grieco@estudiante.uam.es}
\emailAdd{ignacio.ruiz.garcia@cern.ch}
\emailAdd{irene.valenzuela@cern.ch}

\preprint{IFT-UAM/CSIC-26-96\\\ \vspace*{-0.1cm} 
\hfill  CERN-TH-2026-174}

\abstract{We develop a bottom-up framework to reconstruct the asymptotic spectrum of light towers in four-dimensional $\mathcal{N}=1$ effective field theories directly from their K\"ahler potential. Our construction uses the Integral Scaling Relation, which relates the mass scales of towers becoming light at infinite distance to the tensions of EFT strings, together with the Emergent String Conjecture applied recursively upon decompactification. These conditions organize the tower scaling vectors into a lattice generated by the EFT-string vectors and select the globally consistent tower polytopes. When applied to K\"ahler potentials arising from string compactifications, our algorithm precisely reproduces the known arrangements of towers and duality frames. We then classify the admissible polytopes for asymptotic K\"ahler potentials $K\sim-\log P(s)$, with $P(s)$ a homogeneous polynomial of degree at most seven, and show that certain apparently consistent K\"ahler potentials are incompatible with the assumed quantum-gravity constraints. For general polynomials, additional restrictions arise from gluing the tower arrangements across different growth sectors. Remarkably, every tower polytope allowed by our reconstruction can be obtained as a concrete slice of the polytope associated with M-theory on a Joyce $G_2$-manifold. This provides evidence for a form of string universality in which EFT strings act as the fundamental building blocks of the UV tower structure in the EFT perturbative regimes.
 }

\begin{document}
\hypersetup{pageanchor=false}
\makeatletter
\let\old@fpheader\@fpheader
\makeatother

\maketitle

\hypersetup{
    pdftitle={EFTstrings},
    pdfauthor={Alessandra Grieco, Ignacio Ruiz, Irene Valenzuela},
    pdfsubject={EFT strings and Duality frames}
}

\newcommand{\remove}[1]{\textcolor{red}{\sout{#1}}}

\section{Introduction}

Four-dimensional $\mathcal{N}=1$ effective field theories (EFTs) provide a natural starting point for connecting string theory to real-world physics. But does every 4d $\mathcal{N}=1$ EFT admit a consistent UV completion in string theory? Lessons from the Swampland program \cite{Vafa:2005ui,Brennan:2017rbf,Palti:2019pca,vanBeest:2021lhn,Grana:2021zvf,Harlow:2022ich,Agmon:2022thq} strongly suggest that the answer is no: non-trivial consistency conditions must be satisfied for an EFT to arise from quantum gravity, and thus, from string theory. In this work, we will use recent results from string compactifications—originally uncovered in tests of Swampland conjectures—to constraint the perturbative form of K\"ahler potentials that are UV consistent.

The physics of 4d $\mathcal{N}=1$ EFTs is particularly constrained in perturbative regimes exhibiting (approximate) continuous axionic shift symmetries. In such regimes, one can construct $\frac12$-BPS axionic strings magnetically charged under the corresponding axions. The backreaction of these strings induces non-trivial profiles for their saxionic partners. For K\"ahler potentials of logarithmic form, as commonly encountered in string compactifications, these profiles define distinguished trajectories in field space that drive some of the saxions towards infinite distance. Those strings whose associated trajectories run deeper into the perturbative regime—so that the non-perturbative corrections breaking the axionic symmetry remain suppressed—are known as EFT strings. Their physical properties imply that they become tensionless in Planck units at the corresponding infinite-distance limits.

 The identification of EFT strings in a broad range of 4d $\mathcal{N}=1$ theories arising from string compactifications has proved particularly fruitful in connection with the Distance Conjecture \cite{Lanza:2021udy,Lanza:2022zyg,Grimm:2022sbl,Wiesner:2022qys,Marchesano:2023thx,Marchesano:2024tod,Hassfeld:2025uoy,Grieco:2025bjy,Monnee:2025ynn,Kaufmann:2026tsy,Kaufmann:2026mha}; see also \cite{Marchesano:2022avb,Martucci:2022krl,Cota:2022yjw,Marchesano:2022axe,Martucci:2024trp,Licciardello:2026roi} for other applications. The Distance Conjecture \cite{Ooguri:2006in} is one of the central proposals of the Swampland program. It states that every infinite-distance limit in the moduli space of an EFT admitting a consistent quantum-gravitational UV completion must be accompanied by an infinite tower of states whose characteristic mass scale decreases exponentially with the geodesic field-space distance in Planck units. This behavior has been verified in a wide variety of string compactifications (see, e.g. \cite{Baume:2016psm,Klaewer:2016kiy, Blumenhagen:2017cxt, Grimm:2018ohb,Heidenreich:2018kpg, Blumenhagen:2018nts, Grimm:2018cpv, Buratti:2018xjt, Corvilain:2018lgw, Joshi:2019nzi,  Erkinger:2019umg, Marchesano:2019ifh, Font:2019cxq,  Lee:2019wij,Gendler:2020dfp, Lanza:2020qmt, Klaewer:2020lfg, Lee:2021qkx,Lee:2021usk,Alvarez-Garcia:2023gdd,Alvarez-Garcia:2023qqj,Rudelius:2023mjy,Aoufia:2024awo,Hassfeld:2025uoy,Monnee:2025msf}), as well as in AdS/CFT \cite{Baume:2020dqd,Perlmutter:2020buo,Baume:2023msm,Ooguri:2024ofs,Calderon-Infante:2024oed,Calderon-Infante:2026rkj, Mantegazza:2026spd, Baume:2026svx,Calderon-Infante:2026zmd}.

Remarkably, an analysis of numerous string-theory examples \cite{Lanza:2021udy} revealed that the mass scale $m_\star$ of the leading tower becoming light is related to the tension $\mathcal{T}$ of EFT strings, a relation that was subsequently tested further in \cite{Grieco:2025bjy}. More precisely, along the scalar trajectory sourced by a given EFT string, the leading tower satisfies the integral relation

\begin{equation}\label{scalingweight intro}
\left(\frac{m_{\star}}{M_{\rm Pl,4}}\right)^2\sim \left(\frac{\mathcal{T}}{M_{\rm Pl,4}^2}\right)^w\to0\,,\quad\text{with }\; w=1,2,3\,.
\end{equation}

This can equivalently be expressed as the following condition on the corresponding \textit{scaling vectors} \cite{Lee:2018spm,Gonzalo:2019gjp,DallAgata:2020ino,Andriot:2020lea,Benakli:2020pkm,Etheredge:2023odp,Etheredge:2024tok}:

 where the derivatives are taken with respect to the moduli. Using this,  \eqref{scalingweight intro} becomes
 \begin{equation}
		\vec\zeta_{*} \cdot \vec\zeta_{\cT_{\mathbf{e}}}=w|\vec\zeta_{\cT_{\mathbf{e}}}|^2 \ , \quad \text{ with } w=1,\,2,\,3\;,
\end{equation}

defined as

\begin{equation}
\vec{\zeta}_i=-\vec\nabla\log\frac{m_i(\vec{\varphi})}{M_{\rm Pl,4}}\,,\qquad \vec{\zeta}_{\mathcal{T}}=-\vec\nabla\log\frac{\mathcal{T}(\vec{\varphi})^{1/2}}{M_{\rm Pl,4}}\,,
\end{equation}
where the derivatives and metric involved are taken to be those of the 4d moduli space. For later reference, we refer to this proposal as the \textbf{Integral Scaling Relation}.\footnote{Sometimes also referred to as \textit{Integer Scaling Conjecture}, see, e.g., \cite{Grieco:2025bjy}, or as \emph{Conjecture 2: cut-off asymptotics} in \cite{Lanza:2021udy}.}

The significance of this relation is that it connects an infrared quantity that can be determined directly from the K\"ahler potential—the EFT-string tension—to ultraviolet data: the masses of the towers of states required by quantum-gravity consistency. The latter would ordinarily be accessible only once the microscopic UV completion is known.

In previous work \cite{Grieco:2025bjy}, we tested this relation further in top-down string constructions and found that it is satisfied not only by the leading tower along a given asymptotic trajectory, but also by subleading towers that become dominant in other directions. These are precisely the towers whose scaling vectors generate the convex hull of the towers of states \cite{Calderon-Infante:2020dhm,Etheredge:2022opl,Etheredge:2024tok,Calderon-Infante:2023ler}. Moreover, the relation implies that the light towers are organized into a lattice structure determined by the EFT-string vectors, as discussed in \cite{Grieco:2025bjy}. Interestingly, similar lattice structures relating the mass of light towers and the tension of extended objects have recently been discussed in the context of the Distance Conjecture, see \cite{Etheredge:2024amg,Etheredge:2025ahf}.

The purpose of the present paper is to reverse this logic. We assume that the Integral Scaling Relation is a universal consistency condition for 4d $\mathcal{N}=1$ EFTs and investigate, from a bottom-up perspective, how strongly it constrains both their asymptotic K\"ahler potentials and their UV spectra. In particular, we address the following questions:
\begin{itemize}
\item Given a 4d K\"ahler potential, can the Integral Scaling Relation fully determine the asymptotic masses of the UV towers of states that should emerge at its infinite-distance limits?
\item Is every asymptotic K\"ahler potential compatible with the Integral Scaling Relation, or can the conjecture be used to rule out K\"ahler potentials that cannot arise from a consistent quantum-gravitational UV completion?
\end{itemize}
In a nutshell, we find that the Integral Scaling Relation, when combined with the Emergent String Conjecture, is sufficiently restrictive to provide definite answers to both questions. The Emergent String Conjecture (ESC) \cite{Lee:2019xtm,Lee:2019wij} refines the Distance Conjecture by specifying the microscopic nature of the leading tower: it must arise either from a decompactification limit, giving a Kaluza–Klein tower, or from the oscillator modes of a fundamental critical string. Under certain conditions, this restricts the allowed exponential decay rates of the tower masses, see e.g., \cite{Etheredge:2022opl}. Combined with the Integral Scaling Relation, these restrictions single out specific scaling vectors for the UV towers that precisely match those found in explicit string compactifications. The first question therefore has a positive answer:  the asymptotic K\"ahler potential contains enough information to reconstruct the relevant UV tower data upon imposing the above conjectures.

The answer to the second question is negative. Not every K\"ahler potential gives rise to a structure of EFT strings and UV towers compatible with both the ESC and the Integral Scaling Relation. Subject to the universality of these conjectures, our results therefore provide a criterion for ruling out certain perturbative asymptotic forms of the K\"ahler potential from admitting a UV completion in quantum gravity.

Finally, we investigate the interplay between the Integral Scaling Relation and the taxonomy rules derived in \cite{Etheredge:2024tok}. We find that the former is more restrictive: certain arrangements of UV towers that satisfy the taxonomy rules for particles are nevertheless incompatible with the Integral Scaling Relation. Remarkably, for K\"ahler potentials defined by a logarithm of homogeneous polynomials up to degree seven, every arrangement of UV towers compatible with both the Integral Scaling Relation and the ESC can be realized in string theory. More specifically, these arrangements arise as different slices of the 4d $\mathcal{N}=1$ moduli spaces obtained from M-theory compactified on Joyce $G_2$ manifolds. Within the class of models studied here, this establishes a form of string universality for the boundaries of 4d moduli spaces: every asymptotic structure of UV towers allowed by the proposed quantum-gravity consistency conditions admits an embedding into string theory/M-theory. Such universality does not emerge if one imposes only the Distance Conjecture and the ESC.\\

The outline of the paper is as follows. We start with a review of EFT string solutions from purely 4d EFT data in Section \ref{sec.reviewEFT strings}. In Section \ref{sec.reconstr} we outline our assumptions and the bottom-up reconstruction algorithm of the light tower arrangement given a K\"ahler potential. We then use it in Section \ref{sec.constrain kahler} to put constrains on K\"ahler potentials and possible tower arrangements, finding that all possible consistent convex hulls are realized in top-down string/M-theory examples. We finish in Section \ref{sec.concl} with a summary and outlook. Appendix \ref{app.id} offers a review of some basic results on the $\zeta$-vectors and scaling weights of EFT strings, followed by a review of the top-down duality frames and tower arrangement for 4d $\mathcal{N}=1$ compactifications of string and M-theory from \cite{Grieco:2025bjy} in Appendix \ref{app.topdown}.

\section{Review of EFT strings\label{sec.reviewEFT strings}}

\subsection{EFT strings in 4d \texorpdfstring{$\mathcal{N}=1$}{N=1} theories}

Before describing the interplay with the towers of light states, in this section we will review the construction of EFT strings as $\frac{1}{2}$BPS solutions to the 4d $\mathcal{N}=1$ effective action, as described in \cite[Section 2]{Lanza:2020qmt} and \cite{Greene:1989ya}. Take for this the set of chiral fields $\{t^j=a^j+it^j\}_{j=1}^n$  and consider for simplicity a vanishing scalar potential. 
The relevant part of the bosonic action reads
\begin{equation}
S_{\rm 4d}\supseteq\frac{M_{\rm Pl,4}^2}{2}\int\left\{\mathcal{R}\star1-2K_{i\bar j}\dd t^i\wedge\star\bar t^j\right\}\;,\quad\text{with }\;K_{i\bar j}=\partial_{t^i}\partial_{\bar t^j}K\,,
\end{equation}

where $K$ is the K\"ahler potential. In those regions of the moduli space where the action exhibits some approximate axionic shift symmetry preserved at perturbative level, we can construct string-like solutions magnetically charged under the axions. These are given by
\begin{equation}\label{eq.EFT string metric}
\dd s^2=-\dd \tau^2+\dd x^2+e^{2D(z,\bar{z})}\dd z\dd\bar{z}\,,
\end{equation}
 where $\{\tau,x\}$ parametrize the string worldsheet and $\{z,\bar{z}\}$ its traverse directions. A simple solution for the $\{t^j:\mathbb{C}\to\mathcal{M}\}_{j=1}^n$ is given by holomorphic profiles $\bar{\partial}t^j=0$ (here $\partial=\partial_z\wedge \dd z$), while Einstein equations require $e^{2D(z,\bar{z})}=|f(z)|^2e^{-K}$ with $f(z)$ holomorphic. Requiring invariance under the axionic discrete shift symmetry implies a monodromy $t^j\to t^j+e^j$ as one goes around the codimension-2 string (which for simplicity we will locate at $z=0$).  Here $e_i\in\mathbb{Z}$ are the string charges. This results in the following local solution
 \begin{equation}\label{eq.profilesEFT}
 t^j(z)=t^j_0+\frac{e^j}{2\pi i}\log\Big(\frac{z}{z_0}\Big)\,,\quad\text{so that}\quad\left\{\begin{array}{l}
 a^j=a^j_0+\frac{\theta}{2\pi}e^j\\
 s^j=s^j_0-\frac{e^j}{2\pi}\log\Big(\frac{r}{r_0}\Big)
 \end{array}\,,\right.
 \end{equation}
 with $t_0^j=a_0^j+is^j_0$ and $z=re^{i\theta}$. For later convenience, let us define $\sigma=-\frac{1}{2\pi}\log\Big(\frac{r}{r_0}\Big)$. 
 
 Transverse to the string, the configuration has an energy per unit length given by
 \begin{align}
 \mathcal{E}&=M_{\rm Pl,4}^2\int_{\mathbb{C}}K_{j\bar{k}}\Big(\partial t^i\wedge\star_{\mathbb{C}}\bar\partial\bar{t}^j+\bar\partial t^i\wedge\star_{\mathbb{C}}\partial\bar{t}^j\Big)\notag\\
 &=M_{\rm Pl,4}^2\underbrace{\int_{\mathbb{C}}iK_{j\bar{k}}\dd t^j\wedge\dd \bar{t}^k}_{\rm topological}+2M_{\rm Pl,4}^2\int_{\mathbb{C}}\underbrace{K_{j\bar{k}}\bar\partial t^j\wedge\star_{\mathbb{C}}\partial\bar{t}^k}_{\geq 0}\,,
 \end{align}
 where the topological term is set by boundary conditions on $K$. Note that, seeing $K_{j\bar k}$ as a pullback from the moduli space to the transverse $\mathbb{C}$, the above expression takes the form of a BPS bound, which is saturated only if $\bar{\partial}t^j=0$, i.e., the chiral scalars take holomorphic profiles. This then means that solutions \eqref{eq.profilesEFT} will preserve half of the supercharges, and will therefore correspond to $\frac{1}{2}$BPS-strings.\\
 
From \eqref{eq.profilesEFT}, we see that, close to the core of the EFT string, the saxionic component of the chiral scalars diverge, 
\begin{equation}\label{eq.trajectory}
	t^i\to e^i\cdot\infty\;\Longleftrightarrow\; s^i\to e^i\cdot\infty\quad\text{and }\; a^i\text{ constant}\,, \qquad\text{as }\; r\to 0\;,
\end{equation}
so that it drives the solution towards a boundary of the moduli space. This saxionic path (namely $s^i=s^i_0+e^i\sigma$ with $\sigma\to\infty$) is usually referred to as the \emph{string flow}.

Along these limits, the K\"ahler potential $K$ can become approximately invariant under continuous axionic shifts $a^i\to a^i+c^i$. This is the case if, e.g.,
\begin{equation}\label{K}
	K=-\log\Big[P(s)+\mathcal{O}\big(e^{2\pi im_jt^j}\big)\Big]\,,
\end{equation}
where $\mathcal{O}\big(e^{2\pi im_jt^j}\big)$ are instantonic non-perturbative corrections. For the string solution to be valid, these instantonic solutions must be sufficiently suppressed, which bring us to the definition of EFT strings:

\begin{mdframed}
\textbf{EFT strings:}  $\frac{1}{2}$BPS axionic strings with charges $e_i$ such that $e^im_i\geq 0$ for all allowed instanton charges $m_i$. This way all non-perturbative corrections will be suppressed along the limit \eqref{eq.trajectory} (i.e., when $s^j\to\infty$) and we will remain in the perturbative regime of the moduli space as we approach the string core.
\end{mdframed}
We can define the perturbative regime (i.e., the regime where non-perturbative corrections to the K\"ahler potential and superpotential can be neglected) as the deep interior of the saxionic cone $\Delta$ given by
\beq
\Delta\equiv \{s^i\in N_{\mathbb{R}} \, |\, m_is^i>0, \forall m_i\in \mathcal{C}_I\}
\eeq
where $\mathcal{C}_I$ is the set of instanton charges and $N_\mathbb{R}$ is some real vector space of the appropriate dimension. The above definition of EFT strings implies that cone of EFT string charges $\mathcal{C}_S^{\rm EFT}$ is the discretization of the saxionic cone, namely
\beq
\mathcal{C}_S^{\rm EFT}=\bar\Delta\cap N_{\mathbb{Z}}
\eeq
where $\bar \Delta$ is the closure of $\Delta$ and $N_\mathbb{Z}$ is a lattice obtained by a discretization of $N_\mathbb{R}$. Therefore, the EFT string charges act as generators of the saxionic cone parametrizing the EFT perturbative regime. The string flow \eqref{eq.profilesEFT} gets mapped to a trajectory in moduli space which bring us deeper into the perturbative regime. For K\"ahler potentials of the form \eqref{K}, the limit \eqref{eq.trajectory} will be at infinite distance in moduli space. Notice that not every $\frac{1}{2}$BPS string is an EFT string, as its charges might not satisfied that $e^im_i\geq 0$, so that the string flow would bring us away from the perturbative regime.\\

To determine the string tensions, it is useful to consider a dual formulation of the EFT. In $d=4$, axionic scalars are dual to 2-forms, through
\begin{equation}
\dd B_{2\,j}=-M_{\rm Pl,4}^2\mathsf{G}_{jk}\star\dd a^k,
\end{equation}
where $\mathsf{G}_{jk}=\frac{1}{2}\partial_{s^j}\partial_{s^k}K$
defines the moduli space metric of the chiral scalars in terms of the saxions along these limits. Through the metric \eqref{eq.EFT string metric} and the $\bar{\partial}t^j=0$ condition, we can write 
\begin{equation}
	\dd B_{2\,j}=-M_{\rm Pl,4}\dd t\wedge\dd x\wedge(\mathsf{G}_{ij}\dd s^j)=M_{\rm Pl,4}\dd t\wedge\dd x\wedge\dd l_j\,,
\end{equation}
where $l_j=-\frac{1}{2}\partial_{s^j}K$ are known as \emph{dual saxions}. From this is easy to identify $B_{2\,j}=M_{\rm Pl,4}^2l_j\dd t\wedge \dd x$. Since 4d strings can be charged under (a linear combination of) these 2-forms, their worldvolume $\mathcal{W}$ action is given by
\begin{equation}
	S_{\rm string}=e^i\int_\mathcal{W}B_{2\,i}-\int_\mathcal{W}\mathcal{T}_{\mathbf{e}}\star_{\mathcal{W}}1\,.
\end{equation}
From the BPS condition, this translates in the EFT string tension being given by
\begin{equation}\label{eq.EFTstringTension}
	\mathcal{T}_{\mathbf{e}}=M_{\rm Pl,4}^2e^il_i=-\frac{M_{\rm Pl,4}^2}{2}e^i\partial_{s^i}K\,.
\end{equation}
Note that this derivation is valid for any $\frac{1}{2}$BPS strings in this 4d $\mathcal{N}=1$ theory, which implies that any other string whose charge or tension is not given by the above expression is necessarily non BPS.\\

Using \eqref{K}, EFT strings satisfy that their tension vanishes along the string flow as 
\be\label{univtensionflow}
\cT_{\bf e}\sim \frac{M^2_{\rm P}}{\sigma}\rightarrow 0\, ,\quad~~~ \text{for $\sigma\rightarrow\infty$}\, .
\ee
From a string theory perspective, we can identify the microscopic origin of these EFT strings, which -- as expected from the above analysis -- will indeed always become tensionless in Planck units at the infinite distance boundaries of the moduli space. This was first done in plethora of 4d $\mathcal{N}=1$ string compactifications in  \cite{Lanza:2021udy}, and used to provide a bottom-up explanation for the EFT breakdown predicted by the Swampland Distance Conjecture \cite{Ooguri:2006in}, as the string tension \eqref{univtensionflow} provides an upper bound for the EFT cut-off, $\Lambda_{\rm cut-off}\lesssim  \cT^{1/2}$. In general, there can also be other towers of states becoming light asymptotically, but it was found in  \cite{Lanza:2021udy} that they always satisfy the following integral scaling
\begin{align}
\label{scalingweight}
m_*^2\simeq M_{\rm Pl,4}^2\left(\dfrac{\cT_{\bf e}}{M_{\rm Pl,4}^2}\right)^w \quad \text {with } w=1,2,3\quad
\text{ as }s^i=s^i_0+e^i\sigma\text{ with }\sigma\rightarrow \infty \ .
\end{align}
where $m_*$ is the mass scale of the leading tower of states and $w$ is denoted as the scaling weight. This \textbf{Integral Scaling Relation} or Conjecture was further analyzed in plethora of 4d $\mathcal{N}=1$ string compactifications in \cite{Grieco:2025bjy}, where it was found that also the subleading towers of states (below the species scale) satisfy such relation with the EFT strings.

Following  \cite{Grieco:2025bjy}, one can re-write \eqref{scalingweight} as a local condition in moduli space by defining the following scaling vectors
\begin{equation}\label{eq.zeta vec}
	\vec{\zeta}_i=-\vec\nabla\log\frac{m_i(\vec{\varphi})}{M_{\rm Pl,4}}\,,\qquad \vec{\zeta}_{\mathcal{T}}=-\vec\nabla\log\frac{\mathcal{T}(\vec{\varphi})^{1/2}}{M_{\rm Pl,4}}
\end{equation}
 where the derivatives are taken with respect to the moduli. Using this,  \eqref{scalingweight} becomes
 \begin{equation} \label{eq.intscaling}
		\vec\zeta_{*} \cdot \vec\zeta_{\cT_{\mathbf{e}}}=w|\vec\zeta_{\cT_{\mathbf{e}}}|^2 \ , \quad \text{ with } w=1,\,2,\,3\;,
		\end{equation}
		where the scalar product is made with respect to the moduli space metric $\mathsf{G}_{jk}=\frac{1}{2}\partial_{s^j}\partial_{s^k}K$. This implies that the $\zeta$-vectors of all the light towers of states\footnote{More specifically, it applies to all towers of states with a mass scale below or at the species scale, which generate the tower convex hull \cite{Grieco:2025bjy}} necessarily lie at points of the lattice generated by EFT string vectors, as checked in plethora of string theory examples in \cite{Grieco:2025bjy}. In Figure \ref{f.chullex} we show a qualitative example, where the red arrows are EFT string vectors and the blue dots the  $\zeta$-vectors of the towers. In the same colors, we also show the convex hulls of the vectors associated to the EFT strings or the towers of particles. Notice that the later appear to be located on a lattice generated by the $\zeta$-vectors of EFT strings charged along single saxionic directions, as indicated by \eqref{eq.intscaling}.

 \begin{figure}[!ht]
\begin{center}
\includegraphics[width=0.55\textwidth]{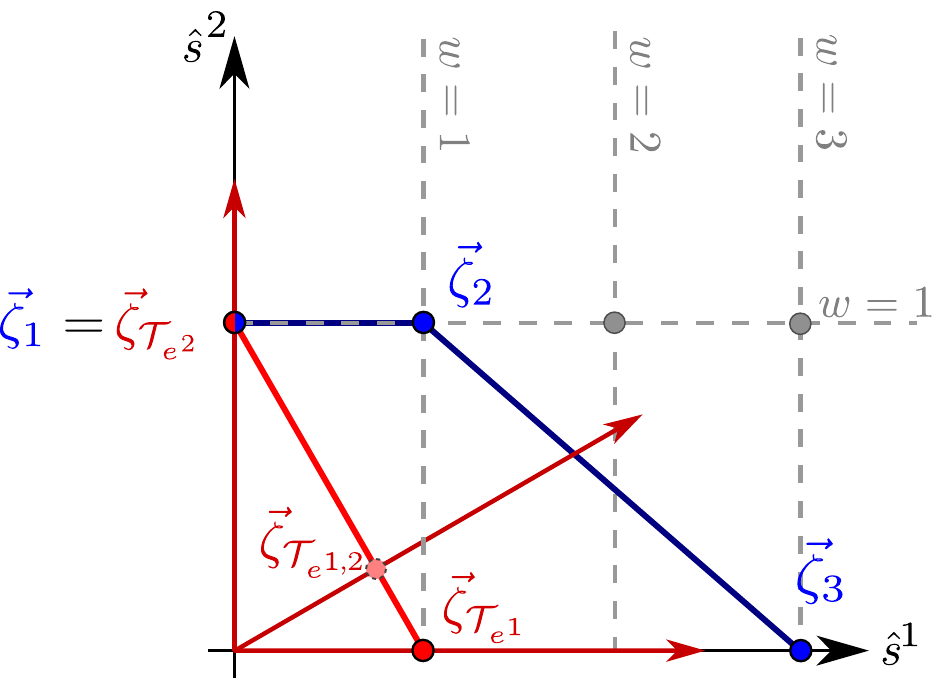}
\caption{Convex hull of $\zeta$-vectors (in blue) corresponding to light towers of states in a given growth sector. In red, $\zeta$-vectors corresponding to EFT strings generating the lattice on which light towers live. Notice that not all points are populated: the ones in grey would correspond to decompactification to a non-integer number of dimensions.
\label{f.chullex}}
\end{center}
\end{figure}

 In this paper, we are going to reverse the story. Rather than performing a top-down analysis to check \eqref{eq.zeta vec} (as done in   \cite{Grieco:2025bjy}), we are going to \emph{assume} \eqref{eq.zeta vec} (or equivalently, \eqref{scalingweight}) as a universal principle of quantum gravity, and derive the possible UV structures of towers consistent with that relation from a bottom-up perspective.  Hence, we will start with generic K\"ahler potentials in supergravity, from which we can derive the EFT string tensions purely from a bottom-up perspective, and put constraints on the UV spectrum of the theory using \eqref{eq.zeta vec}. We will see that the results match the specific structures of UV towers observed in string theory, providing evidence for a sort of string universality.

 \subsection{A comment on 4d \texorpdfstring{$\mathcal{N}=2$}{N=2} theories}
 
The above discussion applies to 4d $\mathcal{N}=1$ theories, where the moduli space $\mathcal{M}$ is naturally endowed with a K\"ahler structure in terms of a K\"ahler potential $K$ determining the metric $\mathsf{G}_{i\bar j}=\partial_i\partial_{\bar j} K$ and the K\"ahler form $\omega=\frac{i}{2}\partial\bar\partial K$. As we have seen above, the EFT string tensions can be computed directly from this K\"ahler potential (i.e., purely EFT data) through \eqref{eq.EFTstringTension}.\\
 
 Now, while in general this is no longer the case for 4d $\mathcal{N}=2$ theories, one can still define a K\"ahler potential function of a subsector of the scalars. The moduli space of 4d $\mathcal{N}=2$ supergravity factorizes into a \emph{vector} and \emph{hypermultiplet} sectors, see e.g., \cite[Chapter 20]{Freedman:2012zz}
 \begin{equation}
 	\mathcal{M}=\mathcal{M}_{\rm V}\times\mathcal{M}_{\rm H}\,,
 \end{equation}
 respectively \emph{(projective) special K\"ahler} and \emph{quaternionic K\"ahler} manifolds of dimensions $2n_{\rm V}$ and $4n_{\rm H}$. The vector sector is parameterized by projective complex coordinates $\{X^I\}_{I=0}^{n_{\rm V}}$, which together with the holomorphic prepotential $\mathcal{F}$ define a symplectic section $\Omega=(X^I,F_I)^{\intercal}$, with $F_I=\partial_{X^I}\mathcal{F}$. Then the K\"ahler potential associated to $\mathcal{M}_{\rm V}$ is defined as \cite{Strominger:1990pd}
 \begin{equation}\label{eq.kahlerprepot}
 	K_{\rm V}=-\log\big(i\langle\Omega |\bar\Omega\rangle\big)=-\log\big[i(\bar{X}^IF_I-X^I\bar{F}_I)\big]\,.
 \end{equation}
 Once $K_{\rm V}$ is known, one can compute the tension of $\frac{1}{2}$-BPS strings charged with respect to the physical scalars $z^I=\frac{X^I}{X^0}$ in a similar way as in the 4d $\mathcal{N}=1$ case through the use of \eqref{eq.EFTstringTension}.
 
 The hypermultiplet case is generally more involved, since $\mathcal{M}_{\rm H}$ is not K\"ahler but \emph{quaternionic K\"ahler}, endowed with a local quaternionic algebra in terms of a $Sp(1)\simeq SU(2)$ structure $\{\vec{J}_X\}_{X=1}^{n_{\rm H}}$ with $[J^a,J^b]=2\epsilon^{abc}J^c$. Generally, we will not be able to define a K\"ahler potential from which to derive the tension of an EFT string from the bottom-up, see \cite[Section 20.2]{Freedman:2012zz} for how to define the metric on $\mathcal{M}_{\rm H}$. Only in particular cases where the $\mathcal{M}_{\rm H}$ manifold presents some isometries (which can be understood as axionic shift symmetries) one can can rewrite the kinetic term of the quaternionic moduli as saxion-axion pair, the latter of which can be dualized to a 2-form under which 4d strings are charged, as well as defining an effective K\"ahler potential for said pair(s), see e.g., \cite{deWit:2001brd,cortes2021symmetries} for more details.
 
 For top-down constructions coming from compactifications of type II string theory on a CY 3-fold $X$, at the classical level $\mathcal{M}_{\rm H}$ is fibered over $\mathcal{M}_{\rm C} (X)$ (the complex structure moduli space of $X$) in type IIA, and over $\mathcal{M}_{\rm K}({X})$ (the moduli space of the the complexified K\"ahler structures, isometric to $\mathcal{M}_{\rm C} (\hat X)$ though Mirror Symmetry) in type IIB, see \cite{Alexandrov:2013yva} for a review. Both of $\mathcal{M}_{\rm C} (X)$ and $\mathcal{M}_{\rm K}({X})$ admit a K\"ahler potential,\footnote{
For  $\mathcal{M}_{\rm C} (X)$ in the type IIA side this is given by $K_{\rm C}=-\log\left(i\int_X\Omega\wedge\bar{\Omega}\right)$ (where $\Omega$ is the holomorphic (3,0)-form), while $\mathcal{M}_{\rm K}({X})$ in type IIB admits a classic prepotential given by $\mathcal{F}=-\kappa_{abc}\frac{X^a X^b X^c}{6X^0}$ (modulo a symmetric quadratic form that does not affect the resulting K\"ahler potential), where $\kappa_{abc}$ are the triple intersection numbers of the 3-fold, from which the K\"ahler potential $K_K$ can be obtained through \eqref{eq.kahlerprepot}.} which modulo an overall rescaling, factorizes in the perturbative metric of $\mathcal{M}_{\rm H}$. The inclusion of quantum corrections through D-instantons generally spoils such clean structure, the same way as in 4d $\mathcal{N}=1$ for the K\"ahler potential in \eqref{K}. We thus conclude that in generic asymptotic regimes of the 4d $\mathcal{N}=2$ moduli space (both in the vector and hypermultiplet sectors) where non-perturbative contributions are suppressed, it is expected that we can define some effective K\"ahler potential on a subsector of the moduli through which the tension of $\frac{1}{2}$-BPS strings can be computed.

 \vspace{0.25cm}
EFT strings have been studied in explicit 4d $\mathcal{N}=2$ scenarios in various instances in the literature, see e.g., \cite{Marchesano:2022avb,Marchesano:2022axe,Grimm:2022sbl,Marchesano:2023thx,Marchesano:2024tod,Hassfeld:2025uoy,Monnee:2025ynn}. As appreciated in, e.g., \cite{Castellano:2023jjt}, the tower arrangements for compactifications down to 4d $\mathcal{N}=2$ and $\mathcal{N}=1$ theories are analogous, even if the microscopic interpretation of the light degrees of freedom varies. While the exhaustive checks from \cite{Grieco:2025bjy} focused in 4d $\mathcal{N}=1$, we expect the Integral Scaling Relation \eqref{scalingweight} to also hold for similar constructions to 4d $\mathcal{N}=2$ with the same effective K\"ahler potential $K\sim -\log P(s)$. Thus, the lessons regarding the reconstruction of towers and duality frame arrangements along this paper should apply to 4d theories with both 4 and 8 supercharges.

\section{Bottom-up reconstruction of string theory dualities}\label{sec.reconstr}

Different string/M-theory compactifications to 4d $\mathcal{N}=1$ supergravity can yield, from the top-down point of view, the same K\"ahler potential with asymptotic form $K\sim -\log P(s)$. These different top-down constructions do not necessarily feature the same arrangement of light towers and types of duality frames along the asymptotic regions of the saxionic cone. For example, as previously shown in \cite{Grieco:2025bjy}, the K\"ahler potential $K\sim -\log[s^0(s^1)^3]$ can be obtained from two families of top-down compactifications. These are respectively given by $E_8\times E_8$ heterotic and type IIA string theories on one hand, and $SO(32)$ heterotic/type I and type IIB string theory/F-theory on the other. The arrangement of towers is different between these two families, but equivalent inside each family (so that both heterotic $SO(32)$ and type I theory, for example, will lead to the same convex hull of $\zeta$-vectors for a given K\"ahler potential). Clearly, the microscopic origin of the light towers and of the EFT strings can be very different even between theories in the same family.  It can also be explicitly checked that the Integral Scaling Relation \eqref{scalingweight} (equivalently \eqref{eq.intscaling} using the $\zeta$-vectors language) holds between any tower generating the convex hull and elementary EFT string, with $w\in\{0,1,2,3\}$. The precise $w$ values for the different towers and EFT strings are the same for each of the two families mentioned above.

The fact that only these two types of arrangements are obtained for the $K\sim -\log[s^0(s^1)^3]$ K\"ahler potential is no coincidence. Indeed, we will show that these are the \emph{only} tower/duality arrangements \emph{consistent} with the Integral Scaling Relation \eqref{scalingweight} and the Emergent String Conjecture for this choice of K\"ahler potential. By reversing the logic, from a bottom-up point of view, we can wonder which tower a duality arrangements are consistent with the Integral Scaling Relation. Note that, since this is formulated in terms of the EFT string tension (which in turn only depends on the K\"ahler potential), this approach relies only on 4d $\mathcal{N}=1$ data, and thus is independent of the microscopic interpretation of the towers or other details of the top-down realization. Then, one could wonder whether all tower/duality arrangements consistent with the Integral Scaling Relation are realized in string and M-theory compactifications. As we will see, this seems to be the case.

In Section \ref{sec.reconstructionalg}, we will outline how to approach this question. Given an asymptotic expression for the K\"ahler potential $K\sim -\log P(s)$, we develop an algorithm that, under some assumptions, returns the possible arrangements for tower/duality frames which are consistent with respect to the Integral Scaling Relation and the ESC recursively. Later in Section \ref{ss.recovering}, we show how, for the choices of K\"ahler potentials analyzed in the top-down constructions studied in \cite{Grieco:2025bjy}, the consistent tower/duality arrangements are in one-to-one correspondence with those observed in string and M-theory compactifications.

\subsection{The reconstruction algorithm\label{sec.reconstructionalg}}
In order to ``reconstruct'' the set of possible consistent tower/duality frames along a given 4d $\mathcal{N}=1$ theory from given a K\"ahler potential $K\sim -\log P(s)$, the following working assumptions are required:
\begin{enumerate}
\item \textbf{The Integral Scaling Relation} \cite{Lanza:2021udy,Lanza:2022zyg}
\begin{equation}\label{e.bottom up integer}
	 	\left(\frac{m_\star}{M_{\rm Pl,4}}\right)^2\sim\left(\frac{\mathcal{T}_{e^i}}{M_{\rm Pl,4}^2}\right)^w
	 	\;\Longleftrightarrow\;{\vec{\zeta}_\star\cdot\vec{\zeta}_{\mathcal{T}_{e^i}}}=w|\vec{\zeta}_{\mathcal{T}_{e^i}}|^2\;,\quad\text{with}\quad w\in\mathbb{Z}_{\geq 0}
	 \end{equation}
\textbf{holds for \emph{all} towers $m_\star$ generating the $\zeta$-vector convex hull and elementary EFT strings $\{\mathcal{T}_{e^i}\}_{i=1}^n$} (this is, with charged only along the $s^i$ direction). As explained in \cite{Grieco:2025bjy} and Appendix \ref{app.id}, the $\zeta$-vector of non-elementary strings $\vec{\zeta}_{\mathcal{T}_{\mathbf{2}}}$, charged along several saxionic directions, ${e}^i\neq 0$ if $i\in I$, is located along the simplex 
\begin{equation}
	\Delta_I=\left\{\sum_{i\in I}\alpha_i\vec{\zeta}_{\mathcal{T}_{e^i}}\;:\;\sum_{i\in I}\alpha_i=1\;\text{and}\;\alpha_i\geq 0\right\}
\end{equation}
generated by the elementary EFT string vectors charged under $I$. The actual position of such vector $\vec{\zeta}_{\mathcal{T}_\mathbb{e}}\in \Delta_I$ depends on the actual asymptotic trajectory taken, having $\vec{\zeta}_{\mathcal{T}_\mathbb{e}}\perp \Delta_I$ when moving along the flow \eqref{eq.trajectory} induced by such string. Under such trajectory, the Integral Scaling Relation \eqref{e.bottom up integer} also follows, but at the same rate (and $w$) as the individual elementary EFT strings in $I$. Thus, it is enough to impose \eqref{e.bottom up integer} for elementary EFT strings, for which the Integral Scaling Relation holds along \emph{any} trajectory.
\item \textbf{The Emergent String Conjecture (ESC) \cite{Lee:2019xtm,Lee:2019wij} holds}. This is, along any infinite distance limit the lightest tower must correspond to either \textbf{(i)} Kaluza-Klein tower associated to some growing dimension(s) or \textbf{(ii)} a tower of oscillator modes of a  \emph{unique critical, weakly coupled fundamental string} (not necessarily BPS). This forbids, e.g., excitation modes from branes different than the F1-string, see \cite{Alvarez-Garcia:2021pxo}. As we will see now, this allows us to strongly constrain the scaling rates of the different light towers.
\item \textbf{All possible warping along the decompactifying internal spaces is diluted as the new dimensions grow, in such a way that the resulting theory is a higher-dimensional Minkowski spacetime}.\footnote{In the language of \cite{Raucci:2026fzp}, for a compact space $\dd s_{d+n}^2=G_{MN}\dd x_M\dd x^N=e^{2\rho(y)}\eta_{\mu\nu}\dd x^\mu\dd x^\nu+e^{2\sigma(y)}h_{ij}\dd y^i\dd y^j$, respectively describing the external $d$-dimensional spacetime and the $n$-compact dimension, we ask for $\rho$ and $\sigma$, as well as other possible scalar profiles over the internal dimension, to have asymptotically vanishing gradients with respect to $\vec{y}$ in the decompactification limit.} This allows for a simple universal expression regarding the norm of the $\zeta$-vector associated to the light towers (corresponding to the exponential rate at which they become light when moving in their gradient flow direction), given by \cite{Etheredge:2022opl,Etheredge:2024tok}
\begin{equation}\label{eq.length}
	 	|\vec{\zeta}_{{\rm KK},n}|=\sqrt{\frac{d+n-2}{n(d-2)}}=\sqrt{\frac{n+2}{2n}}\;,
	 \end{equation}
 where $n$ is the number of decompactifying dimensions and we have set the space-time dimension to $d=4$ in the last step. Moreover, $n=\infty$ recovers the expected result for the oscillator modes of an emergent string, given by
\begin{equation}\label{eq.lengthosc}
	 	|\vec{\zeta}_{{\rm osc}}|=\frac{1}{\sqrt{d-2}}=\frac{1}{\sqrt{2}}\;.
	 \end{equation}
  Since at most we expect to have a Lorentz-invariant vacuum in $D=10$ (for string theory) or $D=11$ from M-theory, this sets $n\in\{1,2,\dots,7,\infty\}$ as the possible values for $n$. This assumption is needed in order to avoid warped decompactifications to running solutions as in \cite{Etheredge:2023odp}, for which the exponential rate is not universal and depends on the warping potential \cite{Raucci:2026fzp}, and generically results in smaller $|\vec{\zeta}_{{\rm KK},n}|$ values.\footnote{It would be interesting to study limits of 4d $\mathcal{N}=1$ theories involving warped compactifications. One possibility would be to start with heterotic string theory on $\mathbb{S}^1$ (a theory with 16 supercharges that has a warped decompactification limit \cite{Etheredge:2023odp}), and then further reduce to 4d on some $Y_5$. Since the resulting $X_6=\mathbb{S}^1\times Y_5$ has $b_1(X_6)\neq 0$, the total $X_6$ space cannot have $SU(3)$ holonomy, see e.g., \cite{Candelas:1985en}, and thus at best one could take $SU(2)$ holonomy and get to 4d $\mathcal{N}=2$. In order to break the supersymmetry down to $\mathcal{N}=1$, one could consider that $Y_5$ supports some NSNS $H_3$ flux, see, e.g., \cite{Fernandez:2008pf}. Due to the bottom-up nature of this paper, we will take as an assumption that all limits are unwarped, though it would be interesting to check whether the Integral Scaling Relation holds in these settings.}
\end{enumerate}
Importantly, we are assuming that the ESC holds \emph{recursively}, so that it should also hold in the higher dimensional theories obtained upon decompactification. This way, further infinite distance limits present in the higher-dimensional theory in $d'=4+n_1$ should contain a new tower (KK from decompactifying to $d''=4+n_1+n_2$ dimensions or string oscillator modes) becoming light. For iterated decompactifications $4\to 4+n_1\to\dots\to 4+\sum_i n_i$, this ensures the existence of KK bound states in the lower 4-dimensional theory. In other words, if we move in the 4d theory along the direction in moduli space at which all the individual internal submanifolds grow homogeneously, we should find a KK tower associated to $n_{\star}=\sum_i n_i$ growing directions. It can be computed that the associated $\zeta$-vector is given by \cite{Castellano:2022bvr,Castellano:2023jjt,Etheredge:2024tok}
\begin{equation}
	 	\vec{\zeta}_{\star}=\frac{1}{n_{\star}}\sum_ in_i\vec{\zeta}_{{\rm KK},n_i}\;,
	 \end{equation}
	 where again $n_i=\infty$ for oscillator modes of an emergent fundamental string. The above relation implies that $\vec{\zeta}_{\star}$ lies at the closest point to the origin in the $\{\vec{\zeta}_{{\rm KK},n_i}\}_i$ polytope generated by the convex hull of the leading towers \cite{Etheredge:2024tok}, and thus scales at the same rate as the individual towers. Note that, from the 4d perspective, some of the $\vec{\zeta}_{{\rm KK},n_i}$ towers may point outside the saxionic cone, and thus the asymptotic direction defined by its $\zeta$-vector takes us out of the perturbative description, while the bounded states with other towers still results in $\vec{\zeta}_\star$ being inside the saxionic cone and should be taken into account.
	 
	 Another important consideration is the fact that fundamental strings can be EFT strings (if BPS) and have their oscillator modes as light towers. To see this, take the NSNS sector (metric $g_{\mu\nu}^{\rm (s)}$, 4d dilaton $\phi_4$ and Kalb-Ramond field $B_2$) of the 4d bosonic action in the string frame and go to Einstein frame through a $g^{\rm (s)}_{\mu\nu}=e^{2\phi_{4}}g^{\rm (E)}_{\mu\nu}$ Weyl transformation:
	 \begin{align}
	 S_{\rm 4d}&\supset\frac{1}{2\kappa_4^2}\int \dd^4 x\sqrt{-g^{\rm (s)}}e^{-2\phi_{\rm 4}}\left[\mathcal{R}^{\rm (s)}+4(\partial\phi_{4})^2-\frac{1}{2}|\dd B_2|^2\right]\notag\\
	 &=\frac{1}{2\kappa_4^2}\int \dd^4 x\sqrt{-g^{\rm (E)}}\Big[\mathcal{R}^{\rm (E)}-2(\phi_4)^2-\frac{1}{2}e^{-4\phi_4}(\partial\sigma)^2\Big]\notag\\
	 &=\frac{1}{2\kappa_4^2}\int\left[\mathcal{R}^{\rm (E)}\star 1-2K_{0\bar{0}}\dd t^0\wedge\star\dd \bar{t}^0\right]\,,
	 \end{align}
	 where we have dualized to the NSNS axion (i.e., $\dd \sigma=\star_{\rm 4d} \dd B_2$)  and defined the following universal chiral scalar and K\"ahler potential:
	 \begin{equation}
	 t^0=e^{-2\phi_4}+i\sigma\,,\quad K=-\log(T+\bar{T})=2\phi_4-\log 2\,.
	 \end{equation}
	 If the fundamental string is BPS, and thus charged under the $B_2$ field, we can follow a similar procedure as in Section \ref{sec.reviewEFT strings}, with a profile for $t^0=e^{-2\phi_4}+i\sigma$ given by \eqref{eq.profilesEFT} with charge $\mathbf{e}=(1,0,\dots,0)$. The fundamental string has thus an interpretation as an EFT string in 4d. The EFT computation results in a tension given by
	 \begin{equation}
	 	\mathcal{T}_{\rm F1}=M_{\rm PL,4}^2l_0=\frac{M^2_{\rm Pl,4}}{2}e^{2\phi_4}\,,
	 \end{equation}
	 precisely as expected for a fundamental string in terms of the 4d dilaton. Since the mass of its oscillator modes is given  $m_{\rm osc}\sim \sqrt{\mathcal{T}_{\rm F1}}$, from \eqref{scalingweight} we have $w=1$ for the scaling weight between EFT string and tower. We thus have that any BPS F1-string is an EFT string with $w=1$ with respect to the tower of its oscillator modes. On the other hand, if we have an EFT string limit where the leading tower scales with $w=1$, we expect to have the QG cut-off ``sandwiched'' between the two scales \cite{Martucci:2024trp},
	 \begin{equation}\label{eq. strings non BPS}
	 \sqrt{\mathcal{T}}\gtrsim \Lambda_{\rm QG}\gtrsim m_{\star}, \quad\text{while} \quad \sqrt{\mathcal{T}}\sim m_{\star}\,.
	 \end{equation}
	 Since from the ESC we only expect the QG cut-off to coincide with the leading tower along emergent string limits \cite{Castellano:2022bvr},\footnote{For decompactification limits the higher-dimensional Planck mass is parametrically heavier than the KK scale, with 
	 \begin{equation}
	 	\frac{M_{{\rm Pl},d+n}}{M_{{\rm Pl},d}}\sim\left(\frac{M_{{\rm KK},n}}{M_{{\rm Pl},d}}\right)^{\frac{n}{d+n-2}}
	 \end{equation}
	 along pure decompactification limits \cite{Castellano:2022bvr}. Note again how in the $n\to\infty$ limit, which we can understood as an emergent string limit, we recover $\Lambda_{\rm QG}=\lim_{n\to\infty}M_{{\rm Pl},d+n}=\lim_{n\to\infty}M_{{\rm KK},n}=m_{\rm osc}$.} then $m_\star=m_{\rm osc}$, and thus we would have a fundamental string with $\mathcal{T}_{\rm F1}\sim \mathcal{T}\sim m_{\rm  osc}^2$. Now, since EFT strings are not purely gauge-theoretical objects, but also couple to the gravitational background, see e.g., \cite[Appendix F]{Lanza:2021udy}, on its excitation states we would find a symmetric spin-2 field. If the EFT string becomes light at the same rate as the fundamental string, then the Emergent String Conjecture forces us to identify both, as otherwise we would find two different weakly coupled strings with a graviton in their spectrum.\footnote{Note that it the EFT string is much heavier than the QG cutoff, $\sqrt{\mathcal{T}}\gg\Lambda_{\rm QG}$, then there is no problem with this, since its spectrum will not contribute to that of the gravitational EFT.} As a result our fundamental string must be BPS in this scenario. See \cite{Kaufmann:2024gqo,Hassfeld:2025uoy} for explicit computations of the EFT string central charge and its matching with that of fundamental heterotic or type II strings. The above can be summarized as
	 \begin{equation}
	 \boxed{\text{EFT string with $w=1$}\Longleftrightarrow \text{BPS F1-string}}
\end{equation}	  
	  
	 In other words, the above arguments imply that \textit{any EFT string with $\vec{\zeta}_{\mathcal{T}}=\vec{\zeta}_{\star}$ must be a BPS fundamental string}. Conversely, given any tower of F1-oscillator modes whose $\zeta$-vector does not correspond to that of an EFT string, one can conclude that such fundamental string is \textit{not BPS}. We will take the above as an assumption for our bottom-up analysis.

\vspace{0.25cm}

Given a K\"ahler potential, the above assumptions imply a \emph{finite} set of possible arrangements of the leading UV towers, together with their scaling along the different asymptotic limits. Equivalently, we can determine the type of asymptotic duality frames on our moduli space. Our three assumptions can be rearranged in a \textit{reconstruction algorithm} that provides the step-by step rules to yield the allowed arrangement of light towers in terms of convex hulls. Since the reconstruction algorithm takes a generic shift-symmetric K\"ahler potential as input, its results will provide a purely bottom-up picture which is, in principle, blind to its top-down realization.
\begin{tcolorbox}[enhanced jigsaw,breakable,pad at break*=1mm,colback=cyan!5!white,colframe=cyan!90!black,title={\centering\large{\textbf{Reconstruction algorithm}}}]
\vspace{0.01\linewidth}
Take a 4d $\mathcal{N}=1$ supergravity theory, with an asymptotically shift symmetric K\"ahler potential with expression  $K\sim -\log P(s)$:
\begin{enumerate}
\item[\textbf{a.}] Split the asymptotic regions of the saxionic cone $\Delta$, parameterized by $n$ saxions, into different \emph{growth sectors} \cite{Grimm:2018cpv,Corvilain:2018lgw} $\mathcal{R}_{i_1\dots i_n}$:
\begin{equation}\label{eq.growth sector}
	\Delta\supseteq\bigcup_{(i_1,\dots,i_n)}\mathcal{R}_{i_1\dots i_n}=\bigcup_{(i_1,\dots,i_n)}\left\{s^{i_1}\gtrsim s^{i_2}\gtrsim\dots \gtrsim s^{i_n}\gg 1\right\}\,,
\end{equation}
under which we can approximate our homogeneous, degree-$p$ polynomial by its leading monomial, $P(s)\sim \prod_{i=1}^n(s^i)^{p_i}$, with $\sum_{i=1}^n p_i=p$. For our purposes of computing asymptotic lengths and angles, this leading term suffices.

	\item[\textbf{b.}] Compute the zeta-vectors $\{\vec{\zeta}_{\mathcal{T}_{e^i}}\}_i$ associated to the mass scale of elementary EFT strings. As shown in Appendix \ref{app.id}, these are all located along the same hyperplane.
	\item[\textbf{c.}] Define the lattice\footnote{Being totally rigorous, $\Lambda_K$ is not a lattice, since $\mathbb{Z}_{\geq 0}$ is not a ring, but for the sake of simplicity we will call it so.}  generated by linear combinations of $\{\vec{\zeta}_{\mathcal{T}_{e^i}}\}_i$ with non-negative integer coefficients,
	\begin{align}\label{eq.lambda K}
		\Lambda_K&=\left\{\vec{x}\cdot\vec{\zeta}_{\mathcal{T}_{e^i}}=w|\vec{\zeta}_{\mathcal{T}_{e^i}}|^2\,,\;\text{with }\,w\in\mathbb{Z}_{\geq 0}\;\forall i=1,\dots,n\right\}\notag\\
		&=\left\{\sum_{i=1}^nw_i\vec{\zeta}_{\mathcal{T}_{e^i}}\,|\,w_i\in\mathbb{Z}_{\geq 0}\;\forall i=1,\dots,n\right\}\,,
	\end{align}
	see Fact \ref{fct.eq} from Appendix \ref{app.id}. We can restrict to  $0\leq w_i\leq\lfloor\sqrt{3p_i}\rfloor$, since in the growth sector $\mathcal{R}_{i_1\dots i_n}$ elementary EFT strings have $|\vec{\zeta}_{\mathcal{T}_{e^i}}|=\frac{1}{\sqrt{2p_i}}$, while $|\vec{\zeta}_\star|\leq \sqrt{\frac{3}{2}}$.
	\item[\textbf{d.}] Out of the points in $\Lambda_K$, take those whose distance from the origin is compatible with the expectation on exponential rates provided by the Emergent String Conjecture and unwarped decompactifications, \eqref{eq.length} and \eqref{eq.lengthosc}, this is
	\begin{equation}\label{eq.lattice}
	\hat{\Lambda}_K=\left\{\vec{x}\in\Lambda_K\,|\, |\vec x|=\sqrt{\frac{2+n}{2n}}\;\text{for}\;n=1,\,2\,\dots\,7\;\text{or}\:\infty\right\}\,.
	\end{equation}
	
	\item[\textbf{e.}]  Consider the possible tower polytopes, convex after adding the origin $\vec 0$, resulting from the union of different simplices generated by elements $\hat{\Lambda}_K$ in such a way that for each $k$-simplex (with $k=1,\dots,\,n-1$) the distance to the origin $\vec 0$ is $\sqrt{\frac{2+n_\star}{2n_\star}}$ with  $n_\star=1,\,\dots,\, 7$ or $\infty$. When such distance is $\frac{1}{\sqrt{2}}$  an oscillator tower is located at the pericenter (the closest point of the $k$-simplex to the origin). Otherwise, for finite $n_\star$ this corresponds to bounded states of KK modes. If the pericenter is in the interior of the saxionic cone $\Delta$, then such configuration is only consistent if $n_\star=\sum_{i=1}^kn_i$, with $\{\vec{\zeta}_{{\rm KK},n_i}\}_{i=1}^k$ the vertices spanning the $k$-simplex.
	 
\end{enumerate}
The resulting possible solutions will be the consistent tower polytopes. Note that all vertices with $|\vec{\zeta}|=\frac{1}{\sqrt{2}}$ that are do not coincide with an EFT string will correspond to a non-BPS fundamental string.
\end{tcolorbox}

Point \textbf{e.} is a technicality that must be imposed for consistency of the global fitting of towers, following the discussion after the assumptions in Section \ref{sec.reconstructionalg}.  Given a set of $k$ different $\zeta$-vectors, when moving perpendicular to the $k$-simplex they span, all the associated towers become light at the same rate, what in \cite{Etheredge:2024tok} was known as ``degenerate'' limits. Two things can happen then. On one hand, if together with the said towers we have a fundamental string becoming light, then the rest of towers are simply KK modes can be seen as excitation modes of the string, since these include the higher-dimensional graviton and other fields. By the universal quantization of the critical string in terms of the dilaton, the decay rate must be $\frac{1}{\sqrt{d-2}}=\frac{1}{\sqrt{2}}$. On the other hand, if all involved towers are KK modes, then when moving perpendicular to the simplex they span, we are decompactifying a set of $n_i$-cycles (with $i=1,\dots,7$) at the same rate. Since the $\{\vec{\zeta}_{{\rm KK},n_i}\}_{i=1}^k$ are linearly independent (as they generate a $k$-simplex), then the associated volumes are too. We thus should have the KK modes associated to the decompactification of the $(n_\star=\sum_{i=1}^kn_i)$-dimensional manifold, and the $k$-simplex must be located at a distance $\sqrt{\frac{2+n_\star}{2n_\star}}$.

One possible caveat occurs when, for a given simplex of KK $\zeta$-vectors, the pericenter to the origin is precisely one of the simplex generators, and it is located on the boundary of the saxionic cone $\Delta$. While directions outside the cone are not accessible within the perturbative regime, there still can be the case that there are towers of states whose $\zeta$-vector points in that direction. While they will never become the leading tower for our trajectories of interest, they still can form bound states with towers within the cone. This way, we can obtain simplices where the pericenter is not associated with a tower given by the bounded states of its generators (which include such pericenter), while still being consistent, as actually the simplex we observe inside the saxionic cone $\Delta$ is actually part of a bigger, consistent simplex partially laying outside the cone. This is e.g. the case for bound states of M2-branes wrapped on a curve and  KK modes of the M-theory interval, as described in \cite{Grieco:2025bjy} for $E_8\times E_8$ Heterotic CY 3-fold/type IIA orientifold compactifications.

 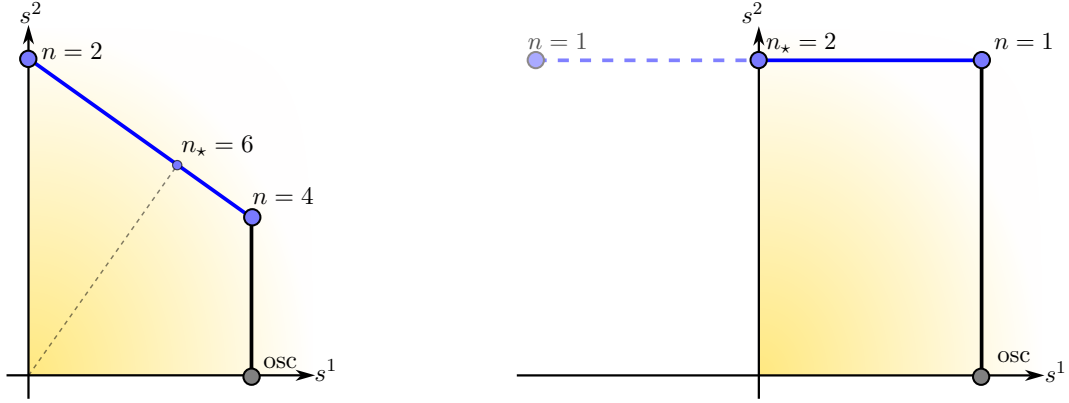
\begin{figure}[h] 	
 	\centering
			\begin{subfigure}[b]{0.41\textwidth}
			\captionsetup{width=.95\linewidth}
\centering
    \resizebox{0.8\textwidth}{!}{%
\begingroup%
  \makeatletter%
  \providecommand\color[2][]{%
    \errmessage{(Inkscape) Color is used for the text in Inkscape, but the package 'color.sty' is not loaded}%
    \renewcommand\color[2][]{}%
  }%
  \providecommand\transparent[1]{%
    \errmessage{(Inkscape) Transparency is used (non-zero) for the text in Inkscape, but the package 'transparent.sty' is not loaded}%
    \renewcommand\transparent[1]{}%
  }%
  \providecommand\rotatebox[2]{#2}%
  \newcommand*\fsize{\dimexpr\f@size pt\relax}%
  \newcommand*\lineheight[1]{\fontsize{\fsize}{#1\fsize}\selectfont}%
  \ifx\svgwidth\undefined%
    \setlength{\unitlength}{163.49774578bp}%
    \ifx\svgscale\undefined%
      \relax%
    \else%
      \setlength{\unitlength}{\unitlength * \real{\svgscale}}%
    \fi%
  \else%
    \setlength{\unitlength}{\svgwidth}%
  \fi%
  \global\let\svgwidth\undefined%
  \global\let\svgscale\undefined%
  \makeatother%
  \begin{picture}(1,1.05407393)%
    \lineheight{1}%
    \setlength\tabcolsep{0pt}%
    \put(0,0){\includegraphics[width=\unitlength,page=1]{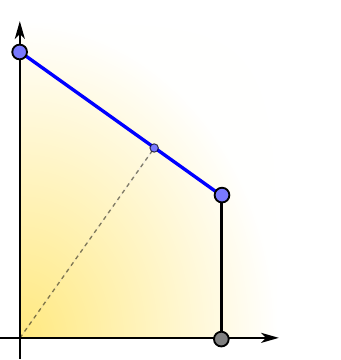}}%
    \put(0.68149575,0.08124525){\color[rgb]{0,0,0}\makebox(0,0)[lt]{\lineheight{1.25}\smash{\begin{tabular}[t]{l}osc\end{tabular}}}}%
    \put(0.09279469,0.90012092){\color[rgb]{0,0,0}\makebox(0,0)[lt]{\lineheight{1.25}\smash{\begin{tabular}[t]{l}$n=2$\end{tabular}}}}%
    \put(0.65234173,0.51782164){\color[rgb]{0,0,0}\makebox(0,0)[lt]{\lineheight{1.25}\smash{\begin{tabular}[t]{l}$n=4$\end{tabular}}}}%
    \put(0.46415379,0.65443306){\color[rgb]{0,0,0}\makebox(0,0)[lt]{\lineheight{1.25}\smash{\begin{tabular}[t]{l}$n_\star=6$\end{tabular}}}}%
    \put(0.81715228,0.04820526){\color[rgb]{0,0,0}\makebox(0,0)[lt]{\lineheight{1.25}\smash{\begin{tabular}[t]{l}$s^1$\end{tabular}}}}%
    \put(0.03465093,0.99657172){\color[rgb]{0,0,0}\makebox(0,0)[lt]{\lineheight{1.25}\smash{\begin{tabular}[t]{l}$s^2$\end{tabular}}}}%
  \end{picture}%
\endgroup%

    }
\caption{1-simplex (in blue) generated by KK modes associated to 2 and 4 compact dimensions. } \label{f.bound1}
\end{subfigure}
\begin{subfigure}[b]{0.57\textwidth}
\captionsetup{width=.95\linewidth}
\centering
    \resizebox{0.8\textwidth}{!}{%
\begingroup%
  \makeatletter%
  \providecommand\color[2][]{%
    \errmessage{(Inkscape) Color is used for the text in Inkscape, but the package 'color.sty' is not loaded}%
    \renewcommand\color[2][]{}%
  }%
  \providecommand\transparent[1]{%
    \errmessage{(Inkscape) Transparency is used (non-zero) for the text in Inkscape, but the package 'transparent.sty' is not loaded}%
    \renewcommand\transparent[1]{}%
  }%
  \providecommand\rotatebox[2]{#2}%
  \newcommand*\fsize{\dimexpr\f@size pt\relax}%
  \newcommand*\lineheight[1]{\fontsize{\fsize}{#1\fsize}\selectfont}%
  \ifx\svgwidth\undefined%
    \setlength{\unitlength}{234.68053382bp}%
    \ifx\svgscale\undefined%
      \relax%
    \else%
      \setlength{\unitlength}{\unitlength * \real{\svgscale}}%
    \fi%
  \else%
    \setlength{\unitlength}{\svgwidth}%
  \fi%
  \global\let\svgwidth\undefined%
  \global\let\svgscale\undefined%
  \makeatother%
  \begin{picture}(1,0.75336234)%
    \lineheight{1}%
    \setlength\tabcolsep{0pt}%
    \put(0,0){\includegraphics[width=\unitlength,page=1]{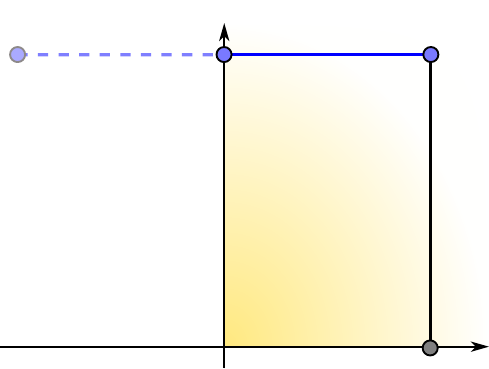}}%
    \put(0.91080435,0.06947193){\color[rgb]{0,0,0}\makebox(0,0)[lt]{\lineheight{1.25}\smash{\begin{tabular}[t]{l}osc\end{tabular}}}}%
    \put(0.90479244,0.66246955){\color[rgb]{0,0,0}\makebox(0,0)[lt]{\lineheight{1.25}\smash{\begin{tabular}[t]{l}$n=1$\end{tabular}}}}%
    \put(1.00039099,0.0313308){\color[rgb]{0,0,0}\makebox(0,0)[lt]{\lineheight{1.25}\smash{\begin{tabular}[t]{l}$s^1$\end{tabular}}}}%
    \put(0.42565557,0.71242167){\color[rgb]{0,0,0}\makebox(0,0)[lt]{\lineheight{1.25}\smash{\begin{tabular}[t]{l}$s^2$\end{tabular}}}}%
    \put(0.02118836,0.66246955){\color[rgb]{0,0,0}\makebox(0,0)[lt]{\lineheight{1.25}\smash{\begin{tabular}[t]{l}$\color{gray}{n=1}$\end{tabular}}}}%
    \put(0.47261381,0.66246955){\color[rgb]{0,0,0}\makebox(0,0)[lt]{\lineheight{1.25}\smash{\begin{tabular}[t]{l}$n_\star=2$\end{tabular}}}}%
  \end{picture}%
\endgroup%

    }
\caption{1-simplex (in blue) where the pericenter, with $n_\star=2$, corresponds to bounded states of two KK-1 towers, one of which has a $\zeta$-vector pointing outside $\Delta$.} \label{f.bound2}
\end{subfigure}
		\caption{Illustration of the bounded KK states that are located in the middle of $k$-simplices, as well as the nuances concerning point \textbf{e.} in the reconstruction algorithm. In both cases, for the black 1-simplex the pericenter corresponds to a fundamental string. In yellow we depict the perturbative regions of the saxionic cone $\Delta$.}\label{fig.bounded}
\end{figure}

The above discussion is illustrated in Figure \ref{fig.bounded}, and it will play an important role in the interplay between the Integral Scaling Relation and the taxonomy rules \eqref{eq. taxonomy} in Section \ref{ss.taxonomy}.

\subsection{Recovering string/M-theory examples\label{ss.recovering}}

Let us illustrate this algorithm and discuss its results for the set of examples analyzed in \cite{Grieco:2025bjy}. The K\"ahler potentials considered in the following should always be understood as the leading behavior of $K$ in a given asymptotic regime, with non-perturbative contributions being neglected. 

\subsubsection*{Reconstruction for $K\sim -\log\left[s^0(s^1)^3\right]$}
We first consider the simple $K\sim \log\left[s^0(s^1)^3\right]$ K\"ahler potential. As discussed in \cite{Lanza:2021udy,Grieco:2025bjy}, this can be obtained from heterotic or type II compactifications on CY 3-folds/orientifolds with homogeneous volume controlled by $(s^1)^3$, or equivalent a slice of M-theory on a Joyce manifold, see Appendix \ref{app.topdown}. We follow the reconstruction algorithm step by step:\\

\noindent \begin{minipage}[b]{0.55\textwidth}
$\bullet$ Given the K\"ahler potential $K\sim \log\left[s^0(s^1)^3\right]$, we compute the $\zeta$-vectors of the elementary EFT strings:
 \begin{equation*}
 \vec{\zeta}_{\mathcal{T}_{e^0}}=\left(\frac{1}{\sqrt{2}},0\right)\,,\quad  \vec{\zeta}_{\mathcal{T}_{e^1}}=\left(0,\frac{1}{\sqrt{2}}\right)\,,
 \end{equation*}
 written in the flat coordinates $\hat{s}^{0}=\frac{1}{\sqrt{2}}\log s^0$ and $\hat{s}^{1}=\sqrt{\frac{3}{2}}\log s^1$.
\end{minipage}
\hfill
\raisebox{-1\baselineskip}{\includegraphics[width=0.44\textwidth]{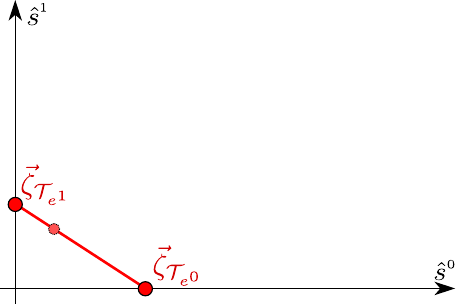}}
\begin{minipage}[b]{0.55\textwidth}
 $\bullet$ We evaluate the lattice $\Lambda_K$ \eqref{eq.lambda K}. Note that we could have restricted to the following weights:
 \begin{equation*}
 w_0\in\{0,1\}\,\quad w_1\in\{0,1,2,3\}\,,
\end{equation*}  
as otherwise the points in $\Lambda_K$ will have too large norms.
\end{minipage}
\hfill
\raisebox{-1\baselineskip}{\includegraphics[width=0.44\textwidth]{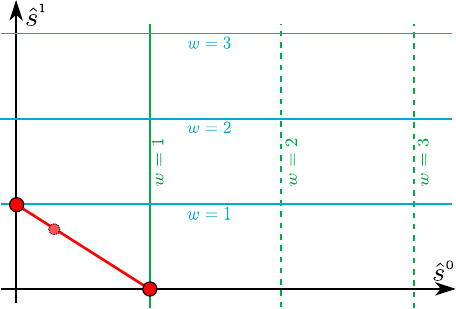}}
\begin{minipage}[b]{0.55\textwidth}
$\bullet$   We compute the refined $\hat{\Lambda}_K$ lattice \eqref{eq.lattice}, by taking only those points with appropriate norms. These are
  \begin{equation*}
  \hat{\Lambda}_K=\left\{\left(\sqrt{\tfrac{3}{2}},0\right),\left(\sqrt{\tfrac{2}{3}},0\right),\left(\sqrt{\tfrac{2}{3}},\tfrac{1}{\sqrt{2}}\right),\left(0,\tfrac{1}{\sqrt{2}}\right)\right\}.
  \end{equation*}
  \vspace{0.2cm}
\end{minipage}
\hfill
\raisebox{-1\baselineskip}{\includegraphics[width=0.44\textwidth]{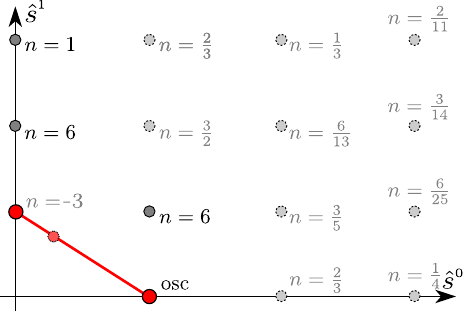}}
\begin{minipage}[b]{0.55\textwidth}
$\bullet$  We compute the admissible simplices by requiring that their pericenter to the origin has the norm of a string oscillator tower (signaling a non-BPS F1-string if such point is not an EFT string $\zeta$-vector) or a KK tower associated to the decompactification of the sum of dimensions given by the generators of the simplex. We obtain two possible arrangements, here depicted in {\color{newBLUE}{\textbf{blue}}} and {\color{newGREEN}{\textbf{green}}}.
\end{minipage}
\hfill
\raisebox{-1\baselineskip}{\includegraphics[width=0.44\textwidth]{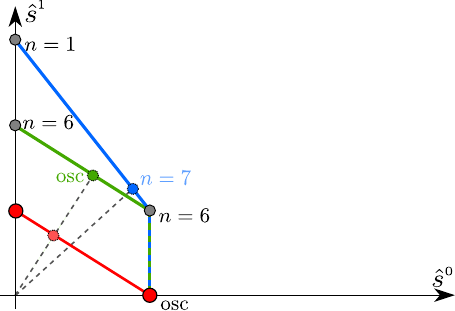}}
\vspace{0.25cm}

We have thus obtained that only two tower polytopes are consistent with our assumptions from Section \ref{sec.reconstructionalg}. This precisely matches the arrangements obtained for top-down 4d $\mathcal{N}=1$ compactifications respectively coming from {\color{newBLUE}heterotic $E_8\times E_8$ and type IIA} string theories, see Figure \ref{fig.bu-HE-1}, or {\color{newGREEN}heterotic $SO(32)$/type I string theory or F-theory}, see Figures \ref{fig.bu-HO} and \ref{fig.bu-IIB} in Appendix \ref{app.topdown}, in all cases with the volume of the compact manifold growing homogeneously. See \cite[Sections 3 and 4]{Grieco:2025bjy} for the details regarding the microscopic identification of the towers and the saxion dependence of their masses.

\subsubsection*{Reconstruction for $K\sim -\log\left[s^0(s^1)^2s^2\right]$}
 For $K\sim \log\left[s^0(s^1)^2s^2\right]$ (the unique quartic K\"ahler potential with three saxions and $P(s)$ monomial), an analogous and slightly more involved procedure can be carried out, as depicted in Figure \ref{f.bu2}. Under our assumptions, only two unique tower arrangements (see Figure \ref{f.bu2-5}) are consistent, precisely corresponding with the HE/type IIA (Figure \ref{fig.bu-HE-2}) and HO/type I/type IIB/F-theory tower polytopes (Figures \ref{fig.bu-HO-2} and \ref{fig.bu-IIB}) found in \cite{Grieco:2025bjy} for this K\"ahler potential, as explained in Appendix \ref{app.topdown}.
 
Note that again, as in the $K\sim -\log[s^0(s^1)^3]$, we obtain two possible consistent arrangements. As a matter of fact, taking the slice for which $s^1\sim s^2$, depicted in dashed lines in Figure \ref{f.bu2-5}, where $(s^1)^2s^2\sim(s^1)^3$, one recovers the two arrangements studied for $K\sim-\log\big(s^0(s^1)^3\big)$.

As a final comment, notice that, even if $K\sim \log\left[s^0(s^1)^2s^2\right]$ is invariant under $s^0\leftrightarrow s^2$, only arrangement \textbf{(A)} in Figure \ref{f.bu2-5} keeps this symmetry. This is an illustration of the fact that possible symmetries in the K\"ahler potential/EFT string tensions need not translate to the arrangement of light towers/duality frames.

    \begin{figure}
\begin{center}
\begin{subfigure}[b]{0.33\textwidth}
\captionsetup{width=.95\linewidth}
\center
\includegraphics[width=\textwidth]{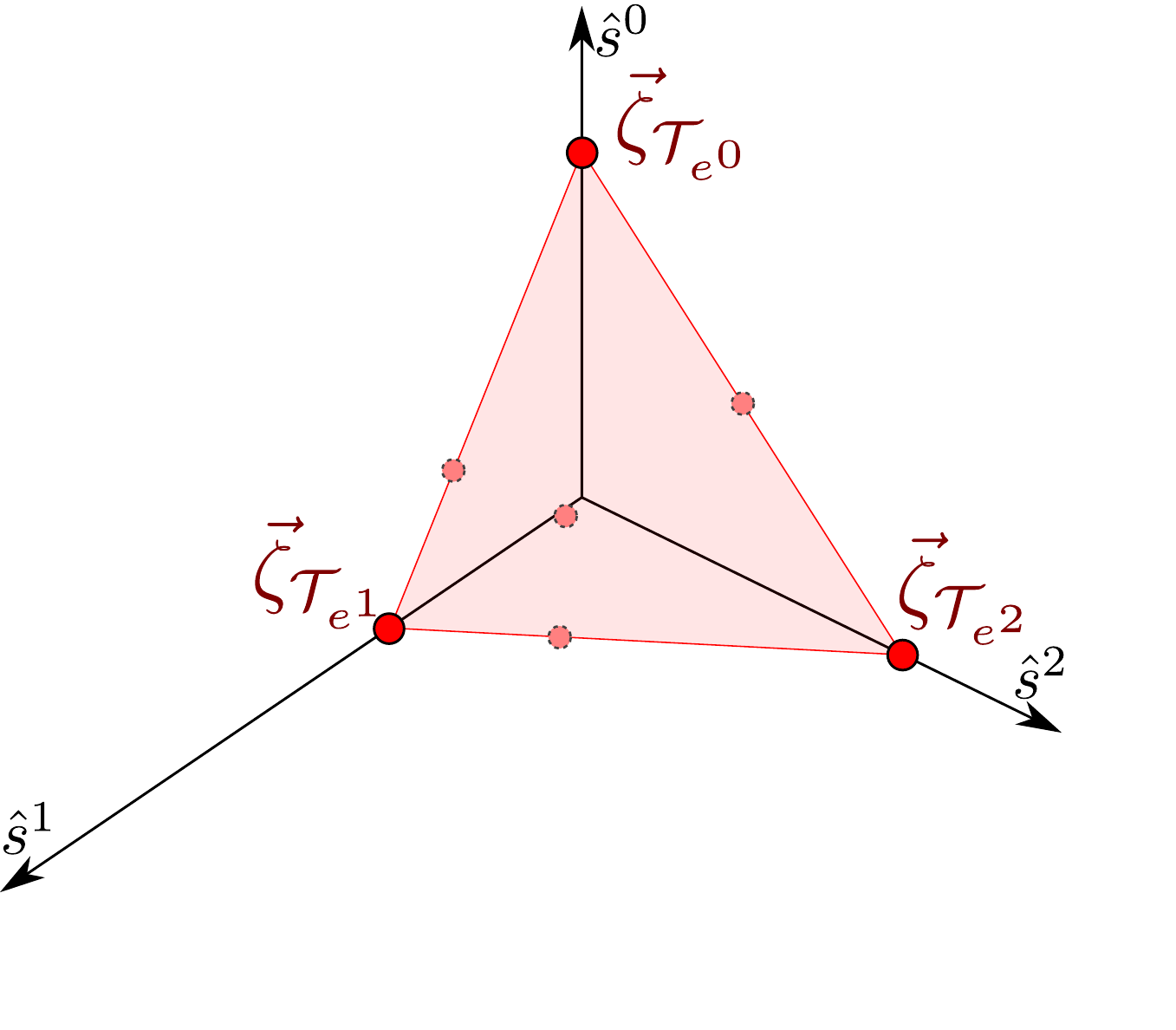}
\caption{\hspace{-0.3em} Plane generated by  $\zeta_{\mathcal{T}}$-vectors of elementary EFT strings.} \label{f.bu2-1}
\end{subfigure}\begin{subfigure}[b]{0.3\textwidth}
\center
\includegraphics[width=\textwidth]{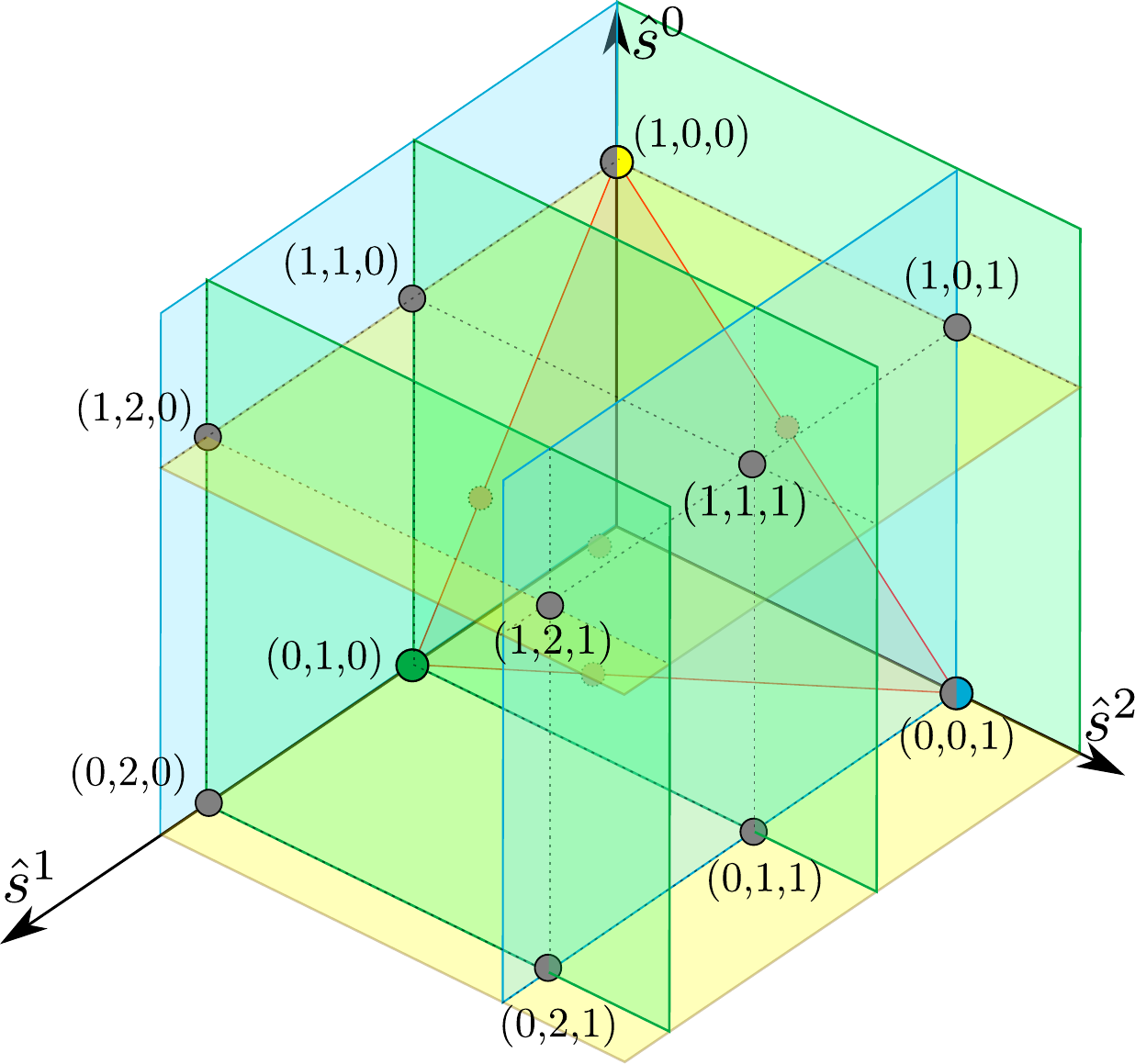}
\caption{\hspace{-0.33em}  Lattice $\Lambda_K$ \eqref{eq.lambda K}, labeled by $\vec{w}=(w_0,w_1,w_2)$, with $0\leq w_0,w_2\leq 1$ and $0\leq w_1\leq 2$.} \label{f.bu2-2}
\end{subfigure}
\begin{subfigure}[b]{0.33\textwidth}
\captionsetup{width=.95\linewidth}
\center
\includegraphics[width=\textwidth]{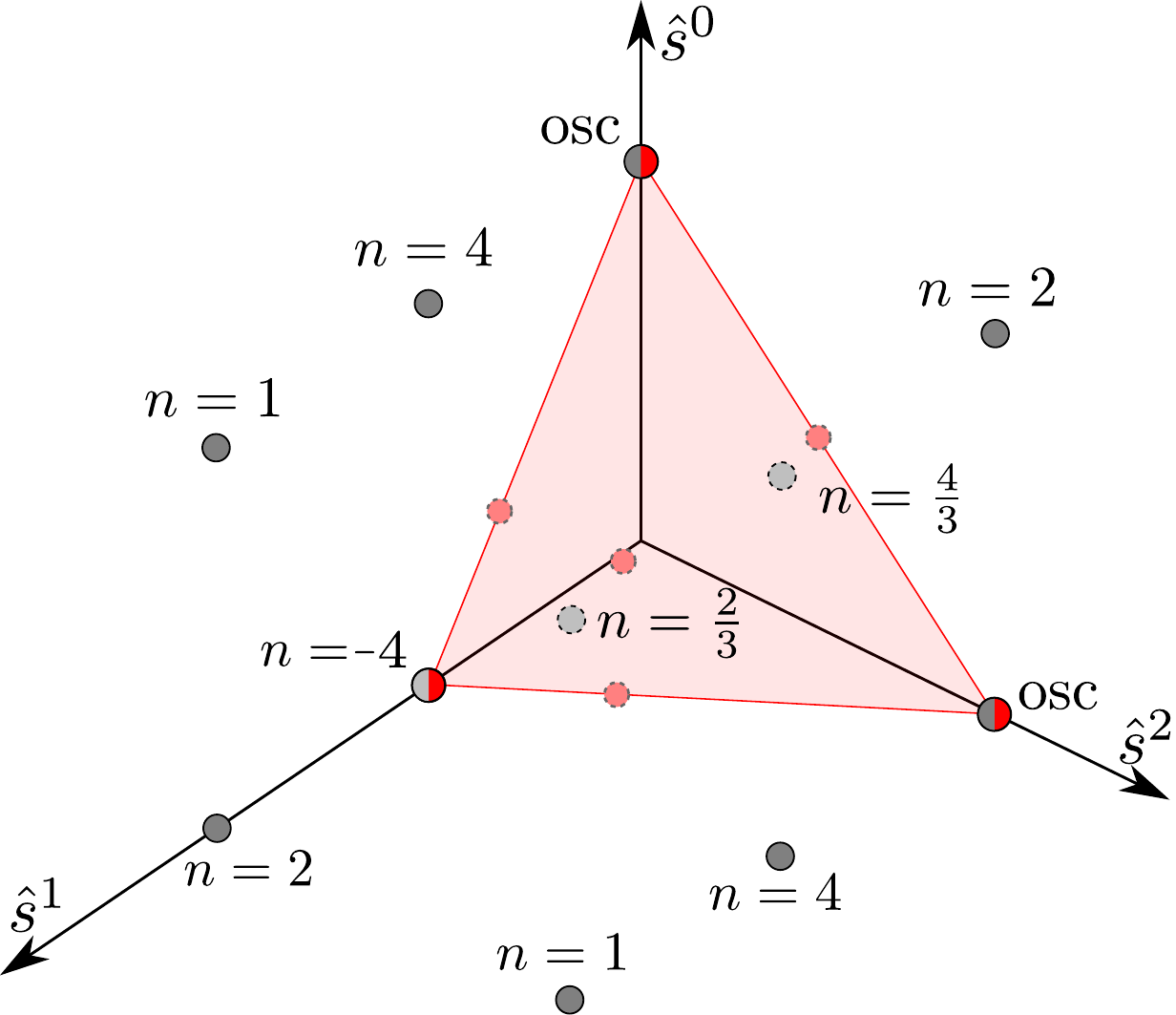}
\caption{\hspace{-0.3em} Refined $\hat{\Lambda}_K$ \eqref{eq.lattice}, labeled by number of decompactified dimensions.} \label{f.bu2-3}
\end{subfigure}
\begin{subfigure}[b]{0.33\textwidth}
\center
\includegraphics[width=\textwidth]{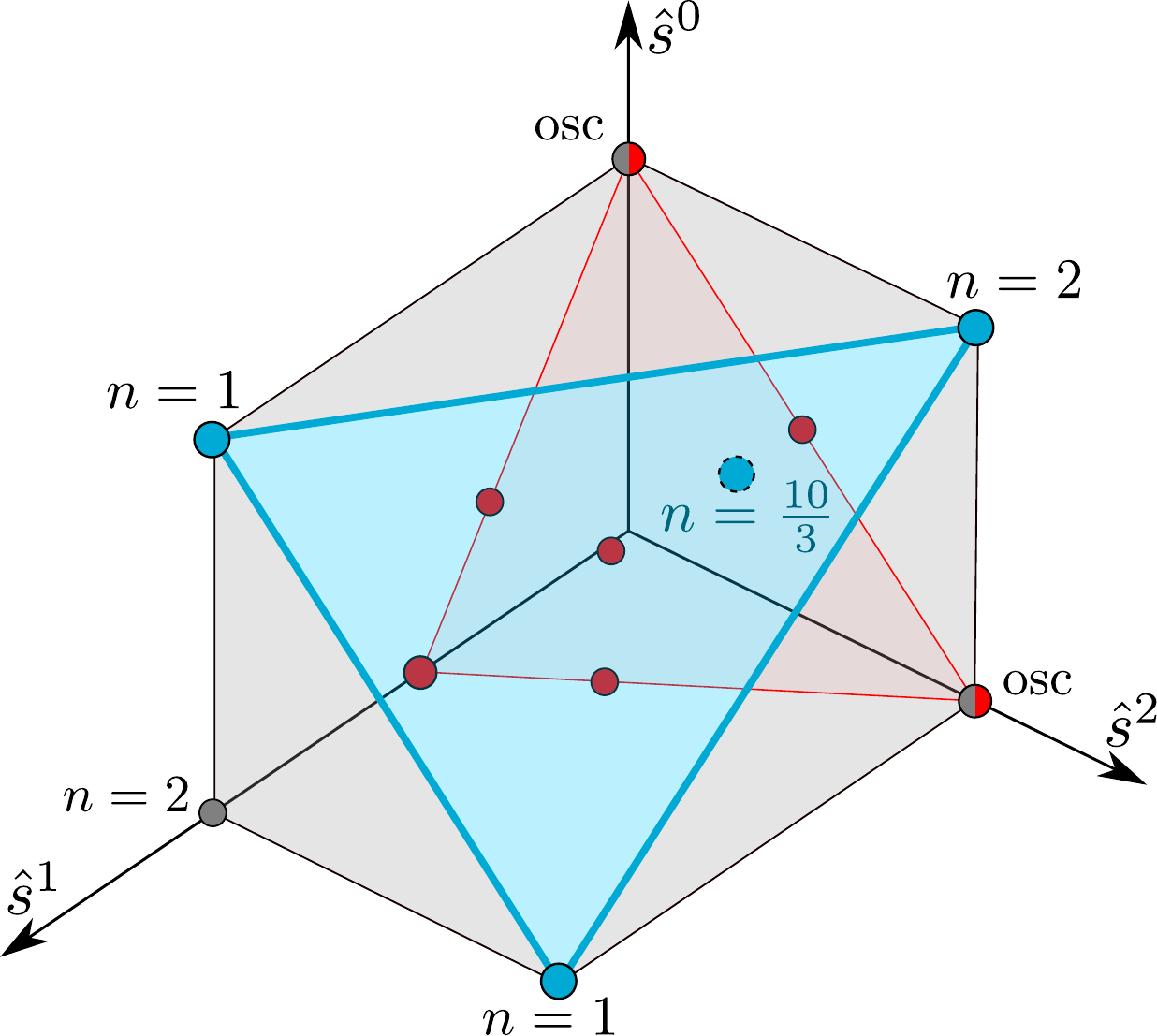}
\caption{\hspace{-0.3em} Inconsistent tower arrangement. The blue simplex is located at distance  corresponding to $n=\frac{10}{3}\neq 1+1+2$.} \label{f.bu2-4}
\end{subfigure}
\begin{subfigure}[b]{0.66\textwidth}
\center
\includegraphics[width=\textwidth]{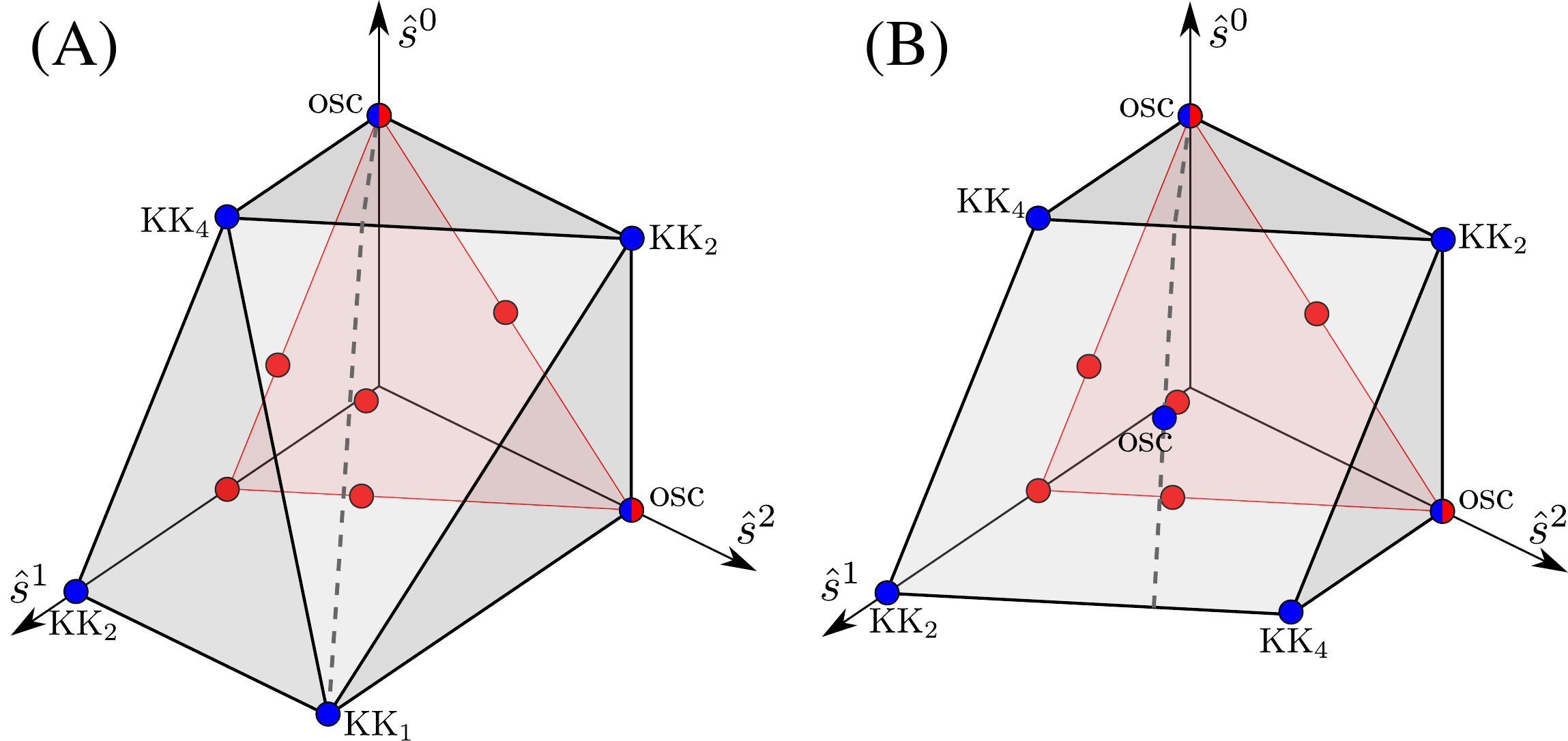}
\caption{\hspace{-0.3em} Unique tower arrangements consistent with our assumptions. Polytope \textbf{(A)} is found in HE, IIA and M-theory compactifications, while \textbf{(B)} appears in HO/type I and F-theory compactifications. In dashed lines we depict the slice $s^1\sim s^2\to\infty$.} \label{f.bu2-5}
\end{subfigure}
\caption{
Reconstruction of the light towers and duality frame arrangement for $K\sim-\log[s^0 (s^1)^3s^2]$ K\"ahler potential, resulting in a unique solutions depicted in Figure \ref{f.bu2-5}. In Figure \ref{f.bu2-4} we depict an inconsistent arrangement.
\label{f.bu2}}
\end{center}
\end{figure}

\subsubsection*{Reconstruction for \texorpdfstring{$K\sim \log\left(s^0s^1 s^2 s^3\right)$}{K~log(s0s1 s2 s3)}}
Take now $K\sim \log\left(s^0s^1 s^2 s^3\right)]$, the only single-term quartic K\"ahler potential with four saxions. Note that we have a $S_4$ symmetry interchanging the different saxions. Similar to the previous example, this does not necessarily translate to the tower and duality frame arrangement. Being in a 4-dimensional moduli space slice, we will not be able to pictorially represent the different steps of the reconstruction, but we can nonetheless apply the reconstruction algorithm as in the previous examples:
\begin{itemize}
	\item We first compute the $\zeta$-vectors for the elementary EFT strings, given by
	\begin{equation}
	\vec{\zeta}_{\mathcal{T}_{e^i}}=\Big(\underline{\tfrac{1}{\sqrt{2}},0,0,0}\Big)\,,\quad i=0,\,1,\,2,\,3\,,
	\end{equation}
where the bar represents all the permutations of the underlined elements. Unsurprisingly, the $\vec{\zeta}_{\mathcal{T}_{e^i}}$ vectors keep the $S_4$ symmetries of the saxions in the K\"ahler potential.
	\item We build the lattice $\Lambda_K$ \eqref{eq.lambda K}, with $w_i\in\{0,1\}$, and its refinement $\hat\Lambda_K$ \eqref{eq.lattice}, whose 14 points correspond to the possible $\zeta$-vectors of the light towers, as follows:
\begin{equation}
\begin{tabular}{|c||c|c|c|}\hline
	$\vec{\zeta}_\star$&$\Big(\underline{\frac{1}{\sqrt{2}},0,0,0}\Big)$&$\Big(\underline{\frac{1}{\sqrt{2}},\frac{1}{\sqrt{2}},0,0}\Big)$&$\Big(\underline{\frac{1}{\sqrt{2}},\frac{1}{\sqrt{2}},\frac{1}{\sqrt{2}},0}\Big)$\\\hline
	$n$& osc &2 &1 \\\hline
	\end{tabular}
\end{equation}	
	
	\item We now consider the possible consistent simplices. Out of the four $\vec{\zeta}_{\rm KK,1}$ possible towers, at most one is present, as if more than one appears, the direction in which all the $\vec{\zeta}_{\rm KK,1}$ towers become light at the same rate results in a non-integer number $n=\frac{4}{3}$ of dimensions decompactifying. We can then consider the two following cases, each associated with a single consistent arrangements:
	\begin{itemize}
		\item We first assume there are no $\vec{\zeta}_{\rm KK,1}$ towers, so that the tower polytope is generated by $\vec{\zeta}_{\rm osc}$ (corresponding to elementary EFT string $\vec{\zeta}_\mathcal{T}$ vectors) and $\vec{\zeta}_{\rm KK,2}$ vectors. Now, we can consider pairs of KK-2 towers becoming light at the same time, of the type $\left\{\left(\frac{1}{\sqrt{2}},\frac{1}{\sqrt{2}},0,0\right),\left(0,0,\frac{1}{\sqrt{2}},\frac{1}{\sqrt{2}}\right)\right\}$ (i.e., with dependence on totally different moduli) or $\left\{\left(\frac{1}{\sqrt{2}},\frac{1}{\sqrt{2}},0,0\right),\left(\frac{1}{\sqrt{2}},0,\frac{1}{\sqrt{2}},0\right)\right\}$. The former, $\vec{\zeta}_{\rm osc}=\Big(\frac{1}{2\sqrt{2}},\frac{1}{2\sqrt{2}},\frac{1}{2\sqrt{2}},\frac{1}{2\sqrt{2}}\Big)$ results a non-BPS F1-string (different from any of the EFT strings), while the later corresponds to a KK-4 tower. Indeed one can check that in order to respect the $\lambda_\star\geq \frac{1}{\sqrt{d-2}}=\frac{1}{\sqrt{2}}$ bound \cite{Etheredge:2022opl} on the lightest towers for any asymptotic limit in the saxionic directions, all of the six $\vec{\zeta}_{\rm KK,2}$ and four $\vec{\zeta}_{\rm osc}$ are required. One thus encounters a tower arrangement with 4 EFT string limits, and a single non-SUSY string one (invariant under $S_4$ and its subgroups, understood as T-dualities), and decompactification limits of 2, 4 and 6 dimensions (these last two types of towers are bounded states of $\vec{\zeta}_{\rm KK,2}$ modes, and thus are not generators of the convex hull):
\begin{equation}\label{eq.ARRR1}
\boxed{\Big(\underline{\frac{1}{\sqrt{2}},0,0,0}\Big)\,,\quad \Big(\underline{\frac{1}{\sqrt{2}},\frac{1}{\sqrt{2}},0,0}\Big)\,,\quad {\color{blue}\Big(\frac{1}{2\sqrt{2}},\frac{1}{2\sqrt{2}},\frac{1}{2\sqrt{2}},\frac{1}{2\sqrt{2}}\Big)}}\,,
\end{equation}		

		respectively corresponding to BPS emergent strings, KK-2 towers and finally the non-BPS F1-string (in {\color{blue}blue}). Note that for this arrangement the $S_4$ symmetry of the K\"ahler potential is respected. This is precisely the type of arrangements in HO/type I/IIB/F-theory compactifications from \cite[Sections 4.2 and 4.3]{Grieco:2025bjy}, provided enough T-dualities can be performed to complete the saxionic cone, which usually is not the case unless we have SUSY enhancement.
		 \item We now add a single KK-1 tower in our configuration, say $\vec{\zeta}_{\rm KK,1}=\left(0,\frac{1}{\sqrt{2}},\frac{1}{\sqrt{2}},\frac{1}{\sqrt{2}}\right)$, then we must consider the limits in which they decay at the same rate as $\vec{\zeta}_{\rm KK,2}$ towers. These can be of the type $\left(0,0,\frac{1}{\sqrt{2}},\frac{1}{\sqrt{2}}\right)$ (dependent in a subset of the moduli on which $\vec{\zeta}_{\rm KK,1}$ depends) or $(\frac{1}{\sqrt{2}},0,0,\frac{1}{\sqrt{2}})$ (there is one saxion on which $\vec{\zeta}_{\rm KK,1}$ does not depend). One can check that only the latter vectors result in exponential rates corresponding to an appropriate number $n=3$ of dimensions being decompactified (the former possibility results in $n=\frac{9}{5}$). In order to have exponential rates of $\lambda_\star\geq \frac{1}{\sqrt{d-2}}=\frac{1}{\sqrt{2}}$ for the lightest tower in any direction, one must take the three $\vec{\zeta}_{\rm KK,2}$ towers with a single overlapping saxionic direction. Modulo the choice of saxionic direction on which $\vec{\zeta}_{\rm KK-1}$ does not depend, this results in the following tower arrangement
\begin{equation}\label{eq.ARRR2}
\boxed{\Big(\underline{\frac{1}{\sqrt{2}},0,0,0}\Big)\,,\quad \Big(\frac{1}{\sqrt{2}},\underline{\frac{1}{\sqrt{2}},0,0}\Big)\,,\quad {\Big(0,\frac{1}{\sqrt{2}},\frac{1}{\sqrt{2}},\frac{1}{\sqrt{2}}\Big)}}\,,
\end{equation}	
this is, four emergent string towers, three KK-2 and a single KK-1. This indeed exactly corresponds with the top-down construction from heterotic $E_8\times E_8$ or type IIA on $\mathcal{V}_X\sim s^1 s^2 s^3$, see \cite[Sections 3 and 4.1]{Grieco:2025bjy}. Note that again this arrangement singles out a universal saxion $s^0$, different from the other $\{s^i\}_{i=1}^3$, corresponding to the direction given by a weak coupling tower of a string living in 10d, with the original $S_4$ symmetry broken to $S_3$.
	\end{itemize}
\end{itemize}

\vspace{0.5cm}

As in the $K\sim -\log[s^0(s^1)^3]$ and $K\sim -\log[s^0(s^1)^2s^2]$ cases, we obtain two consistent arrangements, \eqref{eq.ARRR1} and \eqref{eq.ARRR2}, corresponding to HE/type IIA and HO/type I/type IIB/F-theory realizations (provided one can complete the saxionic cone through enough dualities). In fact it is straightforward to see that the previously studied K\"hler potentials are simply slices of the one we have just considered slices of the ones obtained here, respectively with $s^1\sim s^2\sim s^3$ and $s^1\sim s^2$.

\subsubsection*{Reconstruction for $K\sim \log\left(s^1 s^2 s^3 s^4 s^5 s^6 s^7\right)$}
We finally consider the K\"ahler potential $K\sim \log\left(s^1 s^2 s^3 s^4 s^5 s^6 s^7\right)$, which arises in M-theory compactifications on a Joyce manifold $X_7=\mathbb{T}^7/(\mathbb{Z}_2\oplus\mathbb{Z}_2\oplus\mathbb{Z}_2)$, for untwisted saxion sector (i.e., not considering the contributions from the resolution of the singularities after orbifolding). In \cite{Grieco:2025bjy}, the top-down tower polytope for this theory was obtained. We find here that this is also the \emph{unique} polytope that can be obtained through the reconstruction algorithm for this choice of K\"ahler potential. The procedure carries out analogously to that for $K \sim -\log \left( s^0 s^1 s^2 s^3 \right)$:

\begin{itemize}
\item The EFT string zeta-vectors are given by
\begin{equation}
\vec{\zeta}_{\mathcal{T}_{e^i} =}\Big(\underline{\tfrac{1}{\sqrt{2}},0,0,0,0,0,0}\Big)\,,\quad i=0,\,1,\,2,\dots,7\,,
\end{equation}
 symmetric under permutations in $S_7$
\item We compute the lattice $\Lambda_K$, with $w_i\in\{0,1\}$, which results in $2^7-1=127$ points. After refining to $\hat{\Lambda}_K$ \eqref{eq.lattice} we find the following surviving 63 possible $\zeta$-vectors:
\begin{equation}
\begin{tabular}{|c||c|c|c|}\hline
	$\vec{\zeta}_\star$&$\Big(\underline{\frac{1}{\sqrt{2}},0,0,0,0,0,0}\Big)$&$\Big(\underline{\frac{1}{\sqrt{2}},\frac{1}{\sqrt{2}},0,0,0,0,0}\Big)$&$\Big(\underline{\frac{1}{\sqrt{2}},\frac{1}{\sqrt{2}},\frac{1}{\sqrt{2}},0,0,0,0}\Big)$\\\hline
	$n$& osc & 2 &1 \\\hline
	\end{tabular}
\end{equation}

\item The argument for the need of all the $\binom{7}{2}=21$ KK-2 towers is the same as in the $K\sim \log(s^0s^1s^2s^3)$ case, as otherwise there would be limits where the light towers become light faster than $\lambda_{\rm osc}=\frac{1}{\sqrt{2}}$. However, this is not sufficient to generate a consistent tower polytope on their own. In order to see this, consider that, if the 21 KK-2 towers were the only vectors other than the elementary EFT-strings (which lie on the axes), we would have a 6-simplex generated by $\vec{\zeta}_{\rm KK,2}$ vertices, e.g.,
\begin{equation}
\Delta=\Big\langle\left(0,0,0,0,0,\tfrac{1}{\sqrt{2}},\tfrac{1}{\sqrt{2}}\right),\left(0,0,0,0,\tfrac{1}{\sqrt{2}},\tfrac{1}{\sqrt{2}},0\right),\dots\,\left(\tfrac{1}{\sqrt{2}},\tfrac{1}{\sqrt{2}},0,0,0,0,0\right)\Big\rangle
\end{equation}
This simplex has a distance to the origin associated to $n_\star=-\frac{36}{7}$, making this configuration inconsistent, thus requiring the introduction of KK-1 towers.

\item Not every arrangement including KK-1 towers is consistent. For example, if we include only the vector $\vec{\zeta}_1=\left(\frac{1}{\sqrt{2}},\frac{1}{\sqrt{2}},\frac{1}{\sqrt{2}},0,0,0,0\right)$, there would still be a 4-simplex, given by the $\binom{4}{2}=6$ KK-2 $\big(0,0,0,\underline{\frac{1}{\sqrt{2}},\frac{1}{\sqrt{2}},0,0}\big)$ vectors orthogonal to said $\vec{\zeta}_{\rm KK,1}$, which is too close to the origin, with $n_\star=-\frac{8}{3}$. At least another KK-1 vector would be needed, which in turn cannot share more than one component with the first $\vec{\zeta}_{\rm KK,1}$, as otherwise the edge connecting them would have an inconsistent bound state of $n_\star =\frac{4}{3}$. This line of thought results in the \emph{only} consistent option is to add exactly 7 KK-1 vectors in such a way that each pair shares exactly one saxionic component. 
\end{itemize}

Summing up, applying the reconstruction algorithm to  $K\sim \log\left(s^1 s^2 s^3 s^4 s^5 s^6 s^7\right)$ yields only one consistent tower polytope, precisely  that obtained from a top-down compactification of M-theory a Joyce manifold $X=\mathbb{T}^7/(\mathbb{Z}_2\oplus\mathbb{Z}_2\oplus\mathbb{Z}_2)$, see \cite[Section 5]{Grieco:2025bjy} and Appendix \ref{app.Joyce}.

\section{Constraining K\"ahler potentials and evidence for string universality\label{sec.constrain kahler}}
In this Section, we turn our attention to more general K\"ahler potentials, not necessarily appearing from known string constructions. Our goal is to constrain which K\"ahler potentials are consistent with the assumptions outline at the beginning of Section \ref{sec.reconstructionalg} . We will see that the Integral Scaling Relation together with the ESC are indeed powerful enough to rule out certain asymptotic behaviors of the K\"ahler potential.
\subsection{Constraints on asymptotic K\"ahler potentials}

The reconstruction algorithm that we laid out in Section \ref{sec.reconstr} is based on the Integral Scaling Relation \eqref{e.bottom up integer}, which is an observed feature of EFT strings, and on the Emergent String Conjecture as assumptions\footnote{Recall that we further require the technical assumption that decompactifications are to a Minkowski vacuum, rather than \emph{warped}.}. From the examples studied in the previous sections, we have seen that, starting from the K\"ahler potentials obtained in 4d $\mathcal{N}=1$ string/M-theory compactifications, this algorithm seems to always recover \emph{exactly} the tower arrangements that we would get from top-down stringy constructions. We will now perform a more systematic study of possible arrangements for arbitrary K\"ahler potentials $K\sim -\log P(s)$ with $P(s)$ a polynomial in the saxions with $\deg P(s)\leq 7$, where this upper bound is suggested by top-down constructions.
In general, we will find that only some particular choices for K\"ahler potentials will produce a consistent convex hull by the algorithm: this effectively selects some K\"ahler potentials while ruling out others.

Applying the reconstruction algorithm to any K\"ahler potential of the form  $K\sim - \log P(s)$ is straightforward if $P(s)$ is a monomial. This will be done in a systematic way in Section \ref{sec.mono}. However, for a more general polynomial, the $\zeta$-vectors of light towers will have a direct dependence on saxions and vary along different trajectories in moduli space. The whole convex hull of light towers will then be understood as a local object which has to be computed for each point and will in general be different depending on the asymptotic trajectory. In fact, this is usually the situation we face in top-down constructions of 4d $\mathcal{N}=1$ theories. 
In Section \ref{sec.gluing} we will argue how we can understand these more general cases from gluing the building blocks given by different monomials in $P(s)$, which will correspond to the leading piece dominating in the different \textit{growth sectors} in which we divide the moduli space \eqref{eq.growth sector} \cite{Grimm:2018cpv,Corvilain:2018lgw}. The resulting classification in \ref{sec.mono}, paired with the gluing in \ref{sec.gluing}, provides a complete chart of the regions of 4d $\mathcal{N}=1$ moduli spaces described at perturbative level by shift-symmetric K\"ahler potentials $K\sim - \log P(s)$  with $P(s)$ homogeneous. In particular, it also provides all the candidate light towers in each growth sector, without relying on a specific top-down construction.

\subsubsection{Monomials}\label{sec.mono}

Let us start then by constraining the possible $K\sim -\log P(s)$, with $P(s)$ being a generic monomial of the saxions, that can yield a UV spectrum of towers consistent with the assumptions outlined in the reconstruction algorithm of Section \ref{sec.reconstructionalg}. The results for these monomials will serve as building blocks (valid to leading order in each growth sector of the moduli space) that we will later \emph{glue} to discuss the most general case of a polynomial function $P(s)$ in Section \ref{sec.gluing}.

In order to do this, consider a K\"ahler potential which with the following form at leading order
\begin{equation}
K \sim -\log\left[(s^1)^{p_1} (s^2)^{p_2}\dots (s^k)^{p_k}\right]\,
\end{equation}
with total degree $p= \sum_{k=1}^n p_k$. For practical purposes, in order to perform an analysis of compatible tower arrangements we need to put bounds on the number of saxions $k$ and the degree $p$ of the polynomial. From  known string and M-theory top-down constructions, the degree of the polynomial is upper bounded by the total number of internal dimensions. It is then natural, with string and M-theory compactifications in mind, to impose $p\leq 7$ and, as a consequence, also $k\leq 7$. We also ask for $p_k$ to be integers.  In Appendix \ref{app.id}, we show that a upper bound on $w_i$ can impose a weaker lower bound (not saturated in top-down examples) on the single $p_i$ too, given by \eqref{eq.bounds}. In synthesis, at this level the motivation for $p\leq 7$ remains mostly a top-down input, with additional bottom-up arguments for it presented in \cite{Montero:2015ofa, Maldacena:2026jqd, DiUbaldo:2026rly, Etheredge:2026rio} and based on the Axionic Weak Gravity Conjecture.
We can start with the straightforward exercise of applying the reconstruction algorithm to single-saxion moduli spaces,  displayed in Table \ref{tab:1sax}.

It is interesting to notice that $K\sim \log (s^1)^5$ does not result in any $\zeta$-vector arrangement (which is a point in the one-dimensional saxionic cone, $\mathbb{R}_{>0}$) compatible with the reconstruction algorithm. In one dimension, this just means that for $p=5$ that there is no KK or string tower with mass given by a power of the string tension scale, as required by the Integral Scaling Relation \eqref{scalingweight}. In other words, since   $|\vec{\zeta}_{\mathcal{T}}|=\frac{1}{\sqrt{2p}}=\frac{1}{\sqrt{10}}$, the equation $\frac{w}{\sqrt{10}}=  \sqrt{\frac{2+n}{2n}}$ (see \eqref{eq.int-1mod} later) has no positive integer solution in $(w,n)$ (or $n=\infty$ for strings) if $w$ is integer with $w\leq 3$.\footnote{Recall that, as discussed bellow \eqref{eq.lambda K}, it is enough to take $w\leq \lfloor\sqrt{3p}\rfloor$, which for $p\geq 6$ would in principle allow for $w\geq 4$.} 

\begin{table}[H]\
    \centering
    \begin{tabular}{|c|c|}\hline
    \rowcolor{gray!10!}$K\sim-\log P(s)$     & Generating towers \\\hline
    \rowcolor{cyan!10!}  $s^1$&$\big(\tfrac{1}{\sqrt{2}}\big)$  \\
       $(s^1)^2$& $\big(1\big)$  \\
     \rowcolor{cyan!10!}  $(s^1)^3$& $\textcolor{Green}{a)} \ \ \big(\sqrt{\tfrac{2}{3}}\big)$ \quad;\quad\textcolor{Green}{b)} $\big(\sqrt{\tfrac{3}{2}}\big)$ \\
       $(s^1)^4$  & $\color{blue}\big(\tfrac{1}{\sqrt{2}}\big)$  \\
     \rowcolor{cyan!10!}  $\color{red}(s^1)^5$   & $\color{red}-$\\
       $(s^1)^6$& $\big(\tfrac{\sqrt{3}}{2}\big)$ \\
     \rowcolor{cyan!10!}  $(s^1)^7$& $\big(\tfrac{3}{\sqrt{14}}\big)$ \\\hline
    \end{tabular}
    \caption{Bottom-up reconstructions for a single modulus. Non-supersymmetric strings are colored in blue.}
    \label{tab:1sax}
\end{table}

However, one could wonder if loosening this bound on $w$, which in turn means also loosening the bound on the degree $p$ of the monomial, could allow us to find a light tower compatible with the Integral Scaling Relation. The answer to this question is negative, and by increasing the degree of the monomial it is possible to see that this is not only the case for $(s^1)^5$ but happens for most integers. The reason for this comes purely from noticing the ESC allows only eight different exponential rates for string and KK towers (again, not considering warped (de)compactifications), see \eqref{eq.length}. In particular, from the Integral Scaling Relation \eqref{eq.intscaling} we have, in the case of a single axion in a K\"ahler potential $K \sim -p \log s$,
 
\begin{equation} \label{eq.int-1mod}
 \frac{w^2}{2p} = \frac{n+2}{2n}, \ n=1,\dots,7 \ \text{and} \ \infty
\end{equation} 
which can be rewritten as the following Diophantine equations 
\begin{align}
w^2 - p = 0 \ &\text{for} \ n = \infty \label{eq.linear1} \\
n w^2-(n+2)p= 0 \  &\text{for} \ n=1,\dots,7 \label{eq.linear2}
\end{align}
where the unknowns are the scaling weight $w$ and the degree of the monomial $p$.
It is straightforward to see that equation \eqref{eq.linear1} just selects for values of $w$ which are perfect squares of the degree of the polynomial. Regarding \eqref{eq.linear2}, through the fundamental theorem of arithmetic
we can substitute
\begin{equation}\label{eq.ink1}
p=\left\{\begin{array}{ll}
	rnk^2&\text{for }n\text{ odd, with }n+2=s^2r\\
	\frac{1}{2}rnk^2&\text{for }n\text{ even, with }n+2=2s^2r
\end{array}\right.\,,\quad\text{with }k,\,s,\,r\in\mathbb{Z}_{\geq 0}\text{ and }r\text{ square-free,}
\end{equation}
and
\begin{equation}\label{eq.ink2}
	w=\left\{\begin{array}{ll}
	\sqrt{(n+2)r}k&\text{for }n\text{ odd}\\
	\sqrt{\frac{n+2}{2}r}k&\text{for }n\text{ even}
\end{array}\right.\,,
\end{equation}
with the same $k$, $s$ and $r$ as above. We stress that $r$ and $s$ are uniquely set by the factorization of $n+2$, while $k$ simply labels the different solutions for $(p,w)$ in Table \ref{tab:pw}. The smallest value for $(p,w)$ is then set by $k=1$, with larger $k$ resulting in a finer lattice in $(p,w)$.

In Table \ref{tab:pw}, the solutions to \eqref{eq.ink1} and \eqref{eq.ink2} are displayed, for all the different $n$. Note that the factorization of $n+2=(2)s^2r$ is of great importance. For example, for $n=5$, we have $s=1$ and $r=7$, which implies the large values $(p,w)=(35 k^2,7k)$, while for $n=7$, $s=3$ and $r=1$, so that $(p,w)=(7k^2,3k)$.

\begin{table}[H]
    \centering
    \begin{tabular}{|c|c|}\hline
    \rowcolor{gray!10!}$n$     & $(p,w)$ \\\hline
    \rowcolor{Green!10!}  $\infty$& $(1,1)$, $(4,2)$, $(9,3)$, \dots  \\
       $1$& $(3,3)$, $(12,6)$, $(27,9)$, \dots \\
     \rowcolor{Green!10!}  $2$& $(2,2)$, $(8,4)$, $(18,6)$, \dots \\
       $3$  & $(15,5)$, $(60,10)$, $(135,15)$, \dots  \\
     \rowcolor{Green!10!}  $4$   & $(6,3)$, $(24,6)$, $(54,9)$, \dots \\
       $5$& $(35,7)$, $(140,14)$, $(315,21)$, \dots \\
     \rowcolor{Green!10!}  $6$& $(3,2)$, $(12,4)$, $(27,6)$, \dots\\
       $7$&  $(7,3)$,  $(28,6)$,  $(63,9)$, \dots \\\hline
    \end{tabular}
    \caption{Pairs $(p,w)$ consistent with the Integral Scaling Relation for each possible leading tower consistent with the ESC, in single-moduli spaces. Displayed are solutions of \eqref{eq.ink1} and  \eqref{eq.ink2} for $k=1,2,3$. Notice how the solutions for $k>1$ are just given by a finer lattice choice of the $k=1$ case. Notice also how $n=3$ and $n=5$ for the leading KK-$n$ towers would require $p>7$. Notice also how $n=1$ and $n=6$ give two possible solutions with the same degree of the polynomial.}
    \label{tab:pw}
\end{table}

A similar argument applies in the case of multi-moduli space, with some important differences. First, the Integral Scaling Relation which in single-moduli cases is enforced by \eqref{eq.int-1mod} will now give a more involved Diophantine equations with additional variables, significantly more difficult to solve in general. Second, the reconstruction algorithm in multi-moduli spaces requires the ESC to be applied recursively for partial decompactifications, which in practice is equivalent to imposing the taxonomy rules from \cite{Etheredge:2024tok} to the obtained points of the specific lattice generated by the elementary EFT strings. In fact, both constraints are nothing but Diophantine equations which in general have to be solved computationally. \\

The arrangements of towers given by the reconstruction algorithm for 2 to 7 saxions, for monomial $P(s)$ of degree equal or less than 7, are displayed in Tables \ref{tab:bu2}, \ref{tab:bu3}, \ref{tab:bu4}, \ref{tab:bu5}, \ref{tab:bu6} and \ref{tab:bu7}. In particular, we give the $\zeta$-vectors generating the tower polytope, corresponding to either co-leading string towers or leading KK towers in each decompactification limit. We also display, in {\color{blue}blue},  the $\zeta$-vectors of non-BPS strings. These will always lie, in case they appear, in the pericenter to the origin of a $k$-simplex in place of a bound state of KK towers, and will always be heavier than the leading EFT string in their direction, see discussion around \eqref{eq. strings non BPS}.
\begin{table}[H]
    \centering
    \resizebox{0.8\textwidth}{!}{
    \begin{tabular}{|c|c|}\hline
    \rowcolor{gray!10!}    $K\sim-\log P(s)$     & Generating towers \\\hline
     \rowcolor{cyan!10!}   $s^1s^2$ & $\big(\tfrac{1}{\sqrt{2}},0\big)$, $\big(0,\tfrac{1}{\sqrt{2}}\big)$, $\big(\tfrac{1}{\sqrt{2}},\tfrac{1}{\sqrt{2}}\big)$ \\
        $(s^1)^2s^2$ &\begin{tabular}{@{}c@{}} \textcolor{Green}{a)} \ $\big(1,0\big)$, $\big(0,\tfrac{1}{\sqrt{2}}\big)$, $\big(1,\tfrac{1}{\sqrt{2}}\big)$ ;\\ \textcolor{Green}{b)} \ $\big(1,0\big)$, $\big(0,\tfrac{1}{\sqrt{2}}\big)$, $\big(\tfrac{1}{2},\tfrac{1}{\sqrt{2}}\big)$ \end{tabular} \\
    \rowcolor{cyan!10!}    $(s^1)^2(s^2)^2$ &  \begin{tabular}{@{}c@{}} $\big(1,0\big)$, $\big(0,1\big)$, $\big(\tfrac{1}{2},\tfrac{1}{2}\big)$ \end{tabular}\\
        $(s^1)^3s^2$ &\begin{tabular}{@{}c@{}} \textcolor{Green}{a)} \ $\big(\sqrt{\tfrac{3}{2}},0\big)$, $\big(0,\tfrac{1}{\sqrt{2}}\big)$, $\big(\tfrac{1}{\sqrt{6}},\tfrac{1}{\sqrt{2}}\big)$ ;\\ \textcolor{Green}{b)} \ $\big(\sqrt{\tfrac{2}{3}},0\big)$, $\big(0,\tfrac{1}{\sqrt{2}}\big)$, $\big(\tfrac{1}{\sqrt{6}},\tfrac{1}{\sqrt{2}}\big)$, $\color{blue}\big(\tfrac{1}{2}\sqrt{\tfrac{2}{3}},\tfrac{1}{2\sqrt{2}}\big)$ \end{tabular} \\
    \rowcolor{cyan!10!}     $\color{red}(s^1)^3(s^2)^2$& $\color{red}-$\\
      $(s^1)^3(s^2)^3$  &$\big(\sqrt{\tfrac{3}{2}},0\big)$, $\big(0,\sqrt{\tfrac{2}{3}}\big)$, $\big(\sqrt{\tfrac{1}{6}},\sqrt{\tfrac{2}{3}}\big)$  \\
    \rowcolor{cyan!10!}   $(s^1)^4s^2$  & $\color{blue}\big(\tfrac{1}{\sqrt{2}},0\big)$, $\big(0,\tfrac{1}{\sqrt{2}}\big)$, $\big(\tfrac{1}{\sqrt{2}},\tfrac{1}{\sqrt{2}}\big)$ \\
       $(s^1)^4(s^2)^2$ &\begin{tabular}{@{}c@{}} \textcolor{Green}{a)} \ $\color{blue}{\big(\tfrac{1}{\sqrt{2}},0\big)}$, $\big(0,1\big)$, $\big(\tfrac{1}{\sqrt{2}},1\big)$ ;\\ \textcolor{Green}{b)} \ $\color{blue}{\big(\tfrac{1}{\sqrt{2}},0\big)}$, $\big(0,1\big)$, $\big(\tfrac{1}{\sqrt{2}},\frac{1}{2}\big)$ \end{tabular} \\
    \rowcolor{cyan!10!}   $(s^1)^4(s^2)^3$  & \begin{tabular}{@{}c@{}} \textcolor{Green}{a)} \ $\color{blue}\big(\tfrac{1}{\sqrt{2}},0\big)$, $\big(0,\sqrt{\tfrac{3}{2}}\big)$, $\big(\tfrac{1}{\sqrt{2}},\tfrac{1}{\sqrt{6}}\big)$;\\ \textcolor{Green}{b)} \ $\color{blue}\big(\tfrac{1}{\sqrt{2}},0\big)$, $\big(0,\sqrt{\tfrac{2}{3}}\big)$, $\big(\tfrac{1}{\sqrt{2}},\tfrac{1}{\sqrt{6}}\big)$, $\color{blue}\big(\tfrac{1}{2\sqrt{2}}, \tfrac{1}{2}\sqrt{\tfrac{2}{3}}\big)$ \end{tabular}\\
       $\color{red}(s^1)^5s^2$  &$\color{red}-$ \\
     \rowcolor{cyan!10!}  $\color{red}(s^1)^5(s^2)^2$  &$\color{red}-$ \\
       $(s^1)^6s^2$  & $\big(\sqrt{\tfrac{3}{4}},0\big)$, $\big(0,\sqrt{\tfrac{1}{2}}\big)$, $\big(\tfrac{1}{\sqrt{3}},\tfrac{1}{\sqrt{2}}\big)$ \\\hline
    \end{tabular}}
    \caption{Bottom-up reconstructions for monomials with 2 moduli. Non-BPS F1-strings are colored in blue.}
    \label{tab:bu2}
\end{table}

\begin{table}[H]
    \centering
    \resizebox{1\textwidth}{!}{
    \begin{tabular}{|c|c|}\hline
    \rowcolor{gray!10!}    $K\sim-\log P(s)$     & Generating towers \\\hline
      \rowcolor{cyan!10!}    $s^1 s^2s^3$  & \begin{tabular}{@{}c@{}} \textcolor{Green}{a)} \ $\big(\underline{\tfrac{1}{ \sqrt{2}},0,0}\big)$, $\big(\underline{\tfrac{1}{ \sqrt{2}},\tfrac{1}{ \sqrt{2}},0}\big)$ ; \\ \textcolor{Green}{b)} \  $\big(\underline{\tfrac{1}{ \sqrt{2}},0,0}\big)$, $\big(\underline{\tfrac{1}{ \sqrt{2}},\tfrac{1}{ \sqrt{2}},0}\big)$, $\big(\tfrac{1}{ \sqrt{2}},\tfrac{1}{ \sqrt{2}},\tfrac{1}{ \sqrt{2}}\big)$ \end{tabular}   \\
          $(s^1)^2 s^2s^3$  & \begin{tabular}{@{}c@{}} \textcolor{Green}{a)} \ $\big(1,0,0\big)$, $\big(0,\underline{\tfrac{1}{\sqrt{2}},0}\big)$, $\big(1,\tfrac{1}{\sqrt{2}},0\big)$, $\big(\frac{1}{2},0,\tfrac{1}{\sqrt{2}}\big)$, $\big(0,\tfrac{1}{\sqrt{2}},\tfrac{1}{\sqrt{2}}\big)$ ;\\ \textcolor{Green}{b)} \ $\big(1,0,0\big)$, $\big(0,\underline{\tfrac{1}{\sqrt{2}},0}\big)$, $\big(\tfrac{1}{2},\underline{\tfrac{1}{\sqrt{2}},0}\big)$, $\big(0,\tfrac{1}{\sqrt{2}},\tfrac{1}{\sqrt{2}}\big)$, $\color{blue}\big(\tfrac{1}{2},\tfrac{1}{2\sqrt{2}},\tfrac{1}{2\sqrt{2}}\big)$  \end{tabular}  \\
    \rowcolor{cyan!10!}      $(s^1)^2 (s^2)^2s^3$  &  \begin{tabular}{@{}c@{}} \textcolor{Green}{a)} \ $(\underline{1,0},0)$, $\big(0,0,\tfrac{1}{\sqrt{2}}\big)$, $\color{blue}\big(\tfrac{1}{2},\tfrac{1}{2},0\big)$, $\big(\underline{1,0},\tfrac{1}{\sqrt{2}}\big)$ ;\\ \textcolor{Green}{b)} \ $(\underline{1,0},0)$, $\big(0,0,\tfrac{1}{\sqrt{2}}\big)$, $\color{blue}\big(\tfrac{1}{2},\tfrac{1}{2},0\big)$, $\big(\underline{\tfrac{1}{2},0},\tfrac{1}{\sqrt{2}}\big)$, $\big(\tfrac{1}{2},\tfrac{1}{2},\tfrac{1}{\sqrt{2}}\big)$\end{tabular} \\
          $(s^1)^2 (s^2)^2(s^3)^2$  & $(\underline{1,0,0})$, $\color{blue}\big(\underline{\tfrac{1}{2},\tfrac{1}{2},0}\big)$, $\big(\tfrac{1}{2},\tfrac{1}{2},\tfrac{1}{2}\big)$    \\
    \rowcolor{cyan!10!}      $\color{red}(s^1)^3 s^2s^3$  & $\color{red}-$  \\
          $\color{red}(s^1)^3 (s^2)^2s^3$  & $\color{red}-$  \\
   \rowcolor{cyan!10!}       $\color{red}(s^1)^3 (s^2)^2(s^3)^2$  & $\color{red}-$  \\
          $(s^1)^3 (s^2)^3 s^3$  & \begin{tabular}{@{}c@{}}$\big(\sqrt{\tfrac{3}{2}},0,0\big)$, $\big(0,\sqrt{\tfrac{2}{3}},0\big)$, $\big(0,0,\frac{1}{\sqrt{2}}\big)$, $\big(\tfrac{1}{\sqrt{6}},\sqrt{\tfrac{2}{3}},0\big)$,\\ $\big(\underline{\tfrac{1}{\sqrt{6}},0},\tfrac{1}{\sqrt{2}}\big)$, $\color{blue}\big(0,\tfrac{1}{2}\sqrt{\tfrac{3}{2}},\tfrac{1}{2\sqrt{2}}\big)$, $\big(\tfrac{1}{\sqrt{6}},\tfrac{1}{\sqrt{6}},\tfrac{1}{\sqrt{2}}\big)$  \end{tabular} \\
   \rowcolor{cyan!10!}       $(s^1)^4 s^2s^3$  &   \begin{tabular}{@{}c@{}} \textcolor{Green}{a)} \  $\color{blue}\big(\tfrac{1}{\sqrt{2}},0,0\big)$, $\big(0,\underline{\tfrac{1}{ \sqrt{2}},0}\big)$, $\big(\underline{\tfrac{1}{ \sqrt{2}},\tfrac{1}{ \sqrt{2}},0}\big)$;\\ \textcolor{Green}{b)} \ $\color{blue}\big(\tfrac{1}{\sqrt{2}},0,0\big)$, $\big(0,\underline{\tfrac{1}{ \sqrt{2}},0}\big)$, $\big(\underline{\tfrac{1}{ \sqrt{2}},\tfrac{1}{ \sqrt{2}},0}\big)$, $\big(\tfrac{1}{ \sqrt{2}},\tfrac{1}{ \sqrt{2}},\tfrac{1}{ \sqrt{2}}\big)$\end{tabular} \\
          $(s^1)^4 (s^2)^2s^3$  & \begin{tabular}{@{}c@{}} \textcolor{Green}{a)} \  $\color{blue}\big(\tfrac{1}{\sqrt{2}},0,0\big)$, $\big(0,1,0\big)$, $\big(0,0,\tfrac{1}{\sqrt{2}}\big)$, $\big(\tfrac{1}{\sqrt{2}},1,0\big)$, $\big(0,\frac{1}{2},\tfrac{1}{\sqrt{2}}\big)$, $\big(\tfrac{1}{\sqrt{2}},0,\tfrac{1}{\sqrt{2}}\big)$;\\ \textcolor{Green}{b)} \  $\color{blue}\big(\tfrac{1}{\sqrt{2}},0,0\big)$, $\big(0,1,0\big)$, $\big(0,0,\tfrac{1}{\sqrt{2}}\big)$, $\big(\tfrac{1}{\sqrt{2}},\tfrac{1}{2},0\big)$, $\big(0,\tfrac{1}{2},\tfrac{1}{\sqrt{2}}\big)$, $\big(\tfrac{1}{\sqrt{2}},0,\tfrac{1}{\sqrt{2}}\big)$, $\color{blue}\big(\tfrac{1}{2},\tfrac{1}{2\sqrt{2}},\tfrac{1}{2\sqrt{2}}\big)$  \end{tabular}  \\
     \rowcolor{cyan!10!}     $\color{red}(s^1)^5 s^2s^3$  &  $\color{red}-$ \\\hline
    \end{tabular}}
    \caption{Bottom-up reconstructions for monomials with 3 moduli. Non-BPS F1-strings are colored in blue.}
    \label{tab:bu3}
\end{table}

\begin{table}[H]
    \centering
    \begin{tabular}{|c|c|}\hline
        \rowcolor{gray!10!} $K\sim-\log P(s)$     & Generating towers \\\hline
       \rowcolor{cyan!10!}   $s^1s^2s^3s^4$& \begin{tabular}{@{}c@{}} \textcolor{Green}{a)} \ $\big(\underline{\tfrac{1}{\sqrt{2}},0,0,0}\big)$, $\big(\underline{\tfrac{1}{\sqrt{2}},\tfrac{1}{\sqrt{2}},0,0}\big)$, $\big(0,\tfrac{1}{\sqrt{2}},\tfrac{1}{\sqrt{2}},\tfrac{1}{\sqrt{2}}\big)$ ; \\ \textcolor{Green}{b)} \ $\big(\underline{\tfrac{1}{\sqrt{2}},0,0,0}\big)$, $\big(\underline{\tfrac{1}{\sqrt{2}},\tfrac{1}{\sqrt{2}},0,0}\big)$, $\color{blue}\big(\tfrac{1}{2\sqrt{2}},\tfrac{1}{2\sqrt{2}},\tfrac{1}{2\sqrt{2}},\tfrac{1}{2\sqrt{2}}\big)$         \end{tabular}\\
        $(s^1)^2s^2s^3s^4$ & \begin{tabular}{@{}c@{}}
           $(1,0,0,0)$, $\big(0,\underline{\tfrac{1}{\sqrt{2}},0,0}\big)$, $\big(0,\underline{\tfrac{1}{\sqrt{2}},\tfrac{1}{\sqrt{2}},0}\big)$, \\$\big(1,\tfrac{1}{\sqrt{2}},0,0\big)$, $\big(\frac{1}{2},\underline{\tfrac{1}{\sqrt{2}},0},0\big)$, $\big(0,\tfrac{1}{\sqrt{2}},\tfrac{1}{\sqrt{2}},\tfrac{1}{\sqrt{2}}\big)$, $\color{blue}\big(\tfrac{1}{2},0,\tfrac{1}{2\sqrt{2}},\tfrac{1}{2\sqrt{2}}\big)$ \end{tabular}\\
        \rowcolor{cyan!10!}  $(s^1)^2(s^2)^2s^3s^4$ & \begin{tabular}{@{}c@{}} \textcolor{Green}{a)} \ $(\underline{1,0},0,0)$, $\big(0,0,\underline{\tfrac{1}{\sqrt{2}},0}\big)$, $\color{blue}\big(\tfrac{1}{2},\tfrac{1}{2},0,0\big)$, \\ $\big(0,0,\tfrac{1}{\sqrt{2}},\tfrac{1}{\sqrt{2}}\big)$,$\big(\tfrac{1}{2},\tfrac{1}{2},\underline{\tfrac{1}{\sqrt{2}},0}\big)$, $\color{blue}\big(\underline{\tfrac{1}{2},0},\tfrac{1}{2\sqrt{2}},\tfrac{1}{2\sqrt{2}}\big)$, $\big(\underline{\tfrac{1}{2},0},\underline{\tfrac{1}{\sqrt{2}},0}\big)$;\\ \textcolor{Green}{b)} \ 
        $(\underline{1,0},0,0)$, $\big(0,0,\underline{\tfrac{1}{\sqrt{2}},0}\big)$, $\color{blue}\big(\tfrac{1}{2},\tfrac{1}{2},0,0\big)$,\\ $\big(0,0,\tfrac{1}{\sqrt{2}},\tfrac{1}{\sqrt{2}}\big)$, $\big(\underline{\tfrac{1}{2},0},\tfrac{1}{\sqrt{2}},0\big)$, $\big(\tfrac{1}{2},\tfrac{1}{2},\tfrac{1}{\sqrt{2}},0\big)$, $\big(\underline{1,0},0,\tfrac{1}{\sqrt{2}}\big)$\end{tabular}\\
        $(s^1)^2(s^2)^2(s^3)^2s^4$  & \begin{tabular}{@{}c@{}}$(\underline{1,0,0},0)$, $\big(0,0,0,\tfrac{1}{\sqrt{2}}\big)$, \\ $\big(\tfrac{1}{2},\tfrac{1}{2},\tfrac{1}{2},0\big)$, $\color{blue}\big(\underline{\tfrac{1}{2},\tfrac{1}{2},0},0\big)$, $(\underline{1,0,0},\tfrac{1}{\sqrt{2}})$\end{tabular}\\
       \rowcolor{cyan!10!}  $\color{red}(s^1)^3s^2s^3s^4$ & $\color{red}-$ \\
        $\color{red}(s^1)^3(s^2)^2s^3s^4$ &$\color{red}-$ \\
        \rowcolor{cyan!10!}  $(s^1)^4s^2s^3s^4$ & \begin{tabular}{@{}c@{}} \textcolor{Green}{a)} \ $\color{blue} \big({\tfrac{1}{\sqrt{2}},0,0,0}\big)$,  $\big(0,\underline{\tfrac{1}{\sqrt{2}},0,0}\big)$, $\big(\underline{\tfrac{1}{\sqrt{2}},\tfrac{1}{\sqrt{2}},0,0}\big)$, $\big(0,\tfrac{1}{\sqrt{2}},\tfrac{1}{\sqrt{2}},\tfrac{1}{\sqrt{2}}\big)$;\\ \textcolor{Green}{b)} \ $\color{blue} \big({\tfrac{1}{\sqrt{2}},0,0,0}\big)$,  $\big(0,\underline{\tfrac{1}{\sqrt{2}},0,0}\big)$, $\big(\underline{\tfrac{1}{\sqrt{2}},\tfrac{1}{\sqrt{2}},0,0}\big)$, $\color{blue}\big(\tfrac{1}{2\sqrt{2}},\tfrac{1}{2\sqrt{2}},\tfrac{1}{2\sqrt{2}},\tfrac{1}{2\sqrt{2}}\big)$
         \end{tabular}\\\hline
    \end{tabular}
    \caption{Bottom-up reconstructions for monomials with 4 moduli. Non-BPS F1-strings are colored in blue.}
    \label{tab:bu4}
\end{table}

\begin{table}[H]
    \centering
   
    \begin{tabular}{|c|c|}\hline
        \rowcolor{gray!10!} $K\sim-\log P(s)$     & Generating towers \\\hline
       \rowcolor{cyan!10!}   $s^1s^2s^3s^4s^5$& \begin{tabular}{@{}c@{}}
       $\big(\underline{\tfrac{1}{\sqrt{2}},0,0,0,0}\big)$, $\big(\underline{\tfrac{1}{\sqrt{2}},\tfrac{1}{\sqrt{2}},0,0,0}\big)$, $\big(\tfrac{1}{\sqrt{2}},\tfrac{1}{\sqrt{2}},\tfrac{1}{\sqrt{2}},0,0\big)$,\\ $\big(\tfrac{1}{\sqrt{2}},0,0,\tfrac{1}{\sqrt{2}},\tfrac{1}{\sqrt{2}}\big)$, $\color{blue}\big(0,\tfrac{1}{2\sqrt{2}},\tfrac{1}{2\sqrt{2}},\tfrac{1}{2\sqrt{2}},\tfrac{1}{2\sqrt{2}}\big)$
          \end{tabular}\\
        $\color{red}(s^1)^2s^2s^3s^4s^5$& $\color{red}-$\\
        \rowcolor{cyan!10!}$(s^1)^2(s^2)^2s^3s^4s^5$& \begin{tabular}{@{}c@{}} $(\underline{1,0},0,0,0)$, $\big(0,0,\underline{\tfrac{1}{\sqrt{2}},0,0}\big)$, $\big(0,0,\underline{\tfrac{1}{\sqrt{2}},\tfrac{1}{\sqrt{2}},0}\big)$,\\ $\big(\tfrac{1}{2},\tfrac{1}{2},\underline{\tfrac{1}{\sqrt{2}},0},0\big)$,$\big(\underline{\tfrac{1}{2},0},\underline{\tfrac{1}{\sqrt{2}},0},0\big)$,$(\underline{1,0},0,0,\tfrac{1}{\sqrt{2}})$,\\$\big(0,0,\tfrac{1}{\sqrt{2}},\tfrac{1}{\sqrt{2}},\tfrac{1}{\sqrt{2}}\big)$, $\color{blue}\big(\tfrac{1}{2},\tfrac{1}{2},0,0,0\big)$, $\color{blue}\big(\underline{\tfrac{1}{2},0},\tfrac{1}{2\sqrt{2}},\tfrac{1}{2\sqrt{2}},0 \big)$
         \end{tabular}\\
       $\color{red}(s^1)^3s^2s^3s^4s^5$& $\color{red}-$\\\hline
          \end{tabular}
    \caption{Bottom-up reconstructions monomials with for 5 moduli. Non-BPS F1-strings are colored in blue.}
    \label{tab:bu5}
\end{table}

\begin{table}[H]
    \centering

    \begin{tabular}{|c|c|}\hline
        \rowcolor{gray!10!} $K\sim-\log P(s)$     & Generating towers \\\hline
       \rowcolor{cyan!10!}   $s^1s^2s^3s^4s^5s^6$& \begin{tabular}{@{}c@{}} 
        $\big(\underline{\tfrac{1}{\sqrt{2}},0,0,0,0,0}\big)$, $\big(\underline{\tfrac{1}{\sqrt{2}},\tfrac{1}{\sqrt{2}},0,0,0,0}\big)$, $\big(\tfrac{1}{\sqrt{2}},\tfrac{1}{\sqrt{2}},\tfrac{1}{\sqrt{2}},0,0,0\big)$, \\ $\big(\tfrac{1}{\sqrt{2}},0,0,\tfrac{1}{\sqrt{2}},\tfrac{1}{\sqrt{2}},0\big)$, $\big(0,\tfrac{1}{\sqrt{2}},0,\tfrac{1}{\sqrt{2}},0,\tfrac{1}{\sqrt{2}}\big)$,$\big(0,0,\tfrac{1}{\sqrt{2}},0,\tfrac{1}{\sqrt{2}},\tfrac{1}{\sqrt{2}}\big)$, \\$\color{blue}\big(\underline{\tfrac{1}{2\sqrt{2}},0,0},\tfrac{1}{2\sqrt{2}},\tfrac{1}{2\sqrt{2}},\tfrac{1}{2\sqrt{2}}\big)$, $\color{blue}\big(0,\tfrac{1}{2\sqrt{2}},\tfrac{1}{2\sqrt{2}},\tfrac{1}{2\sqrt{2}},\tfrac{1}{2\sqrt{2}},0\big)$, \\ $\color{blue}\big(\tfrac{1}{2\sqrt{2}},0,\tfrac{1}{2\sqrt{2}},\tfrac{1}{2\sqrt{2}},0,\tfrac{1}{2\sqrt{2}}\big)$, $\color{  blue}\big(\tfrac{1}{2\sqrt{2}},\tfrac{1}{2\sqrt{2}},0,0,\tfrac{1}{2\sqrt{2}},\tfrac{1}{2\sqrt{2}}\big)$
          \end{tabular}\\
          $\color{red}(s^1)^2s^2s^3s^4s^5s^6$&$\color{red}-$
          \\\hline
           \end{tabular}
    \caption{Bottom-up reconstructions for monomials with 6 moduli. Non-BPS F1-strings are colored in blue.}
    \label{tab:bu6}
\end{table}

\begin{table}[H]
    \centering
    
    \begin{tabular}{|c|c|}\hline
        \rowcolor{gray!10!} $K\sim-\log P(s)$     & Generating towers \\\hline
        \rowcolor{cyan!10!}
          $s^1s^2s^3s^4s^5s^6s^7$& \begin{tabular}{@{}c@{}}
       $\big(\underline{\tfrac{1}{\sqrt{2}},0,0,0,0,0,0}\big)$, $\big(\underline{\tfrac{1}{\sqrt{2}},\tfrac{1}{\sqrt{2}},0,0,0,0,0}\big)$, $\big(\tfrac{1}{\sqrt{2}},\tfrac{1}{\sqrt{2}},\tfrac{1}{\sqrt{2}},0,0,0,0\big)$, \\ $\big(\tfrac{1}{\sqrt{2}},0,0,\tfrac{1}{\sqrt{2}},\tfrac{1}{\sqrt{2}},0,0\big)$, $\big(\tfrac{1}{\sqrt{2}},0,0,0,0,\tfrac{1}{\sqrt{2}},\tfrac{1}{\sqrt{2}}\big)$, $\big(0,\tfrac{1}{\sqrt{2}},0,\tfrac{1}{\sqrt{2}},0,\tfrac{1}{\sqrt{2}},0\big)$,\\ $\big(0,\tfrac{1}{\sqrt{2}},0,0,\tfrac{1}{\sqrt{2}},0,\tfrac{1}{\sqrt{2}}\big)$, $\big(0,0,\tfrac{1}{\sqrt{2}},0,\tfrac{1}{\sqrt{2}},\tfrac{1}{\sqrt{2}},0\big)$, $\big(0,0,\tfrac{1}{\sqrt{2}},\tfrac{1}{\sqrt{2}},0,0,\tfrac{1}{\sqrt{2}}\big)$, \\$\color{blue}\big(0,0,0,\tfrac{1}{2\sqrt{2}},\tfrac{1}{2\sqrt{2}},\tfrac{1}{2\sqrt{2}},\tfrac{1}{2\sqrt{2}}\big)$, $\color{blue}\big(0,\tfrac{1}{2\sqrt{2}},\tfrac{1}{2\sqrt{2}},0,0,\tfrac{1}{2\sqrt{2}},\tfrac{1}{2\sqrt{2}}\big)$, \\$\color{blue}\big(0,\tfrac{1}{2\sqrt{2}},\tfrac{1}{2\sqrt{2}},\tfrac{1}{2\sqrt{2}},\tfrac{1}{2\sqrt{2}},0,0\big)$, $\color{blue}\big(\tfrac{1}{2\sqrt{2}},0,\tfrac{1}{2\sqrt{2}},0,\tfrac{1}{2\sqrt{2}},0,\tfrac{1}{2\sqrt{2}}\big)$,\\ $\color{blue}\big(\tfrac{1}{2\sqrt{2}},0,\tfrac{1}{2\sqrt{2}},\tfrac{1}{2\sqrt{2}},0,\tfrac{1}{2\sqrt{2}},0\big)$, $\color{blue}\big(\tfrac{1}{2\sqrt{2}},\tfrac{1}{2\sqrt{2}},0,\tfrac{1}{2\sqrt{2}},0,0,\tfrac{1}{2\sqrt{2}}\big)$, \\$\color{blue}\big(\tfrac{1}{2\sqrt{2}},\tfrac{1}{2\sqrt{2}},0,0,\tfrac{1}{2\sqrt{2}},\tfrac{1}{2\sqrt{2}},0\big)$ 
          \end{tabular}
          \\\hline
           \end{tabular}
    \caption{Unique bottom-up reconstruction fora  monomial with 7 moduli. Non-BPS F1-strings are colored in blue.}
    \label{tab:bu7}
\end{table}

It is tantalizing to see that in Table \ref{tab:bu7}, the only arrangement of towers allowed coincides exactly with the tower polytope of M-theory on the Joyce $G_2$-manifold $\mathbb{T}^7/(\mathbb{Z}_2\oplus\mathbb{Z}_2\oplus\mathbb{Z}_2)$.         
Also notice how, for up to 4 moduli, there are monomials for which the reconstruction algorithm gives two possible convex hull, while all the others yield one unique arrangement of towers. In particular, for the monomials that allow two different tower arrangements, the only difference between the two is the presence of a single $n=1$ KK tower, exactly what we observed being the difference between type IIA/heterotic $E_8 \times E_8$  (which have an M-theory limit) and type IIB/type I/heterotic $SO(32)$ (which do not) in \cite{Grieco:2025bjy}.

Finally, we see that not all choices of monomial are allowed. In particular, it is always forbidden for a saxion to have a power of 5, and if a monomial yields no convex hull in a lower dimensional moduli space, then that precise arrangement of monomials will be forbidden also if we add more saxions. An intuitive explanation for the later is the fact that the larger saxionic cone would have the smaller one as a boundary, which does not exist to start with.

\subsubsection{Constraints on gluing and general polynomials} \label{sec.gluing}

Assuming a K\"ahler potential of the form $K\sim -\log P(s)$, with $P(s)$ a  monomial in the saxions with $\deg P(s)\leq 7$, we have put constrains on the possible monomials $P(s)$ consistent with the Integral Scaling Relation and the ESC. In doing so, we have seen that there are some choices of K\"ahler potentials that are not consistent with an arrangement of light towers satisfying the above assumptions. However, in general, $P(s)$ will not be a monomial but a homogeneous polynomial to leading order. This is for example, the case for generic CY threefolds, where the volume is given by $\mathcal{V}_X=\frac{1}{3!}\kappa_{abc}s^as^bs^c$ (this remains the case for CY orientifolds), and where the K\"ahler potential for the K\"ahler moduli in type IIA or heterotic compactifications to 4d $\mathcal{N}=1$ results in an asymptotic form $K\sim-\log \mathcal{V}_X$. One can then wonder how our approach extends to the general case, and whether there are additional constrains taking place there.

Once $P(s)$ is not a single monomial, the moduli space metric $\mathsf{G}_{ij}=-\frac{1}{2}\partial_i\partial_j\log P(s)$ will be non-diagonal and generically with non-vanishing curvature.\footnote{For two saxions straightforward computation shows that the saxionic part of the metric is still flat. This is no longer the case for 3 or more saxions. See \cite{trenner2010asymptoticcurvaturemodulispaces,Marchesano:2023thx,Marchesano:2024tod,Castellano:2024gwi,Blanco:2025qom,Castellano:2026bnx,Aoufia:2026mqb} for more results regarding the curvature of  moduli spaces in CY compactifications.} This would in principle prevent us from doing a classification of the global arrangements of light towers along infinite distance limits, since the saxionic cone $\Delta$ cannot be identified with the set of asymptotic trajectories on its tangent space. Fortunately for us, as explored in \cite{Grimm:2018cpv,Corvilain:2018lgw}, we can split the asymptotic regions of our saxionic cone in a set of \emph{growth sectors} $\mathcal{R}_{i_1\dots i_n}$ \eqref{eq.growth sector},
\begin{equation}\label{eq. growth sec4}
	\Delta\supseteq\bigcup_{(i_1,\dots,i_n)}\mathcal{R}_{i_1\dots i_n}=\bigcup_{(i_1,\dots,i_n)}\left\{s^{i_1}\gtrsim s^{i_2}\gtrsim\dots \gtrsim s^{i_n}\gg 1\right\}\,,
\end{equation}
where there is a strict hierarchy between the different saxions, and $P(s)$ can be approximated by its leading term as we approach an infinite distance limit. Up to subleading terms, the metric is approximately diagonal,\footnote{Note that, in doing this, for those saxions not entering the leading term of $P(s)$, the metric becomes effectively degenerate as we approach an infinite distance limit. This means that, in dividing $\Delta$ in growth sectors, the effective dimensionality of each of them, given by the number of saxions in the leading term of $P(s)$, can change from one sector to the other.} with corrections affecting only second order terms like curvature, being subleading in the computation of angles and distances.\\
When there is a parametric separation in the growth of the different saxions in a given growth sector \eqref{eq. growth sec4}, i.e., $s^{i_1}\gg s^{i_2}\gg\dots \gg s^{i_n}$, subleading corrections are parametrically suppressed. One could then take the possible tower arrangements consistent with the leading term of $P(s)$ within each growth sector (selecting only the light towers that actually become light as one moves within $\mathcal{R}_{i_1\dots i_n}$), and try to glue the different arrangements together. In this sense, the arrangement of towers found in the previous section will serve as building blocks to construct the more general tower convex hull associated to a general polynomial $P(s)$.

The non-trivial aspects come from how to glue in a consistent way the convex hulls valid in each individual growth sector. To do this, we need to understand how to continuously interpolate between different tower arrangements, which implies understanding what happens to the towers when moving along the interfaces connecting different growth sectors. Things become more interesting in these interfaces, where two or more saxions scale at similar rates, e.g., $s^i\sim \alpha s^j\to\infty$, with $\alpha$ a positive number,\footnote{By construction these loci have measure zero along the set of infinite distance limits, with the interior of the different growths sectors covering \emph{almost all} the asymptotic regions of the saxionic cone, \cite{Castellano:2023jjt,Etheredge:2024tok}.} and several terms in $P(s)$ being co-leading. The coefficient(s) $\alpha$ can be understood as an \emph{impact parameter}, and asymptotic trajectories with different values of $\alpha$ can be shown to be parallel (i.e., always having a finite distance between them), hence the name. Since $\alpha\to 0$ or $\infty$ recovers $s^j\gg s^i$ and $s^i\gg s^j$, respectively, and the two growth sectors will generically have different leading $P(s)$ terms, and thus incompatible tower arrangements, as one varies $\alpha$ the $\zeta$-vectors of the light towers will move or \emph{slide} perpendicularly to the trajectory taken \cite{Etheredge:2024tok,Grieco:2025bjy}, interpolating between the tower arrangements of the two sectors. The perpendicular character of the sliding means that, for finite $\alpha$, the exponential rates of the light towers do not change, and thus depend only on the direction of the asymptotic trajectory (and not its impact parameter). 

\subsection*{Calabi-Yau examples and changes in the fibration structure}

A natural arena where we can illustrate the above discussion on sliding of the $\zeta$-vectors for non-homogeneous $P(s)$ are CY compactifications. In heterotic $E_8\times E_8$ compactifications (similarly for type IIA string theory on CY orientifolds, see Appendix \ref{app.topdown}), the 4d $\mathcal{N}=1$ K\"ahler potential is given by \cite{Grimm:2005fa}
\begin{equation}\label{eq.kahler cy}
	K=-\log s^0-\log\mathcal{V}_X
\end{equation}
modulo instantonic corrections, where $s^0=e^{-2\phi_4}$ (with $\phi_4$ the 4d dilaton) and $\mathcal{V}_X=\frac{1}{3!}\kappa_{abc}s^as^bs^c$ is the 3-fold volume in string units, expressed in terms of the intersection numbers and K\"ahler saxions
\begin{equation}
	J=s^a\omega_a\,,\quad \kappa_{abc}=\int_{X}\omega_a\wedge\omega_b\wedge\omega_c\,,
\end{equation}
with $J$ the K\"ahler form of $X$ and $\{\omega_a\}_{a=1}^{h^{1,1}}$ are an integral base for $H^{1,1}(X)$, dual to a set of divisors $\{[\omega_a]\}_{a=1}^{h^{1,1}}$. We refer to \cite{Lanza:2021udy,Grieco:2025bjy} and Appendix \ref{app.topdown} for more details. Due to the factorization of $s^0$, the associated saxionic direction is always perpendicular to that spanned by the K\"ahler saxions, and thus $\Delta=\mathbb{R}^+\oplus\mathcal{K}(X)$, where $\mathcal{K}(X)$ is the K\"ahler cone of $X$, parametrized by $\{s^a\geq 0\}_{a=1}^{h^{1,1}}$.

For illustrative purposes, we will take $h^{1,1}=2$ and volume forms of the type
\begin{equation} \label{eq.2saxpoly}
	3!\mathcal{V}_X=\kappa_{111}(s^{1})^3+3\kappa_{112}(s^1)^2s^2+3\kappa_{122}s^1(s^2)^2\,,
\end{equation} 
with all the above intersection numbers non-vanishing. Hodge index theorem \cite{hu2025positivityshadowhodgeindex} requires that the matrix $\kappa_1=\begin{psmallmatrix}\kappa_{111}&\kappa_{112}\\\kappa_{112}&\kappa_{122}\end{psmallmatrix}$ has signature $(+,-)$, which requires
\begin{equation}\label{eq. ineq}
	\det(\kappa_1)=\kappa_{111}\kappa_{122}-\kappa_{112}^2\leq 0\,.
\end{equation}
For this specific volume form, the above condition is enough to guarantee the positive definiteness of the moduli space metric for the saxions, $\mathsf{G}_{ij}=\frac{1}{2}\partial_i\partial_jK(s)$, which for our particular setting is automatically flat and admits global flat coordinates in terms of $s^0$, the overall volume $\mathcal{V}_X$ and a transverse coordinate to this. We will see that in this setup requiring the positive semi-definiteness of the metric is enough to guarantee that the towers behave as expected under sliding, so that we can reach one growth sector's polytope from the other through sliding of the perpendicular coordinates to the total volume. This also guarantees that, from the bottom-up point of view, the gluing of single monomials in \eqref{eq.2saxpoly} is not obstructed.

Regarding the light towers, we we will borrow the results from \cite[Section 3]{Grieco:2025bjy} on the mass scales for the leading towers of states and refer to that paper for more details. As explained there, the string oscillator modes, the KK tower associated to the $\mathcal{V}_X$  and the M-theory interval KK modes  give rise to vectors that do not slide, since they are well defined over the whole saxionic cone:
\begin{equation}
\vec{\zeta}_{\rm osc}=\left(\frac{1}{\sqrt{2}},0,0\right)\,,\quad \vec{\zeta}_{\rm KK,X}=\left(\frac{1}{\sqrt{2}},\frac{1}{\sqrt{6}},0\right)\,,\quad \vec{\zeta}_{\rm KK,M\text{-}th}=\left(0,\sqrt{\frac{3}{2}},0\right)
\end{equation}

The only towers sliding from one growth sector to the other are the ones associated to 4- and 2- cycles, because the fibration structure might change between sectors. Our interest will focus in the KK towers associated to the $[\omega_1]$ divisor, with volume $\mathcal{V}_1=\frac{1}{2}\kappa_{1ab}s^as^b=\frac{1}{2}\kappa_{111}(s^1)^2+\kappa_{112}s^1s^2+\frac{1}{2}\kappa_{122}(s^2)^2$, as well as the tower resulting from the bound states of the M-theory interval KK modes and the wrapped M2-branes towers on the curve $\mathcal{C}_{22}=[\omega_2]\cap[\omega_2]$, with $\mathcal{V}_{22}=\kappa_{122}s^1$. Since $X$ is not a direct product, the fibration structure changes depending on the direction along moduli space, with the associated $\zeta$-vectors given by\footnote{The mass scale of the two towers is given by \cite[Section 3.1]{Grieco:2025bjy}
\begin{equation}
\frac{m_{\rm KK,[\omega_1]}}{M_{\rm Pl,4}}\sim(s^0)^{-1/2}\mathcal{V}_1^{-1/4}\,,\qquad \frac{m_{\rm KK,M\text{-}th+M2}}{M_{\rm Pl,4}}\sim (\mathcal{V}_{22})^{1/2}\mathcal{V}_X^{-1/2}\,.
\end{equation}
The multi-term expression of the volume forms results in the following non-trivial components along the directions transverse to the overall volume:
\begin{equation}
	\hat{f}(\alpha)=\frac{\alpha  \sqrt{\alpha\kappa_{122}(\alpha\kappa_{122}+\kappa_{112})+\kappa_{112}^2-\kappa_{111} \kappa_{122}}}{\sqrt{3} \left(2 \alpha ^2 \kappa_{122}+4 \alpha  \kappa_{112}+2 \kappa_{111}\right)}\,,\quad \hat{g}(\alpha)=\frac{2 \alpha  \kappa_{122}+\kappa_{112}}{2 \sqrt{3} \sqrt{\alpha\kappa_{122}(\alpha\kappa_{122}+\kappa_{112})+\kappa_{112}^2-\kappa_{111} \kappa_{122}}}
\end{equation}
}

\begin{subequations}\label{eq.sliding.general}
\begin{align}
	\vec{\zeta}_{\rm KK,[\omega_1]}=\left(\frac{1}{\sqrt{2}},\frac{1}{\sqrt{6}},\hat{f}(\alpha)\right)\,,\; \text{with}\; \hat{f}(\alpha)\approx\left\{\begin{array}{ll}
	0&\text{for }\alpha\to 0\\
	\frac{1}{2\sqrt{3}}&\text{for }\alpha\to\infty
	\end{array}\right.\\
	\label{eq.sliding.gen.2}
	\vec{\zeta}_{\rm KK,M\text{-}th+M2}=\left(0,\sqrt{\frac{2}{3}},\hat{g}(\alpha)\right)\,,\; \text{with}\;\, \hat{g}(\alpha)\approx\left\{\begin{array}{ll}
	\frac{\kappa_{112}}{2\sqrt{3}\sqrt{\kappa_{112}^2-\kappa_{111}\kappa_{122}}}&\text{for }\alpha\to 0\\
	\frac{1}{\sqrt{3}}&\text{for }\alpha\to\infty
	\end{array}\right.
	.\end{align}
\end{subequations}
in flat coordinates, where we are taking $s^2\sim\alpha s^1$, see Figure \ref{f.slidingex}.

 \begin{figure}[ht!]
\begin{center}
\includegraphics[width=0.5\textwidth]{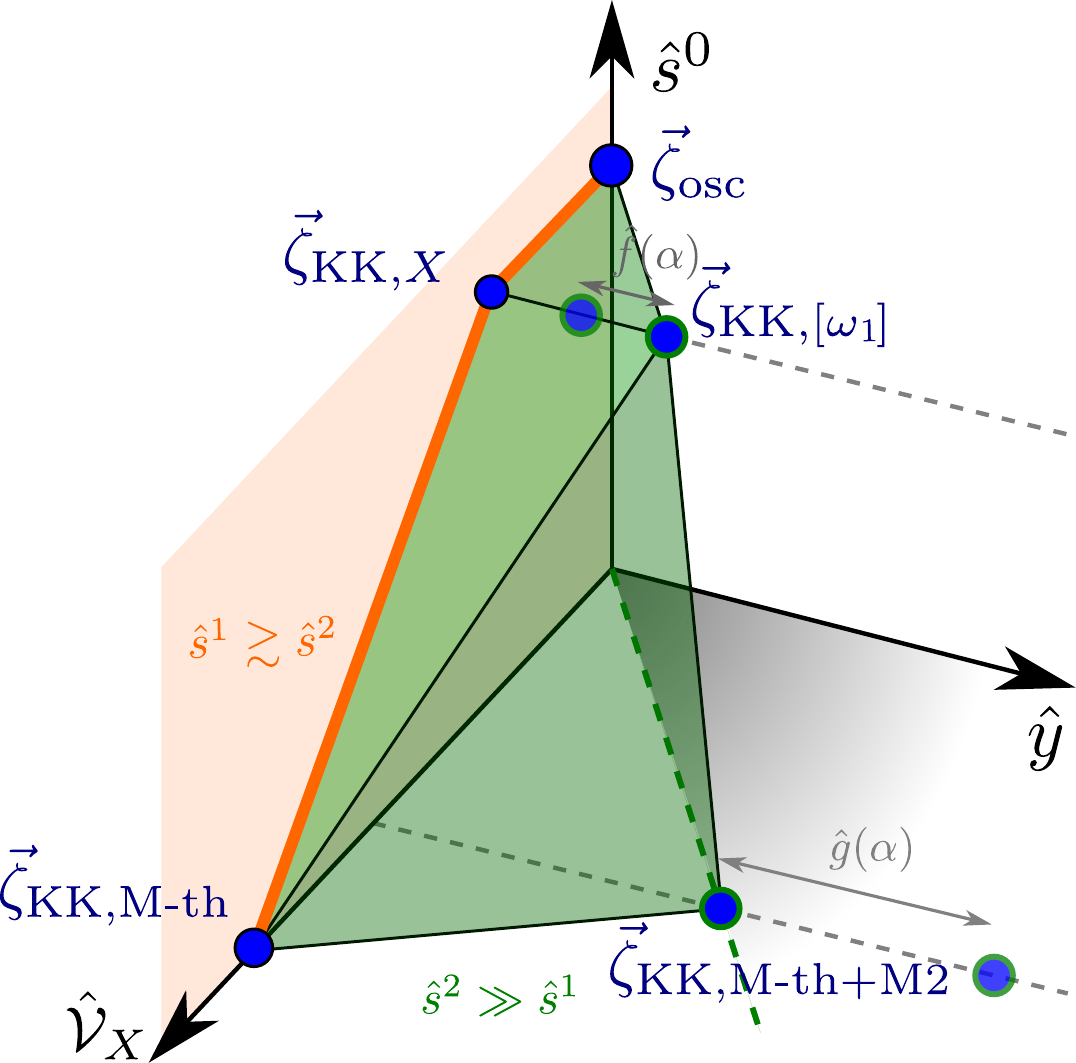}
\caption{Sliding of towers associated to fibrations of CY three-fold with volume \eqref{eq.2saxpoly}, in canonical coordinates, with saxionic plane given by the canonical volume $\hat{
\mathcal{V}}_X$ and its perpendicular direction $\hat{y}$. The two growth sectors $\hat{s}^1\gtrsim \hat{s}^2$ and $\hat{s}^2\geq \hat{s}^1$ are depicted in orange and green respectively. Towers on the orange plane do not slide, while on the green polytope they do according to \eqref{eq.sliding.general}. Notice how, when sliding from the green to orange sector, $	\vec{\zeta}_{\rm KK,[\omega_1]}$ collapses onto the total volume tower while $\vec{\zeta}_{\text{ KK,M-th+M2}}$ becomes subleading. 
\label{f.slidingex}}
\end{center}
\end{figure}

We then have that finite $\alpha$ interpolates between the following two limiting behaviors:
\begin{itemize}
	\item[$\alpha\to\infty$]: In this case we find ourselves in the bulk of the $\mathcal{R}_{2,1}=\{s^2\gg s^1\}$ growth sector, the one we are starting from, with $\mathcal{V}_X\sim \frac{1}{2}\kappa_{122}s^1(s^2)^2$, $\mathcal{V}_1\sim\frac{1}{2}\kappa_{122}(s^2)^2$ and $\mathcal{V}_{22}=\kappa_{122}s^1$. The different relevant submanifolds can be rescaled independently. The length of the $\zeta$-vectors is $|\vec{\zeta}_{\rm KK,[\omega_1]}|=\frac{\sqrt{3}}{2}$ and $|\vec{\zeta}_{\rm KK,M\text{-}th+M2}|=1$, the expected length for KK towers respectively decompactifying 4 and 2 internal dimensions.
	\item[$\alpha\to0$]: Here we are sliding to the $\mathcal{R}_{1,2}=\{s^1\gg s^2\}$ growth sector, with volumes $\mathcal{V}_X\sim \frac{1}{6}\kappa_{111}(s^1)^3$, $\mathcal{V}_1\sim\frac{1}{2}\kappa_{111}(s^1)^2$ and $\mathcal{V}_{22}=\kappa_{122}s^1$, all proportional to each other. All dependence of $s^2$ becomes subleading and the moduli space becomes effectively two dimensional, parametrized by $s^0$ and $\mathcal{V}_X$. This means that one cannot rescale independently the $[\omega_1]$ and $[\omega_2]\cap[\omega_2]$ cycles, which grow at the same rate as the whole $X$. This is very clear in the norm $|\vec{\zeta}_{\rm KK,[\omega_1]}|=\sqrt{\frac23}=\sqrt{\frac{2+6}{2\cdot 6}}$, i.e., what is expected from a KK tower decompactifying 6 dimensions. Indeed the vector collapses to $\vec{\zeta}_{\rm KK,[\omega_1]}=\vec{\zeta}_{\rm KK,X}=\left(\frac{1}{\sqrt{2}},\frac{1}{\sqrt{6}},0\right)$. Something different occurs for $\vec{\zeta}_{\rm KK,M\text{-}th+M2}$, since it involves an M2-brane wrapped on the curve $[\omega_2]\cap[\omega_2]$, which in this growth sector does not become small relatively to the overall volume. The tower $\vec{\zeta}_{\rm KK,M\text{-}th+M2}$ is then subleading with respect to $\vec{\zeta}_{\rm KK,M\text{-}th}$, which becomes the leading tower in this regime, but remains lighter than the species scale $M_{\rm Pl,5}$. As a subleading tower, the vector $|\vec{\zeta}_{\rm KK,M\text{-}th+M2}|=\sqrt{\frac{2}{3}+\hat{g}(\alpha)^2}$ does not follow the taxonomy rules from \cite{Etheredge:2024tok} (see \eqref{eq. taxonomy} later), but will still follow the Integral Scaling Relation \eqref{scalingweight}.
	
	\end{itemize}
	It is interesting to notice that, depending on whether the inequality \eqref{eq. ineq} is saturated or not, the $\zeta$-vector slides to infinity and disappears or remains at finite distance. 
This can be see by noting that when $\kappa_{111}\kappa_{122}=\kappa_{112}^2$, $\hat{g}(\alpha)$ in \eqref{eq.sliding.gen.2} diverges as $\alpha\to 0$. In this case the metric $\mathsf{G}_{ij}=-\frac{1}{2}\partial_i\partial_j\log \mathcal{V}_X$ has
	\begin{equation}
		\det\mathsf{G}=\frac{3}{4\mathcal{V}_X}\left[(\kappa_{112}^2-\kappa_{111}\kappa_{122})+\tfrac{s^2}{s^1}\kappa_{112}\kappa_{122}+\Big(\tfrac{s^2}{s^1}\Big)^2\kappa_{122}^2\right]\,,
	\end{equation}
	so that the moduli space metric becomes degenerate if $\kappa_{111}\kappa_{122}=\kappa_{112}^2$ and $\frac{s^2}{s^1}=\alpha\to 0$ (even if we keep the overall volume $\mathcal{V}_X$ constant). We see that in this case, along these limits not only the moduli space metric gets effectively dominated by $s^1$, but actually the first subleading term depending in $s^2$ disappears.\footnote{This can also be seen from
	\begin{equation}
	\mathsf{G}_{22}=\frac{9\kappa_{112}^2-6\kappa_{111}\kappa_{122}}{2(s^1)^2\kappa_{111}^2}-\frac{27\kappa_{112}}{(s^2)^2\kappa_{111}^3}\alpha(\kappa_{112}^2-\kappa_{111}\kappa_{122})+\mathcal{O}(\alpha^3)\,,
	\end{equation}
	where the dependence on $\alpha$ (i.e., $s^2$) is higher order for the saturation of $\kappa_{112}^2\geq\kappa_{111}\kappa_{122}$.
	} It is expected that this degeneracy is lifted after the inclusion of subleading corrections to \eqref{eq.kahler cy}.

This saturation of the inequality occurs in famous examples such as the widely known CY threefold $\mathbb{P}^{(1,1,1,6,9)}\left[18\right]$ \cite{Candelas:1994hw}, for which the only non-vanishing intersection numbers are $\kappa_{111}=9$, $\kappa_{112}=3$ and $\kappa_{122}=\kappa_{212}=\kappa_{221}=1$, which indeed saturate  $\kappa_{111}\kappa_{122}=\kappa_{112}^2$.
While the above example shows that one can find threefolds saturating $\kappa_{111}\kappa_{122}\leq\kappa_{112}^2$, we do not expect this to be always the case. A check in the Kreuzer-Skarke database \cite{Kreuzer:2000xy, Kreuzer:2000qv} of the 36 three-folds which are hypersurfaces in toric varieties and have two K\"ahler  moduli ($h^{(1,1)}=2$) reveals that there are multiple examples where the inequality is strict at least for some combination of monomials/intersection numbers.

As mentioned above, the inequality \eqref{eq. ineq} in the intersection numbers comes from Hodge index theorem, and is equivalent to the positivity of the K\"ahler metric for this class of CY threefolds. In these examples, \eqref{eq. ineq} guarantees that the tower rates in \eqref{eq.sliding.general} are well defined, and that the convex hulls of the tower vectors in the two growth sectors can be ``glued'' to each other via the sliding of the $\zeta$-vectors perpendicularly to the asymptotic direction.
Similar inequalities should apply to general constructions with three or more moduli. As the number of monomials and saxions in the volume polynomial grows, more involved inequalities on the intersection numbers are expected to come from Hodge index theorem/positivity of the metric (see \cite{NaomiLuca}). Similarly, one would be able to derive involved expressions for the sliding functions of the KK towers.\\

In summary, we have seen that from a top-down perspective, if we assume that the sliding towers are associated to KK modes or branes wrapping internal cycles of the compact manifold, the various growth sectors are associated to different emergent fibration structures with a parametric separation between base(s) and fiber(s). This difference in the growth rates of cycles precisely allows us to obtain light towers with different scaling rates. Thus, when interpolating from one growth sector to another, the $\zeta$-vectors of the light towers can experience two behaviors: \textbf{(A)} they collapse on top of another tower $\vec{\zeta}_{{\rm KK},\mathcal{C}_n}\to \vec{\zeta}_{{\rm KK},\mathcal{D}_m}$, as now the associated cycles $\mathcal{C}_n$ and $\mathcal{D}_m$ have the same moduli dependence and grow at the same rate, rather than independently as before, or \textbf{(B)} the tower now becomes subleading, and the $\zeta$-vector projects in the interior of the convex hull. Since the sliding is perpendicular to the direction we are moving in moduli space, precisely given by $\vec{\zeta}_{\mathcal{T}}$ for EFT string limits, the Integral Scaling Relation holds between tower and EFT string. However, since the tower is now subleading, the perpendicular components need not follow the taxonomy rules on the vector length \eqref{eq.length}
 \cite{Etheredge:2022opl,Etheredge:2024tok}. The end-point of the sliding depends on the particular details of the tower and the moduli space metric and indeed, if the later becomes degenerate (at least at leading order in the perturbative corrections), the $\zeta$-vector can even slide to infinity. Regarding towers coming from the oscillator modes of fundamental strings, they are not affected by the change of fibration structure, since their tension follows a universal relation with the 4d dilaton, $\mathcal{T}_{\rm F1}\sim M_{\rm Pl,4}^2e^{2\phi_4}$.

\subsection*{General picture for two saxions and higher dimensional moduli spaces}
We can now expand the above discussion to general K\"ahler potentials with two saxions, of the form $K=-\log P(s^1,s^2)$. As anticipated above, we first need to select those $P(s^1,s^2)$ resulting in a positive definite metric $\mathsf{G}_{ij}=\frac{1}{2}\partial_{s^i}\partial_{s^j}K$, that from the bottom-up point of view is necessary (but not sufficient, as we will see) for the towers to be able to connect from one growth sector to the other. Then we also need to discard those polynomials for which the tower arrangements of the different growth sectors cannot be glued together, which requires having the same tower arrangement along the interface between growth sectors. Taking a given homogeneous polynomial of degree $p$,
\begin{equation}\label{eq. hom pol}
P(s^1,s^2)=\sum_{k=0}^pa_k(s^1)^k(s^2)^{p-k}
\end{equation}
we have that in the given growth sectors \eqref{eq.growth sector}, the polynomial is approximated as
\begin{equation}
	P(s^1,s^2)\approx \left\{\begin{array}{ll}
	a_{k_{\rm max}}(s^1)^{k_{\rm max}}(s^2)^{p-{k_{\rm max}}}&\text{for } s^1\gg s^2\\
	a_{k_{\rm min}}(s^1)^{k_{\rm min}}(s^2)^{p-{k_{\rm min}}}&\text{for } s^1\ll s^2
	\end{array}\right.\,,
\end{equation}
where $k_{\rm min}$ and $k_{\rm max}$ are the minimum and maximum values of the set $\{k\,:\,a_k\neq 0\}$. The above two terms are those that will set the tower arrangement in the given growth sectors, while the additional intermediate terms will be subleading, and thus irrelevant, save for the interface $s^1\sim s^2$, where $P(s^1,s^2)$ grows homogeneously. Thus, in order for the $(s^1)^{k_{\rm max}}(s^2)^{p-{k_{\rm max}}}$ and $(s^1)^{k_{\rm min}}(s^2)^{p-{k_{\rm min}}}$ terms to ``glue'' consistently, they need to allow for the same tower arrangement on the interface $s^1\sim s^2$, see Figure \ref{fig.growthGLUE}. As an example, in the interface between the growth sectors with $(s^1)^4s^2$ there are no towers (as expected from  Table \ref{tab:1sax}), which means that such term cannot be glued with another, while it is consistent on its own. 

\begin{figure}
\begin{center}
\includegraphics[width=\textwidth]{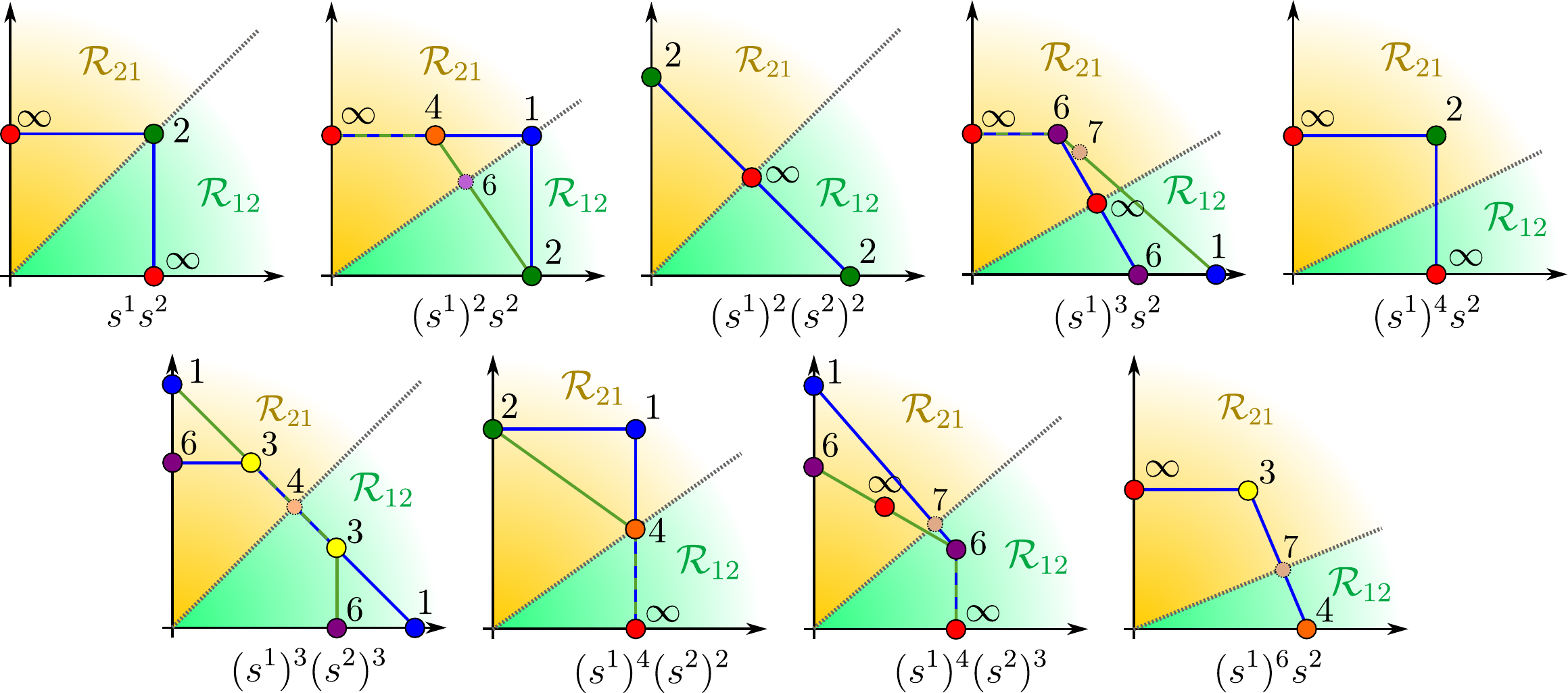}
\caption{Possible tower arrangements consistent with $K\sim -\log (s^1)^k(s^2)^{p-k}$, with $2\leq p\leq 7 $ and $1\leq k\leq 6$ (the towers arrangements for a single modulus appear in Table \ref{tab:1sax}). Notice the different coloring when there is more than one consistent tower arrangement. In green and yellow we depict the growth sectors $\mathcal{R}_{12}=\{s^1\gg s^2\}$ and $\mathcal{R}_{21}=\{s^2\gg s^1\}$, together with their interface in gray. All quantities are depicted in canonically normalized coordinates.
\label{fig.growthGLUE}}
\end{center}
\end{figure}

Thus, given a degree $p$ of our homogeneous polynomial \eqref{eq. hom pol}, we can study if the gluing condition imposes additional constrains on their leading terms along the two growth sectors. For $p=2$, 3 and 4 there is no surprise: all combinations are consistent. This is evident from Figure \ref{fig.growthGLUE} and Table \ref{tab:1sax}, as all monomials of said degrees are allowed and have the same tower along the interface between growth sectors. In Figure \ref{f.glue} all the polytopes resulting from gluing of $p=3$ two-saxion monomials are depicted.
    \begin{figure}
\begin{center}
\begin{subfigure}[b]{0.49\textwidth}
\captionsetup{width=\linewidth}
\center
\includegraphics[width=\textwidth]{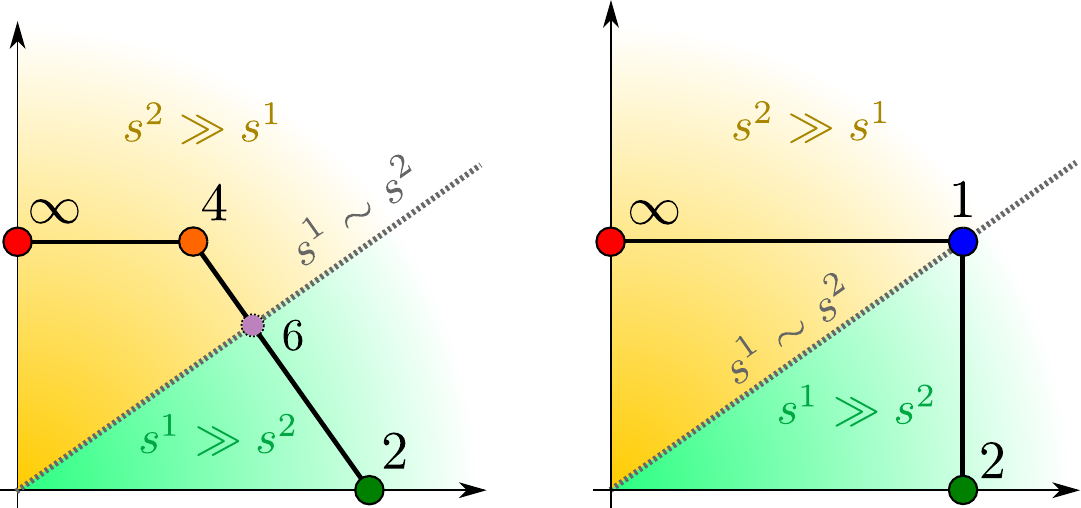}
\caption{\hspace{-0.3em} $P(s)=3\kappa_{112}(s^1)^2s^2$} \label{f.glue-1}
\end{subfigure}\hfill
\begin{subfigure}[b]{0.49\textwidth}
\center
\includegraphics[width=\textwidth]{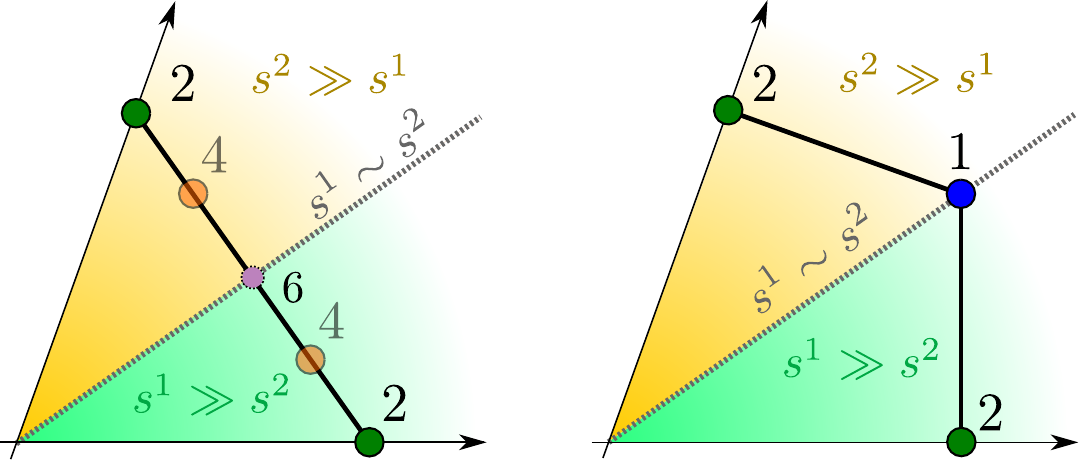}
\caption{\hspace{-0.33em}  $P(s)=3\kappa_{112}(s^1)^2s^2+3\kappa_{122}s^1(s^2)^2$} \label{f.glue-2}
\end{subfigure}
\begin{subfigure}[b]{0.49\textwidth}
\center
\includegraphics[width=\textwidth]{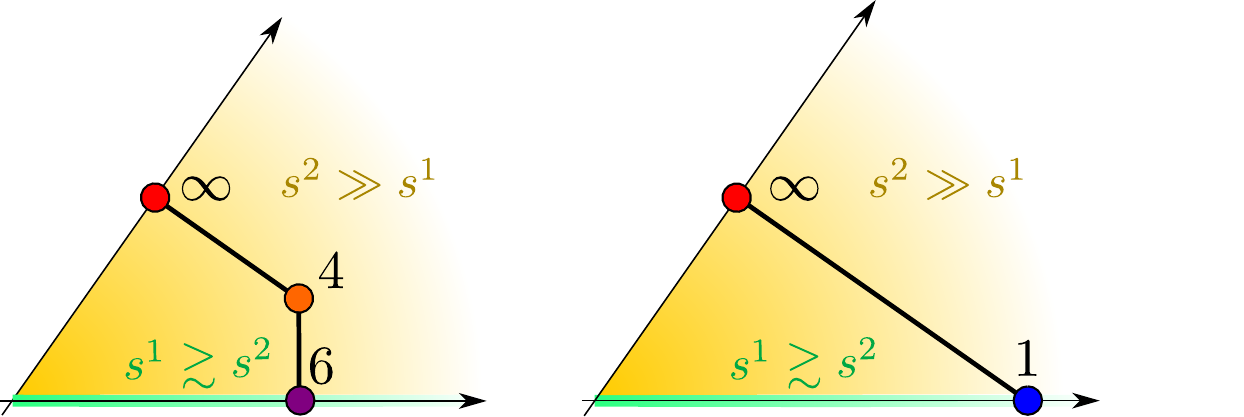}
\caption{\hspace{-0.33em}  $P(s)=\kappa_{111}(s^1)^3+3\kappa_{112}(s^1)^2s^2$} \label{f.glue-3}
\end{subfigure}
\begin{subfigure}[b]{0.49\textwidth}
\captionsetup{width=\linewidth}
\center
\includegraphics[width=\textwidth]{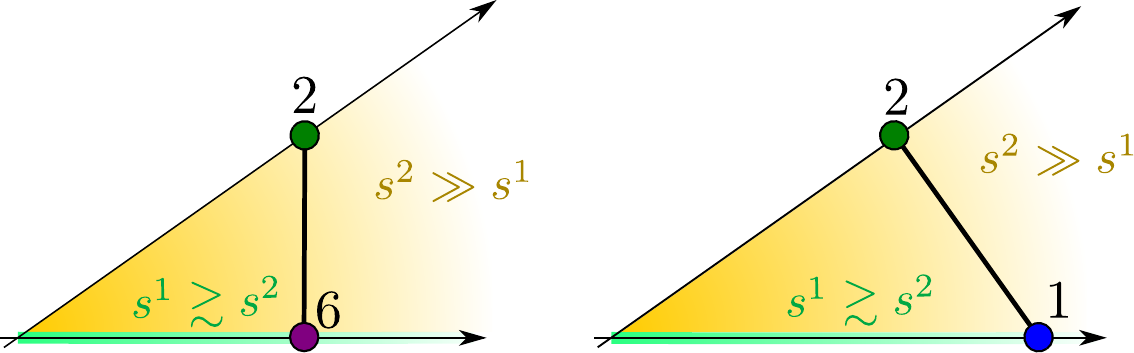}
\caption{\hspace{-0.3em} $P(s)\kappa_{111}(s^1)^3+3\kappa_{112}(s^1)^2s^2+3\kappa_{122}s^1(s^2)^2$} \label{f.glue-4}
\end{subfigure}
\begin{subfigure}[b]{0.49\textwidth}
\center
\includegraphics[width=\textwidth]{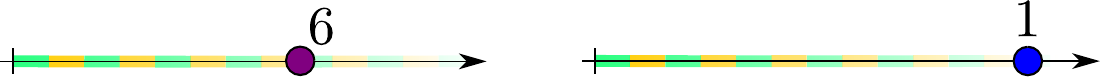}
\captionsetup{width=1.5\linewidth}\caption{\hspace{-0.3em} $P(s)=\kappa_{111}(s^1)^3+3\kappa_{112}(s^1)^2s^2+3\kappa_{122}s^1(s^2)^2+\kappa_{222}(s^2)^3$} \label{f.glue-5}
\end{subfigure}
\caption{
Different consistent arrangements after ``gluing'' together the different growth sectors for a K\"ahler potential $K\sim-\log P(s)$, with $P(s)=\kappa_{ijk}s^is^js^k$ a degree-3 homogeneous polynomial on $s^1$ and $s^2$. Note that for the first option in Figure \ref{f.glue-2} the $\zeta$-vectors of the KK-4 towers point out of their respective growth sectors, and thus are never leading towers. They are needed, however, in order to obtain the bounded states of the KK-6 towers, sliding as we move from one growth sector to the other. Note that for Figure \ref{f.glue-5} the moduli space is effectively 1-dimensional in the asymptotic regions, with all trajectories being parallel.
\label{f.glue}}
\end{center}
\end{figure}

Now, as one goes to $p=5$ the situation changes. From Table \ref{tab:1sax} we can see that $(s^1)^5$ does not admit consistent towers, while the only degree-5 monomial that does in Figure \ref{fig.growthGLUE}, $(s^1)^4s^2$, does not admit towers along the interface between growth sectors. While in principle polynomials such as $P(s^1,s^2)=a_1s^1(s^2)^4+{\color{gray}a_2(s^1)^2(s^2)^3+a_3(s^1)^3(s^2)^2}+a_4(s^1)^4s^2$ can result in a positive-definite metric, since there is no consistent gluing between the growth sectors there it would not be admissible for our assumptions. This results in 

\begin{mdframed}
\centering
The only consistent K\"ahler potentials $K=-\log P(s^1,s^2)$ with $\deg P(s^1,s^2)=5$ are $K=-\log\big[a_4(s^1)^4s^2\big]$ and $K=-\log\big[a_1s^1(s^2)^4\big]$, no other polynomial with more than one term is allowed.
\end{mdframed}
 Similarly, we can see from Figure \ref{fig.growthGLUE} that $(s^1)^5s^2$ and $(s^1)^5(s^2)^2$ are not allowed as independent monomials for degrees $p=6$ and 7, respectively. This means that we cannot choose them as leading terms in either the $s^1\gg s^2$ or $s^2\gg s^1$ growth sectors. All the other degree-6 (degree-7) monomials that are allowed individually can be glued together (provided the moduli space metric is positive definite) into a consistent polynomial, since they all have a KK-4 (KK-7) tower on the interface between growth sectors. Thus, we conclude that:
\begin{mdframed}
\centering
There is no admissible K\"ahler potential $K=-\log P(s^1,s^2)$ with $\deg P(s^1,s^2)=6$ and $(s^1)^5s^2$ or $s^1(s^2)^5$ as leading terms in $\mathcal{R}_{12}$ or $\mathcal{R}_{21}$. All other combinations with positive definite metric are allowed.
\end{mdframed}
\begin{mdframed}
\centering
There is no admissible K\"ahler potential $K=-\log P(s^1,s^2)$ with $\deg P(s^1,s^2)=7$ and $(s^1)^5(s^2)^2$ or $(s^1)^2(s^2)^5$ as leading terms in $\mathcal{R}_{12}$ or $\mathcal{R}_{21}$. All other combinations with positive definite metric are allowed.
\end{mdframed}
Note that the above excludes seemingly fine K\"ahler potentials such as $P(s^1,s^2)=a_6 (s^1)^6+a_5(s^1)^5s^2$ \emph{or} $P(s^1,s^2)=a_6(s^1)^6s^2+\dots+ a_2(s^1)^2(s^2)^5$.\\

Since they never become leading, our approach does not put constrains on the intermediate monomials between the leading ones. In general they might have to be turned on in order to have a positive definite metric. Luckily for us, such problem has received great attention on the Mathematics literature. A sufficient condition for  $\mathsf{G}_{jk}=-\frac{1}{2}\partial_{s^j}\partial_{s^k}\log P(s)$ to be positive definite is $P(s)$ being a \emph{Lorenztian polynomial}, i.e., its Hessian having signature $(+,-,\dots,-)$ (for two saxions $(+,-)$). A given homogeneous polynomial on $s^1$ and $s^2$ of degree $p$ as in \eqref{eq. hom pol} is Lorentzian iff the following conditions are met \cite[Theorem 2.25]{branden2020lorentzian}
\begin{enumerate}
	\item The sequence $(a_0,\dots,a_p)$ is non-negative.
	\item The sequence $(a_0,\dots,a_p)$ is \emph{ultra-$log$ concave}, i.e.,
	\begin{equation}\label{eq.ultralog}
		\frac{a_k^2}{\binom{p}{k}^2}\geq \frac{a_{k-1}}{\binom{p}{k-1}}\frac{a_{k+1}}{\binom{p}{k+1}}\;\Leftrightarrow\; a_k^2\geq \frac{(k+1)(p-k+1)}{k(p-k)}a_{k-1}a_{k+1}\;\forall\; k=1,\dots p-1\,.
	\end{equation}	 
	\item The sequence $(a_0,\dots,a_p)$ has no internal zeroes, i.e.,
	\begin{equation}
		a_{i}a_j>0\Longrightarrow a_k>0\;\forall\;	0\leq i<k<j\leq p\,.
	\end{equation}
\end{enumerate}
Note that the second condition only refers to the numerical value of the polynomial coefficients, and is equivalent to condition \eqref{eq. ineq} for $p=3$, while the third forces intermediate monomials to appear once we know which ones dominate. We must stress that the above requirements are \emph{sufficient}, but not necessary. Take for example $P(s^1,s^2)=a_4(s^1)^4(s^2)^2+a_2(s^1)^2(s^2)^4$, which results in a positive definite metric for any $a_4,\, a_2>0$. However, in general the relations on the coefficients become quite involved, and there are no general results on them.\\

The generalization to more than two moduli is relatively straightforward. Modulo the polynomial $P(s^1,\dots, s^n)$ resulting in a positive-definite moduli space metric (see \cite{branden2020lorentzian} for the characterization of Lorentzian polynomials with more than 2 variables), only those whose leading terms have the same tower arrangement along the interface between growth sectors are allowed. Since now the interface will be a $(n-1)$-dimensional hyperplane, checking all the possible arrangements becomes more involved, which together with the number of possible monomials growing as $\binom{n+p-1}{p}$, makes the classification doable but computationally involved. We will thus not consider such task in this paper. In any case, the above examples already serve as a proof of principle that a consistent gluing of the towers between different growth sectors can rule out certain seemingly fine K\"ahler potentials.
\subsection{Interplay with taxonomy rules\label{ss.taxonomy}}

In the previous sections, we made use of the assumptions laid out in Section \ref{sec.reconstructionalg} involving the scaling of light towers of states and EFT strings to reconstruct the consistent arrangements of towers/duality frames in each growth sector. At the same time, the reconstruction algorithm we laid out includes in its assumptions the ESC imposed both on individual towers (imposing their exponential rates) and on their bound states. It is then natural to ask what is the interplay between the arrangements of towers obtained in this paper and those obtained in  \cite{Etheredge:2024tok} using the taxonomy rules for light towers, as the latter were also derived from assuming the ESC.

In \cite{Etheredge:2024tok}, it was shown that, assuming that the Emergent String Conjecture \cite{Lee:2019wij} holds recursively upon decompactification, the following relations between $\zeta$-vectors can be inferred:\footnote{Similarly as in our approach, in \cite{Etheredge:2024tok} it was assumed that decompactification is to a Minkowski vacuum, in such a way that possible internal warping in the internal manifold gets diluted as the internal space grows. When this is not the case, and decompactification occurs to a warped spacetime, the exponential rates at which towers become light change, see e.g., \cite{Etheredge:2023odp}. We refer to \cite{Raucci:2026fzp} for a generalization of the Taxonomy rules when this is taken into account.}
\begin{equation}\label{eq. taxonomy}
\vec{\zeta}_{{\rm KK},n}\cdot\vec\zeta_{{\rm KK}',m}=\frac{1}{d-2}+\frac1n\delta_{nm}=\frac{1}{2}+\frac1n\delta_{nm}\quad\text{for }\,d=4\,,
\end{equation}
where $n$ and $m$ are the numbers of decompactifying dimensions for each light tower, and $n=\infty$ formally corresponds to the case of string oscillator modes. These were denoted as \emph{taxonomy rules}. Moreover, assuming that the $\zeta$-vectors are globally defined along the asymptotic regions of moduli space and their expression is fixed (so that they do not ``slide'' as we change the asymptotic trajectory), and that our moduli space admits some flat slice, it is possible to use the above taxonomy rules to obtain a finite set of solutions describing possible global arrangements of light towers/duality frames, similar to what we did in previous sections by imposing the Integral Scaling Relation to hold globally. 

In the following, we will give an answer to the following questions: are the taxonomy rules \eqref{eq. taxonomy}  and the scaling relations \eqref{eq.intscaling} equivalent in 4d $\mathcal{N}=1$ theories? If this is not the case, does one of them imply the other? As we will explain now, the answer to both questions is negative. As the reconstruction algorithm in section \ref{sec.reconstructionalg} involves both the Integral Scaling Relation and taxonomy rules, it will impose stronger constraints on the allowed arrangements than taking them separately.

As explained in \cite{Etheredge:2024tok}, while there is a single global solution covering the complete $\mathbb{R}^n$ for theories in $d=11-n$ dimensions, precisely corresponding to the tower/duality arrangement of M-theory on $\mathbb{T}^n$, there are more possibilities when the number of non-compact moduli is less than $n<11-d$, or the complete $\mathbb{R}^n$ is not covered, but only some (convex cone) subsector of it, with the number of possible arrangements growing fast as we lower $d$. Already in \cite{Etheredge:2023odp} it was shown how some tower/duality arrangements that were solutions of the taxonomy rules \eqref{eq. taxonomy} do not correspond to any known string/M-theory embedding. \\

\begin{figure}[hbt]
\begin{center}
\includegraphics[width=\textwidth]{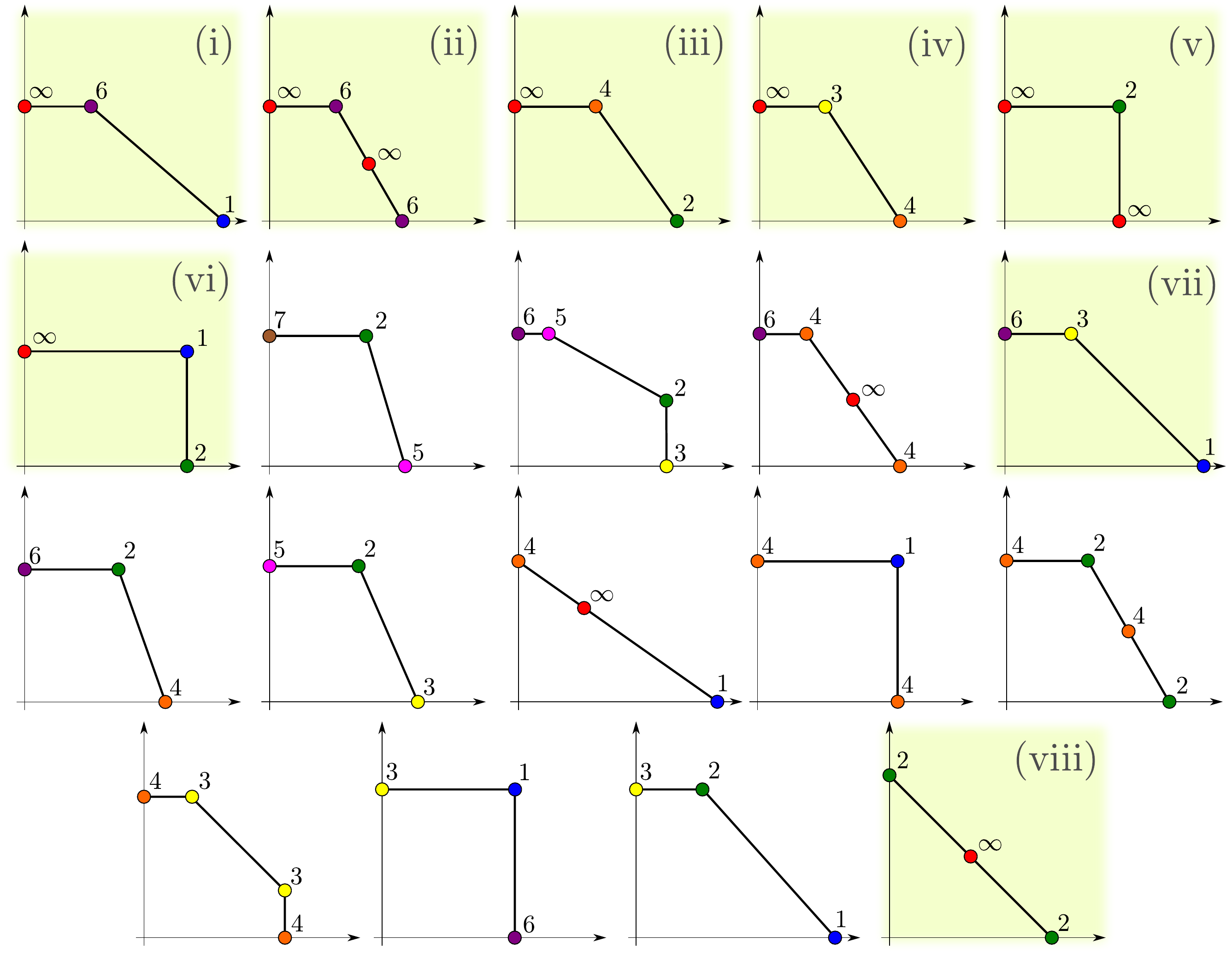}
\caption{Possible solutions to the taxonomy rules filling the first quadrant with two moduli. The number of decompactifying dimensions for every KK tower accompanies each vertex, with $\infty$ corresponding to a (non-necessarily supersymmetric) emergent string limit. In yellow those fulfilling the Integral Scaling Relation with $K\sim -\log P(s^1,s^2)$ and $\deg P(s^1,s^2)\leq 7$ are highlighted.
\label{f.bu2mod}}
\end{center}
\end{figure}

\begin{figure}[hbt]
\begin{center}
\includegraphics[width=\textwidth]{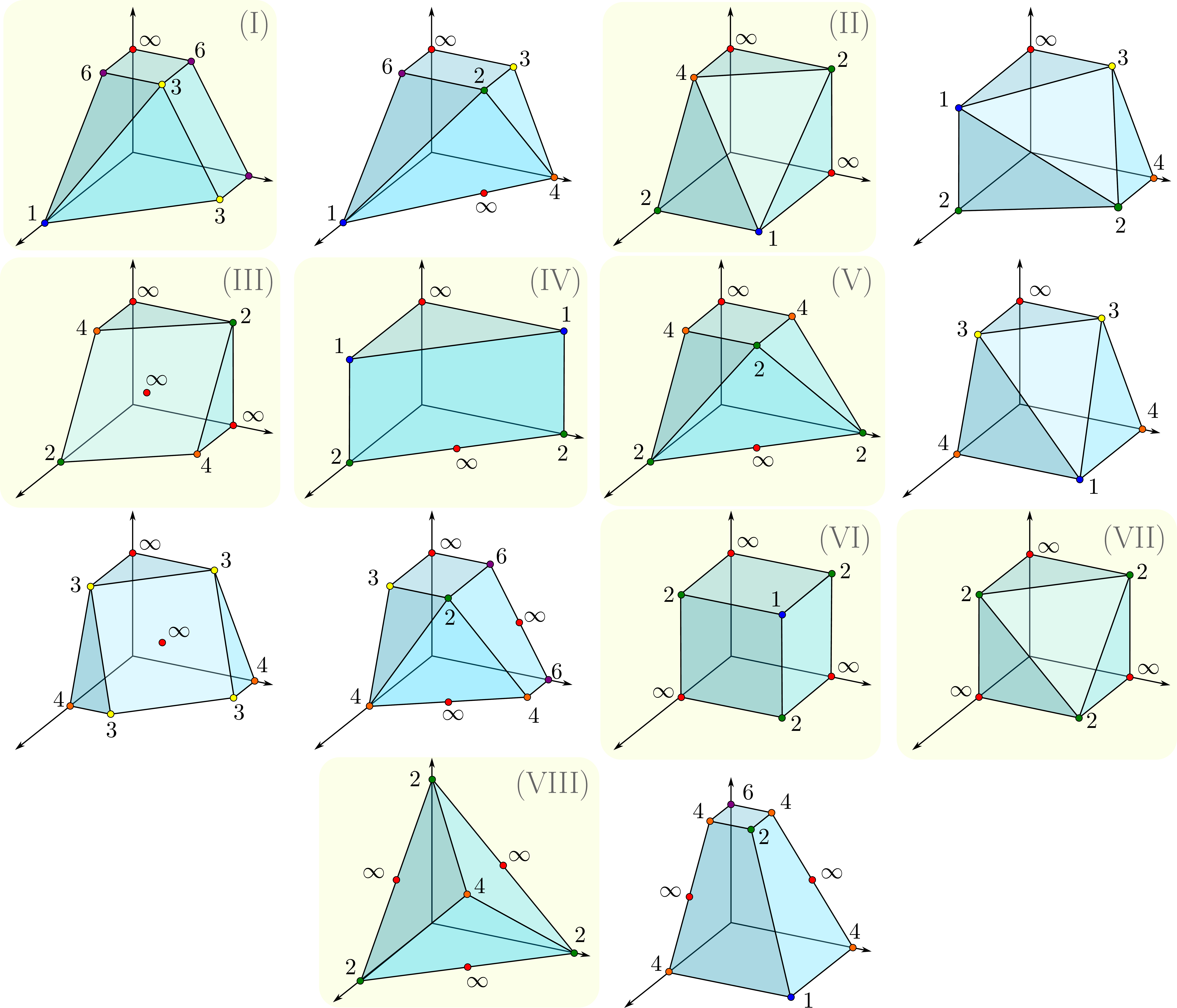}
\caption{Possible solutions to the taxonomy rules filling the first octant with three moduli. The number of decompactifying dimensions for every KK tower accompanies each vertex, with $\infty$ corresponding to a (non-necessarily supersymmetric) emergent string limit. In yellow those fulfilling the Integral Scaling Relation with $K\sim -\log P(s^1,s^2,s^3)$ and $\deg P(s^1,s^2,s^3)\leq 7$ are highlighted.
\label{f.bu3mod}}
\end{center}
\end{figure}

As an illustration, in Figures \ref{f.bu2mod} and \ref{f.bu3mod} we respectively depict the solutions to \eqref{eq. taxonomy} in $d=4$ spacetime dimensions, respectively for 2 and 3 moduli, with the asymptotic directions respectively corresponding to a quadrant/octant of $\mathbb{R}^2$ and $\mathbb{R}^3$, understood as the saxionic cone. We see that there are 19 solution for two moduli and 14 for three. For such solutions, in both cases, only 8 arrangements in both cases are also solutions of the Integral Scaling Relation \eqref{eq.intscaling} for some K\"ahler potential $K\sim -\log P(s^1,s^2)$ and $-\log P(s^1,s^2,s^3)$, with $P({s})$ a monomial of at most degree 7. In practice, this means that only a subset of the convex hulls generated through taxonomy rules are compatible with embedding in a lattice generated by EFT strings. In Tables \ref{tab.bu2mod} and \ref{tab.bu3mod} we show which $P({s})$  monomials correspond to which taxonomy solutions from Figures \ref{f.bu2mod} and \ref{f.bu3mod}.

One could think that, had we allowed $P(s)$ to have arbitrarily high degree and the Integral Scaling Relation to hold for any $w\in\mathbb{Z}_{\geq 0}$, then the lattice \eqref{eq.lambda K} would be very dense, in such a way that eventually all solutions to the taxonomy rules would be consistent with such integral scaling. This turns out to be the case for every convex hull in Figures \ref{f.bu2mod} and \ref{f.bu3mod}. In fact, the Integral Scaling Relation \eqref{eq.intscaling} imposes a Diophantine equation on each lattice point $\vec{\zeta}_{{\rm KK},n}$ (including $\vec{\zeta_{\rm osc}}$ for $n\to\infty$), 
\begin{equation}\label{eq.int.moremod}
\sum_{i=1}^r \frac{w_i^2}{2p_i} = \frac{n+2}{2n}, \ n=1,\dots, 7 \ \text{and} \ \infty,
\end{equation}
which is a generalization of \eqref{eq.int-1mod} to a $r$-dimensional moduli space.\footnote{ The fact that equation \eqref{eq.int.moremod} always has integer solutions has to be checked case by case for $r=2,3$. For $r\geq 4$, the \textit{Hasse–Minkowski theorem} (see, e.g., \cite[Section 3, Theorem 8]{Serre1973ACI}) guarantees that there will always be an integer solution.} Particularly, looking at Table \ref{tab:pw}, we see the same thing happening here: there will always be classes of solutions parametrized by $k_i, \ i=1,\dots,r$, but it is enough to chose $k_i=1 \ \forall i$ as the representative of each class and there will always be a suitable solution to \eqref{eq.int.moremod} for at least one KK rate.

\begin{table}[hbt]
    \centering
\begin{tabular}{|c||c|c|c|c|c|c|c|c|}
\hline 
$K\sim -\log(s^1)^a(s^2)^b$ & (i) & (ii) & (iii) & (iv) & (v) & (vi) & (vii) & (viii) \\ 
\hline \hline
(1,1) &   &  &   &   & \checkmark &   &   &   \\ 
\hline 
(1,2) &   &  & \checkmark &  &  & \checkmark &  &  \\ 
\hline 
(1,3) & \checkmark & \checkmark &  &  &  &  &  &  \\ 
\hline 
(1,4) &  &  &  &  & \checkmark &  &  &  \\ 
\hline \rowcolor{swamp}
(1,5) &  &  &  &  &  &  &  &  \\ 
\hline 
(1,6) &  &  &  & \checkmark &  &  &  &  \\ 
\hline 
(2,2) &  &  &  &  &  & \checkmark &  & \checkmark \\ 
\hline  \rowcolor{swamp}
(2,3) &  &  &  &  &  &  &  &  \\ 
\hline 
(2,4) &  &  & \checkmark &  &  &  &  &  \\ 
\hline \rowcolor{swamp}
(2,5) &  &  &  &  &  &  &  &  \\ 
\hline
(3,3) &  &  &  &  &  &  & \checkmark &  \\ 
\hline 
(3,4) & \checkmark & \checkmark &  &  &  &  &  &  \\ 
\hline 
\end{tabular} 
\caption{Solution to the taxonomy rules \eqref{eq. taxonomy} for the first quadrant parametrized by two flat scalars $s^1$ and $s^2$ that are also consistent for the scaling relation \eqref{e.bottom up integer} with $K\sim -\log(s^1)^a(s^2)^b$ of at most degree 7. To each solution (i) to (viii) we associate the K\"ahler potential realizing such arrangement. Those K\"ahler potentials for which no solution is consistent with taxonomy is depicted in red. \label{tab.bu2mod}}
\end{table}

\begin{table}[hbt]
    \centering
\begin{tabular}{|c||c|c|c|c|c|c|c|c|}
\hline 
$K\sim -\log(s^1)^a(s^2)^b(s^3)^c$ & (I) & (II) & (III) & (IV) & (V) & (VI) & (VII) & (VIII) \\ 
\hline \hline
(1,1,1) &   &  &   &   &  & \checkmark  &\checkmark   &   \\ 
\hline 
(1,1,2) &   &\checkmark  & \checkmark &  &  &  &  &  \\ 
\hline \rowcolor{swamp}
(1,1,3) &  &  &  &  &  &  &  &  \\ 
\hline 
(1,1,4) &  &  &  &  &  & \checkmark & \checkmark &  \\ 
\hline \rowcolor{swamp}
(1,1,5) &  &  &  &  &  &  &  &  \\ 
\hline 
(1,2,2) &  &  &  &  \checkmark &  \checkmark &  &  &  \\ 
\hline \rowcolor{swamp}
(1,2,3) &  &  &  &  &  &  &  &  \\ 
\hline 
(1,2,4) &  & \checkmark & \checkmark &  &  &  &  &  \\ 
\hline 
(1,3,3) & \checkmark &  &  &  &  &  &  &  \\ 
\hline 
(2,2,2) &  &  &  &  &  &  &  &  \checkmark \\ 
\hline \rowcolor{swamp}
(2,2,3) &  &  &  &  &  &  &  &  \\ 
\hline 
\end{tabular} 
\caption{Solution to the taxonomy rules \eqref{eq. taxonomy} for the first octant parametrized by three flat scalars $\{s^1,s^2,s^3\}$ that are also consistent for the scaling relation \eqref{e.bottom up integer} with $K\sim -\log(s^1)^a(s^2)^b(s^3)^c$ of at most degree 7. To each solution (I) to (VIII) we associate the K\"ahler potential realizing such arrangement. Those K\"ahler potentials for which no solution is consistent with taxonomy is depicted in red. \label{tab.bu3mod}}
\end{table}

In summary, around half of the solutions to the taxonomy rules do not admit a (sensible) choice of K\"ahler potential for which the tower arrangement and the EFT strings satisfy the Integral Scaling Relation if the degree of the polynomial is bounded by $\deg P(s)\leq 7$.

On the other hand, as previously discussed in Section \ref{sec.constrain kahler}, there are some seemingly fine choices of $K\sim -\log P(s)$ (as in they result in a positive-definite moduli space metric) for which the lattice \eqref{eq.lambda K} does not define convex hulls consistent the ESC and the Integral Scaling Relation (and thus with the taxonomy rules \eqref{eq. taxonomy}). In other words, there are seemingly consistent tower arrangements that do not admit a K\"ahler potential and vice versa.

We finish by noticing that, even if enough individual points of the lattice \eqref{eq.lambda K} might be consistent to define a global convex hull, this does not mean that the global arrangement itself is consistent, as we illustrate in Figure \ref{fig.FISHY}. Both the (IV) and (V) arrangements in Figure \ref{f.bu3mod} have a $\mathbb{Z}_2$ symmetry, associated to the switching of two saxionic directions. The fixed plane under such symmetries contains the same arrangements of towers $\vec{\zeta}_{\rm osc}$, $\vec{\zeta}_{\rm KK,2}$ and another $\vec{\zeta}_{\rm osc}$, with the only difference that $\vec{\zeta}_{\rm KK,2}$ is a generating tower in (V), while in (IV) is the bounded states of two $\vec{\zeta}_{\rm KK,1}$ towers. Furthermore, as seen from Table \ref{tab.bu3mod}, both arrangements are consistent with a K\"ahler potential of the form $K\sim-\log s^a(s^b)^2(s^c)^2$. This means that towers from the arrangement depicted in  Figure \ref{fig.FISHY}, constructed by gluing halves of (IV) and (V) along their common fixed plane, are consistent with the Integral Scaling Relation \eqref{eq.lattice}, and the individual lengths of the towers satisfy the expected lengths from \eqref{eq.length}. However, the global arrangement is inconsistent with the recursive application of the ESC (see assumptions in Section \ref{sec.reconstructionalg}), and thus with the taxonomy rules \eqref{eq. taxonomy}, since the middle $\vec{\zeta}_{\rm KK,2}$ and $\vec{\zeta}_{\rm KK,1}$ have
\begin{equation}
	\vec{\zeta}_{\rm KK,2}\cdot\vec{\zeta}_{\rm KK,1}=\left(\frac{1}{\sqrt{2}},\frac{1}{2},\frac{1}{2}\right)\cdot\left(\frac{1}{\sqrt{2}},0,1\right)=1\neq\ \frac{1}{2}
\end{equation}
The explanation behind this is simple. In arrangement (IV), the $\vec{\zeta}_{\rm KK,2}$ vector is a bounded state between the two $\vec{\zeta}_{\rm KK,1}$, and thus it is not a generating tower, and does not follow the taxonomy rules with other vectors. When building the lattice \eqref{eq.lattice} from the K\"ahler potential, if only the individual points are required to follow the taxonomy rules by having the appropriate lengths, then we might be actually selecting points that are not globally consistent with the surrounding towers, just as in the example we have just discussed. This is why it is vital that the step \textbf{(e.)} in the reconstruction algorithm from Section \ref{sec.reconstructionalg} is required, as argued by requiring that the Emergent String Conjecture holds recursively upon decompactification, rather than simply using the Integral Scaling Relation and the expected length of the $\zeta$-vectors.

\begin{figure}[H]
\begin{center}
\includegraphics[width=\textwidth]{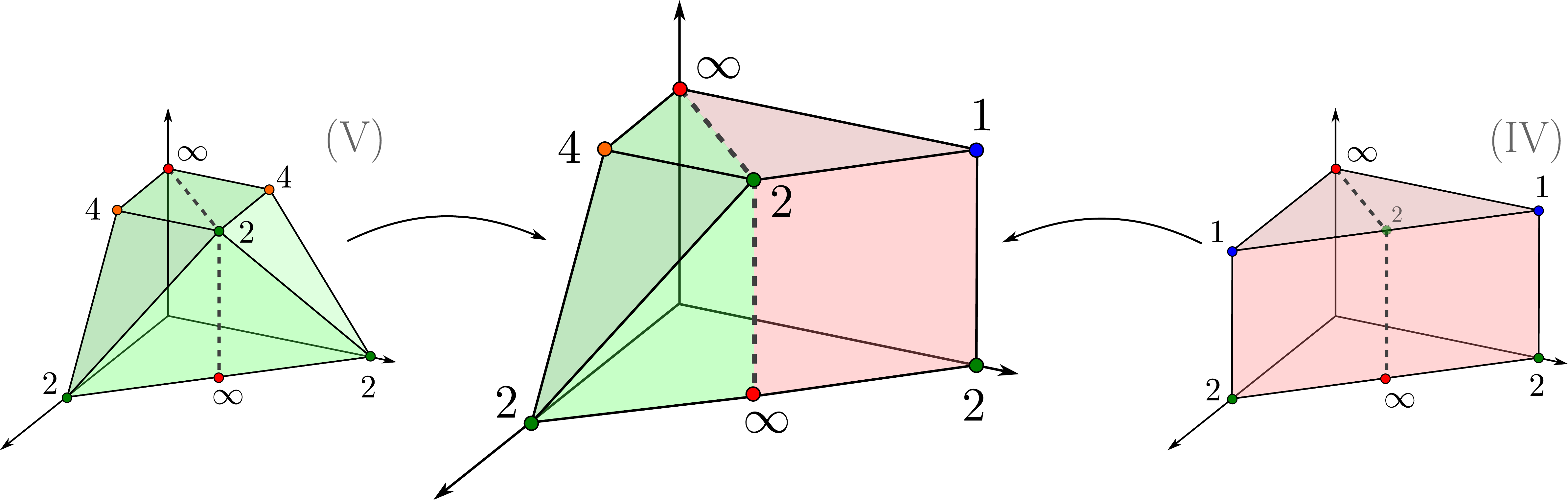}
\caption{{Sketch of how one can build an inconsistent tower arrangement (center) that respects the Integral Scaling Relation and length of the $\zeta$-vectors from two consistent arrangements (left and right), as the product between the middle $\vec{\zeta}_{\rm KK,2}$ and $\vec{\zeta}_{\rm KK,1}$ (both generating towers) is not $\frac{1}{d-2}=\frac{1}{2}$.}
\label{fig.FISHY}}
\end{center}
\end{figure}

\subsection{Evidence for String Universality\label{s.String UNI}}

In Section \ref{sec.mono}, we listed all the monomial terms that could play the role of the leading term in the polynomial setting the K\"ahler potential $K=-\log P(s)$ in a given growth sector of the moduli space of a generic 4d $\mathcal{N}=1$ EFT. 
In particular, we ran all possible monomials of degree less or equal than 7 through the reconstruction algorithm, yielding the following two insights. First, we obtained the subset of monomials compatible with the assumptions of the Integral Scaling Relation and ESC laid out in Section \ref{sec.reconstructionalg}. Secondly, we derived, for each monomial in this subset, the possible arrangement(s) of towers in terms of a convex hull of $\zeta$-vectors, with at most two compatible convex hulls which respect all our assumptions for each monomial. These are listed in Tables \ref{tab:1sax} and \ref{tab:bu2}, and \ref{tab:bu3} to \ref{tab:bu7}.

In Section \ref{sec.mono} we anticipated that the 7-dimensional polytope in Table \ref{tab:bu7} is the only tower arrangement allowed by our assumptions for the choice of monomial. At the same time, it is precisely the tower arrangement arising from M-theory compactified on a $G_2$ manifold. A similar coincidence happens for 4-dimensional (and lower) monomials of degree 4 for which the compatible polytopes are described microscopically by type IIA/heterotic $E_8 \times E_8$ or type IIB/type I/heterotic $SO(32)$. One could then wonder if we can understand all the allowed polytopes in the tables of Section \ref{sec.mono} through a similar lens. Is any polytope compatible with the bottom-up assumptions of the reconstruction algorithm  realized microscopically by a precise top-down construction? The answer to this question is yes, as we will explain in the following. As a matter of fact, we find a complete match between allowed slices from a top-down point of view, by slicing the 7-dimensional polytope of M-theory on $G_2$ and polytopes reconstructed by the bottom-up algorithm.

In general,  the polytopes listed in Tables \ref{tab:1sax} to \ref{tab:bu6} will have to be interpreted as slices of bigger moduli spaces. An easy and illustrative example is given by the polytope in Table \ref{tab:bu6}, which can be seen as slices of the 7-dimensional one in Table \ref{tab:bu7}. The EFT strings and KK-2 6-dimensional $\zeta$-vectors come straightforwardly by choosing, from the correspondent EFT strings and KK-2 7-dimensional counterparts, only the ones with one component (say the 7$^{\rm th}$ one) equal to zero. In synthesis, 

\begin{equation}
\begin{split}
\text{EFT strings:} \ & \underbrace{\big(\underline{\tfrac{1}{\sqrt{2}},0,0,0,0,0,0}\big)}_{7d} \longrightarrow \underbrace{\big(\underline{\tfrac{1}{\sqrt{2}},0,0,0,0,0}\big)}_{6d},\\
\text{KK-2:} \ &
\underbrace{\big(\underline{\tfrac{1}{\sqrt{2}},\tfrac{1}{\sqrt{2}},0,0,0,0,0}\big)}_{7d}  \rightarrow \underbrace{\big(\underline{\tfrac{1}{\sqrt{2}},\tfrac{1}{\sqrt{2}},0,0,0,0}\big)}_{6d}.
\end{split}
\end{equation}

Regarding the KK-1 vectors, the exact same thing can be done by fixing one coordinate to always be zero, but in this case only 4 (in green) or the 7 $\zeta$-vectors survive the projection:

\begin{small}
\begin{equation}
 \underbrace{\begin{aligned}
\color{Green}{\big(\tfrac{1}{\sqrt{2}},\tfrac{1}{\sqrt{2}},\tfrac{1}{\sqrt{2}},0,0,0,0\big)},& \color{Green}{\big(\tfrac{1}{\sqrt{2}},0,0,\tfrac{1}{\sqrt{2}},\tfrac{1}{\sqrt{2}},0,0\big)},\\ \big(\tfrac{1}{\sqrt{2}},0,0,0,0,\tfrac{1}{\sqrt{2}},\tfrac{1}{\sqrt{2}}\big), & \color{Green}{\big(0,\tfrac{1}{\sqrt{2}},0,\tfrac{1}{\sqrt{2}},0,\tfrac{1}{\sqrt{2}},0\big)}, \\ \big(0,\tfrac{1}{\sqrt{2}},0,0,\tfrac{1}{\sqrt{2}},0,\tfrac{1}{\sqrt{2}}\big),& \color{Green}{\big(0,0,\tfrac{1}{\sqrt{2}},0,\tfrac{1}{\sqrt{2}},\tfrac{1}{\sqrt{2}},0\big)},  \\ \big(0,0,\tfrac{1}{\sqrt{2}},&\tfrac{1}{\sqrt{2}},0,0,\tfrac{1}{\sqrt{2}}\big)
 \end{aligned}}_{7d}
 \ \longrightarrow \
 \underbrace{\begin{aligned}
 \big(\tfrac{1}{\sqrt{2}},\tfrac{1}{\sqrt{2}},\tfrac{1}{\sqrt{2}},0,0,0\big), \big(\tfrac{1}{\sqrt{2}},0,0,\tfrac{1}{\sqrt{2}},\tfrac{1}{\sqrt{2}},0\big), \\ \big(0,\tfrac{1}{\sqrt{2}},0,\tfrac{1}{\sqrt{2}},0,\tfrac{1}{\sqrt{2}}\big), \big(0,0,\tfrac{1}{\sqrt{2}},0,\tfrac{1}{\sqrt{2}},\tfrac{1}{\sqrt{2}}\big).
 \end{aligned}}_{6d}
\end{equation}
\end{small}

An identical argument yields the 4 non-BPS strings in the 6-dimensional case descending directly from the 7-dimensional ones.
 
From the K\"ahler potential point of view, this operation is equivalent to keep one of the saxions fixed at finite value $s^i\sim s^i_0$. In principle, one could also move along a direction where two saxions are bound to grow at the same rate and the others are free,  $s^i \sim s^j$ which would yield $s^1s^2s^3s^4s^5s^6s^7 \to (s^1)^2s^2s^3s^4s^5s^6$ if, e.g., $s^1\sim s^7$. In Table \ref{tab:bu6}, the latter is highlighted in red because the reconstruction algorithm does not give a consistent convex hull for the towers, from the bottom-up point of view. We can also understand this from the point of view of the structure of the 7d Joyce manifold on which we compactify M-theory: this will give us the understanding of what is exactly a ``good" slice and what is not. 

Consider Joyce compact manifolds of the kind described in \cite{joyce1996a,joyce1996b,joyce2004constructing}  and summarized in Appendix \ref{app.topdown}: they consist in a toroidal orbifold of the kind $X_7=\mathbb{T}^7/(\mathbb{Z}_2\oplus \mathbb{Z}_2 \oplus \mathbb{Z}_2)$ with torus radii $\{R_a\}_{a=1}^7$ in 11d units. As illustrated in \eqref{eq.radii saxion joyce}, the relation between the seven untwisted saxions and the radii is of the kind $s^i= R_{a_I}R_{a_J}R_{a_K}$. Due to the $G_2$ holonomy, saxions measure volumes of 3-cycles, and as a result, when we move in moduli space with one saxionic direction faster than the others, we will always be scaling suitably the three radii on which the saxion depends. Another way to see this is that we cannot grow one single radius without it affecting at least three saxionic directions. This interlinking between saxions and multiple radii is exactly what forbids the entries in red in the tables of section \ref{sec.mono} from being ``good" slices of the M-theory polytope in Table \ref{tab:bu7}. 
To see this, we will now explain what goes wrong in the slicing $s^1s^2s^3s^4s^5s^6s^7 \to (s^1)^2s^2s^3s^4s^5s^6$ and at the same time allows $s^1s^2s^3s^4s^5s^6s^7 \to (s^1)^2(s^2)^2s^3s^4s^5$.
Following the expressions for the KK masses given in \cite{Grieco:2025bjy}, a KK tower associated to the decompactification of one radius $R_2$ in M-theory on $G_2$ is given by
\begin{equation}\label{eq.m_R2}
m_{{\rm KK},R_2} = \frac{M_{\rm Pl,11}}{R_2} = \frac{M_{\rm Pl,4}}{R_2\sqrt{V_X}} = \frac{M_{\rm Pl,4}}{\sqrt{s^1s^4s^5}},
\end{equation}
where $V_X = R_1\dots R_7 = (s^1\dots s^7)^{1/3}$, see Appendix \ref{app.Joyce}. Now, we take a direction in the 7-dimensional moduli space by identifying $s^1\sim s^2$: because of their dependence on the radii, having in common exactly $R_1$, this imposes $R_2R_3=R_4R_5$ as a constraint.

When we project in this direction, the exponential rate that we see for $m_{KK,R_2}$ is not that of a $n=1$ tower or of any tower with integer $n$: the projected length of its $\zeta$-vector is $|\zeta|^2_{R_2}=5/4$, which would correspond to $n= 4/3$. The situation is the same for $m_{KK,R_3}$, $m_{KK;R_4}$ and $m_{KK,R_5}$: any choice of identification of a single couple $s^i \sim s^j$ (they are 21 in total) leads to the projection of exactly four KK-1 towers to be incompatible with ESC rates. This is why the slicing $s^1s^2s^3s^4s^5s^6s^7 \to (s^1)^2s^2s^3s^4s^5s^6$ is not possible: without a ``good" projection of these towers, i.e. a projection landing on ESC rates, the polytope that we obtain will inevitably have $k$-simplices which are at the wrong distance from the origin, violating the ESC recursion condition. 

Instead, the identification of \textit{two} couples of saxions, with a suitable choice of the saxions identified, makes it possible for a ``good" slice to be taken for $s^1s^2s^3s^4s^5s^6s^7 \to (s^1)^2(s^2)^2s^3s^4s^5$. If we choose to identify $s^1\sim s^2$, for example then a choice for the other pair yielding a ``good" slice is $s^5\sim s^6$. This second identification puts further constraints on the radii growth, so that the combination of the constraints results in imposing $R_2=R_5$ and $R_3=R_4$. Linking these pairs of radii means that we are decompactifying two radii at the same rate: the projection of the $m_{{\rm KK},R_2}$ tower on this slice has exponential rate compatible with a KK-2 decompactification. The remaining original KK-1 towers will either behave in the same way (projecting onto the KK-2 vector) or be completely embedded on the slice (keeping their original KK-1 rate). Note that had we taken  $s^1\sim s^2$ together with any other identification outside of $s^5\sim s^6$ or $s^4\sim s^7$ the slice would have been ``skew" and not ``good", so that not only is the number of identifications important but also the exact pairs that we are taking.  

A \textit{good} slice is then defined as a slice of the higher-dimensional moduli space where the corresponding tower polytope is generated by leading towers or bounded states contained in such slice. This guarantees that the new arrangement is consistent with the taxonomy rules, and therefore their length is that expected for a KK or string oscillator tower.

 Of course, one can always move along \emph{not good} or ``skew" slices that do not satisfy the above definition. However, they will always need the presence of towers pointing outside of them for their arrangement to satisfy the Integral Scaling Relation and ESC recursively, not yielding consistent standalone lower-dimensional tower polytopes. For this reason, this type of slices are not included in our classification.

This analysis can be repeated for any $r$-dimensional slicing of the 7-dimensional polytope, for each choice of $N$-fold identifications $s^{i_1} \sim s^{i_2} \sim \dots \sim s^N$ and/or freezing of moduli $s^{i}\sim s^i_0$. To repeat the message again, these directions are always allowed when picking a specific trajectory in a given moduli space. What might not always happen is that the $r$-dimensional slice yields a consistent standalone polytope without ``needing" towers spanning additional directions, compatible with the recursive ESC.

Let us then take the 7-dimensional M-theory polytope in Table \ref{tab:bu7} and  perform all the possible combinations of co-scaling and freezing of saxions. We obtain that, from this top-down point of view, the only combinations that yield ``good" slices are exactly the ones that are listed as allowed monomials in Tables \ref{tab:1sax} and \ref{tab:bu2} to \ref{tab:bu6}. Not only this, but the tower content of the ``good" slices is exactly the one listed in the tables, reproducing also the non-BPS strings. Therefore, it is enough to consider the slices of the polytope in Table \ref{tab:bu7} to recover all the allowed polytopes of lower dimension, which also happen to be the \textbf{only} \emph{good} slices that we can take of the 7-dimensional polytope. As we explained in Section \ref{sec.mono}, the 7-dimensional polytope is precisely the arrangement of towers for M-theory on the Joyce $G_2$-manifold $\mathbb{T}^7/(\mathbb{Z}_2\oplus\mathbb{Z}_2\oplus\mathbb{Z}_2)$ discussed in Appendix \ref{app.Joyce}.

This means that the results of the reconstruction algorithm, which is a bottom-up procedure, match \textit{exactly} the \emph{good} slices of this M-theory construction.  Not only this, but the polytopes obtained by top-down string compactifications (see Appendix \ref{app.topdown}) can also be obtained as slices of the M-theory one. In practice, it is enough to know the M-theory polytope to describe the tower arrangement of every possible growth sector of a 4d $\mathcal{N}=1$ moduli space. 
This is not trivial: there is no input in the assumptions of Section \ref{sec.reconstructionalg} that would evidently single out M-theory in particular; in principle it could have been possible for the reconstruction algorithm to yield polytopes not corresponding to slices of the 7-dimensional one. This suggests that, even if 4d $\mathcal{N}=1$ is not obtained by toroidal compactifications, the type of limits one obtains in a given growth sector of more complicated compactifications are a subset of those appearing for toroidal ones.

This seems to hint to some form of string universality regarding the boundaries of 4d $\mathcal{N}=1$ moduli spaces: the only tower arrangements satisfying the observed patterns described in Section \ref{sec.reconstructionalg} in relation to the Distance Conjecture are obtained through string/M-theory compactifications. 
To obtain these results, it has been fundamental to combine the ESC with the Integral Scaling Relation.
We want to stress that this alone does not provide a bottom-up rationale for the SDC nor it attempts to prove string universality in 4d $\mathcal{N}=1$ supergravity. It could very well be that a string lamppost effect is generated by one or more of the assumptions in \ref{sec.reconstructionalg}. For instance, the ESC in the algorithm can be regarded as a top-down input, since a complete bottom-up explanation for it is lacking (see \cite{Aoufia:2024awo,Bedroya:2024ubj,Basile:2023blg,Kaufmann:2024gqo} though).  
 If an interpretation of the former, possibly of the $w\leq3$ bound too, is found, this would corroborate our result as more genuine evidence for string universality.\\
Lastly, the freezing and co-scaling of $N$ moduli are very reminiscent of the way ``descendant polytopes'' are obtained by trees of algebraic embeddings of duality groups in \cite{Baines:2026aug}. In that paper it is shown that all consistent 4d arrangements in locally symmetric moduli spaces can be derived from subalgebras of the $\mathfrak{e}_{7(7)}$ duality algebra (precisely that of M-theory on $\mathbb{T}^7$), while in our case this seems to be the case for 4d $\mathcal{N}=1$ theories and M-theory on a Joyce $G_2$-manifold $\mathbb{T}^7/(\mathbb{Z}_2\oplus\mathbb{Z}_2\oplus\mathbb{Z}_2)$. While it is clear that 4d $\mathcal{N}=1$ moduli space is not locally symmetric, in \cite{Baines:2026aug} it is argued that the results should still apply in asymptotic regimes where an approximate local symmetry appears. It would be interesting to understand whether the two results are equivalent, though we leave this for a future work.

\section{Conclusions\label{sec.concl}}

In this work, we have investigated how much information about the ultraviolet spectrum of a four-dimensional $\mathcal N=1$ effective theory can be extracted directly from its asymptotic K\"ahler potential at the infinite distance boundaries of the moduli space -- where approximate axionic shift symmetries arise. Our starting point was the Integral Scaling Relation between the tensions of EFT strings and the characteristic mass scales of the towers of states becoming light at infinite distance. This relation has been observed in a broad class of string compactifications, both for the leading tower and for the subleading towers generating the tower polytope (i.e. the convex hull of the towers of states). Here, we have reversed the usual top-down logic and treated it as a quantum-gravity consistency condition on otherwise generic $4d$ $\mathcal N=1$ EFTs.

More concretely, we have developed a bottom-up reconstruction algorithm whose input is an asymptotically shift-symmetric K\"ahler potential of the form   $K \sim -\log P(s) $ at perturbative level.
The EFT-string tensions, and therefore their associated $\vec\zeta_T$-vectors, can be computed directly from this low-energy data. The Integral Scaling Relation then implies that the $\vec\zeta$-vectors of the light towers must belong to the lattice generated by the elementary EFT-string vectors. We further impose the Emergent String Conjecture (ESC), both in four dimensions and recursively upon decompactification, together with the assumption that the decompactification limits approach higher-dimensional Minkowski vacua where internal warping becomes asymptotically diluted. These conditions select the lattice points with the decay rates appropriate to either Kaluza--Klein towers or string oscillator modes and constrain how they can be assembled into globally consistent tower polytopes.\\

When applied to K\"ahler potentials arising in known string and M-theory compactifications, the reconstruction algorithm reproduces precisely the expected arrangements of towers and duality frames. For the cubic polynomials $P(s)$ defining K\"ahler potentials associated with Calabi--Yau compactifications, we recover the two families realized, respectively, in heterotic $E_8\times E_8$/type IIA constructions and in heterotic $SO(32)$/type I/type IIB/F-theory constructions. The same K\"ahler potential can therefore admit more than one consistent UV tower arrangement, even though the number of possibilities is highly constrained and matches string theory results. For the rank-seven K\"ahler potential associated with M-theory compactified on a Joyce $G_2$ manifold, the algorithm instead selects a unique tower polytope, coinciding with the one obtained from the top-down compactification.

We have subsequently used the reconstruction algorithm to constrain more general asymptotic K\"ahler potentials
with $P(s)$ homogeneous and of degree at most seven. When $P(s)$ is dominated by a single leading monomial on the entire saxionic cone, we have classified the allowed tower polytopes for up to seven saxions. Not every seemingly consistent K\"ahler potential passes the reconstruction test: certain monomials admit no tower spectrum compatible with both the Integral Scaling Relation and the recursive ESC. For every admissible monomial, the algorithm produces at most two possible tower polytopes, and in most cases the result is unique. 

For a general homogeneous polynomial, the different monomials provide the local building blocks controlling the various growth sectors of the saxionic cone. Consistency then requires the tower arrangements associated with adjacent growth sectors to agree along their common interfaces. This gluing condition produces additional restrictions that cannot be detected by considering the individual monomials separately. We have analyses this condition explicitly for two-saxion K\"ahler potentials. In particular, for degree five the only admissible possibilities are the individual monomials $(s^1)^4s^2$ and $s^1(s^2)^4$, while more general polynomials interpolating between different degree-five growth sectors fail the gluing condition. Analogous obstructions exclude degree-six and degree-seven polynomials containing $(s^1)^5s^2$ or $(s^1)^5(s^2)^2$, respectively, as leading terms in an asymptotic growth sector. The same construction extends in principle to more saxions, although the complete classification becomes computationally more involved.

In conclusion, under the assumptions detailed above, we can rule out certain perturbative forms of 4d K\"ahler potentials, provided that the ESC and the Integral Scaling Relation are universal UV consistency criteria. We should also remark that, even if the discussion seems focused on $\mathcal{N}=1$ EFTs, our results equally apply to extended supersymmetric setups where infinite distance boundaries are also characterized by axionic shift symmetries. In particular, our bottom-up analysis equally applies to 4d $\mathcal{N}=2$ EFTs exhibiting the type of K\"ahler potentials considered here, constraining in the same way their UV spectrum. In fact, in those settings, one often does not need to worry about potential quantum obstructions (as those in \cite{Cvetic:2024wsj,Kaufmann:2026tsy,Kaufmann:2026fli,Kaufmann:2026mha}), allowing one to recover the complete tower polytopes constructed here rather than only subsets thereof. \\

We have also clarified the relation between the Integral Scaling Relation and the taxonomy rules for particles of \cite{Etheredge:2024tok}, which are derived from applying the ESC recursively. The two sets of conditions are neither equivalent nor does either one imply the other. The taxonomy rules constrain the relative scalar products between the vectors of generating towers, whereas the Integral Scaling Relation constrains their embedding into the lattice determined by the EFT strings. In the examples with two and three saxions, only a subset of the solutions to the taxonomy rules can be embedded into an EFT-string lattice associated with a K\"ahler potential of degree at most seven. 
It would be interesting, though, to compare the Integral Scaling Relation with the taxonomy rules for branes derived in \cite{Etheredge:2024amg,Etheredge:2025ahf}. Although at first glance, the Integral Scaling Relation is more constraining than the taxonomy rules for extended objects (as they allow other types of strings as well), it could be interesting to study if the former can be recover from the latter upon imposing some additional constraint that identifies the EFT strings (see \cite{TBAeftMUNICH} for future work on this).

Finally, perhaps one of the most striking results is the relation between this bottom-up classification and M-theory on Joyce manifolds. We find that every tower polytope allowed by the reconstruction algorithm for monomials of degree at most seven can be obtained as a ``good" slice of the rank-seven polytope arising from M-theory on    $\mathbb{T}^7/(\mathbb Z_2\oplus\mathbb Z_2\oplus\mathbb Z_2) $. Even those associated to cubic polynomials -- that typically arise in Heterotic or Type II Calabi-Yau compactifications -- can be embedded in this M-theory construction.
Conversely, the slices of this M-theory polytope -- when properly taken as explained in Section \ref{s.String UNI} -- reproduce precisely only the monomials and tower arrangements selected by the bottom-up analysis, including their non-BPS string content. This suggests that reducing the amount of supersymmetry relative to a toroidal compactification obstructs certain limits and modifies the global structure, but does not give rise to new types of limits: the arrangement of towers within each growth sector remains unchanged, although their microscopic interpretation could be different.

This matching provides evidence for a sort of string universality at the asymptotic boundaries of 4d $\mathcal N=1$ moduli spaces: within the class considered here, every tower arrangement compatible with our quantum-gravity assumptions admits a realization in string or M-theory. This statement should not be interpreted as a proof that every consistent $\mathcal N=1$ supergravity belongs to the string landscape. Our reconstruction relies on the Integral Scaling Relation, the recursive Emergent String Conjecture, the bound on the number of decompactifying dimensions and the assumption of asymptotically unwarped Minkowski decompactifications. Nevertheless, it is non-trivial that these inputs single out precisely the slices of concrete M-theory compactifications, since the algorithm contains no explicit information about them. In previous analysis of applying the ESC recursively from a bottom-up perspective \cite{Etheredge:2024tok}, there were always certain tower polytopes that didn't have yet a known stringy embedding, but these are ruled out when imposing the Integral Scaling Relation.\\

There are several natural directions for future work. A central question is whether the Integral Scaling Relation, including the observed bound on its scaling weight $w$, can be derived from more fundamental quantum-gravity principles. This would tell us if it is only a property of the known string theory landscape or a more fundamental UV consistency condition. At the very least, even in the former case, it is impressive how such a simple relation between EFT strings and particles can encode so much information about the UV spectrum of the EFT, in such a way that it reproduces only those arrangements of UV towers that exist in known string theory constructions.

 It would also be interesting to relax the assumption of unwarped Minkowski decompactifications and determine how the reconstruction algorithm is modified for highly warped limits \cite{Etheredge:2023odp,Raucci:2026fzp}. This would change the KK exponential rate and potentially allow for other types of tower polytopes. In such limits, warping modifies the exponential decay rates of the KK towers and may allow for additional types of tower polytopes. Recent work \cite{Raucci:2026fzp} suggests that this can occur only when the warping is sourced by codimension-one extended objects, substantially restricting the number of possibilities. It is also known that the standard taxonomy rules no longer apply in such strongly warped limits, raising the possibility that the Integral Scaling Relation may likewise need to be modified.
 
 On a different note, the phenomenon of \emph{sliding} of towers characteristic of warped compactifications (namely, that the tower vector changes as we move in moduli space perpendicularly to the asymptotic trajectory) also happens at the interfaces between growth sectors (non-regular limits in the notation of \cite{Etheredge:2024tok}). However, in these cases, even if the taxonomy rules do not apply, we have found that the Integral Scaling Relation still holds, as the exponential rates of the UV towers along the EFT string flows remain unchanged. In this sense, the Integral Scaling Relation appears better suited than the taxonomy rules to characterizing the global structure of the tower polytope, as it remains valid throughout the entire saxionic cone, including both regular and non-regular limits. 
 
 Other extensions of our work include a systematic implementation of the gluing constraints for polynomials with more than two saxions. Although conceptually clear, it becomes technically involved due to the increasing number of possibilities. For cubic polynomials, the gluing constraints seem equivalent to having a positive definite field metric, whereas for polynomials of degree greater than four they impose additional non-trivial conditions on the global arrangement of towers. It would also be interesting to compare our results with those of \cite{Baines:2025upi,Baines:2026aug},  which rely on the ESC and properties of the duality groups, in order to establish whether the Integral Scaling Relation can be obtained from those duality properties. Finally, it would be valuable to extend our analysis to non-supersymmetric settings, as the existence of the EFT strings does no rely on supersymmetry but on the approximate axionic shift symmetries emerging at the boundaries of the moduli space. The interplay between EFT strings, towers of particles and EFT membranes (sourcing a scalar potential) will be presented in \cite{EFTstr-and-mem}. The inclusion of the potential breaking supersymmetry spontaneously may obstruct certain limits, but it does not change the asymptotic results found in this paper as long as it vanishes at infinite distance.
 
 All these questions may help establish whether the universality observed in this work reflects a more general organizing principle of quantum gravity, and whether EFT string tower building provides a general bottom-up framework for reconstructing the spectrum in the EFT perturbative regimes.\\

\textbf{Acknowledgments}: 
We are grateful to Muldrow Etheredge, Bernardo Fraiman, Damian van de Heisteeg, \'Alvaro Herr\'aez, Luca Melotti, Timo Weigand and Max Wiesner,  for very illuminating discussions and comments. A.G. also thanks the Institute for Theoretical Physics in KU Leuven for hospitality in the late stage of this work, while I.R. thanks IFT UAM-CSIC for hospitality during the early stages of the development of this paper.
The authors thank the Spanish Agencia Estatal de Investigaci\'on through the grant “IFT Centro de Excelencia Severo Ochoa” CEX2020-001007-S and the grants PID2021-123017NB-I00 and PID2024-156043NB-I00, funded by MCIN/AEI/10.13039/ 501100011033 and by ERDF “A way of making Europe”.  
 The work of A.G. is supported by the fellowship LCF/BQ/DI23/11990073 from ``La Caixa'' Foundation (ID 100010434), with additional funding from a STSM grant within the COST action CA22113 ``Fundamental Challenges in Theoretical Physics''. The work of I.R. is supported by the European Union through ERC Starting Grant SymQuaG-101163591 StG-2024, and he acknowledges the additional funding of the Spanish FPI grant No. PRE2020-094163 and the ERC Starting Grant QGuide101042568 - StG 2021. The work of I.V. has been supported by the ERC Starting Grant QGuide101042568 - StG 2021 and the Project ATR2023-145703 funded by MCIN/AEI/ 10.13039/501100011033.

\appendix
\section{Basic identities for EFT strings\label{app.id}}
In this appendix, we describe some identities and results concerning the $\zeta$-vectors of the EFT strings. We recall that, given some K\"ahler potential in a given 4d $\mathcal{N}=1$ theory enjoying an approximate shift axionic symmetry on their chiral scalars $t^j=a^j+is^j$, the EFT string tension with charges $\mathbf{e}$ is given by \eqref{eq.EFTstringTension}
\begin{equation}\label{eq. tension app}
\mathcal{T}_{\mathbf{e}}=M_{\rm Pl,4}^2 e^jl_j\,,\quad\text{where}\quad l_k=-\frac{1}{2}\partial_{s^j}K\,.
\end{equation}
This motivates the definition of the associated $\zeta$-vector as \eqref{eq.zeta vec}
\begin{equation}
	\vec{\zeta}_{\mathcal{T}}=-\vec\nabla\log\frac{\mathcal{T}(\vec{\varphi})^{1/2}}{M_{\rm Pl,4}}\,.
\end{equation}
We will now describe some basic results concerning the above $\zeta$-vectors, some of which already appeared in \cite{Grieco:2025bjy}. In the following, we will take $K\sim -\log P(s)$, with $P(s)$ a homogeneous polynomial on the saxions.
\begin{fct}
	Along the saxionic flow induced by approaching the core of an EFT string of charge $\mathbf{e}$, we have $\vec{\zeta}_{\mathcal{T}_{\mathbf{e}}}\propto\mathbf{e}$, proportional to the trajectory taken.
\end{fct}
\begin{proof}
	For a given trajectory \eqref{eq.profilesEFT}, which we can parametrize as $\mathbf{s}(\sigma)=\mathbf{s}_0+\mathbf{e}\,\sigma$, with $\sigma>0$, it is clear that $\partial_\sigma\mathbf{s}(\sigma)=\mathbf{e}$ is a tangent vector (generally not of unit norm). On the other hand,
	\begin{equation}
	\zeta^i_{\mathcal{T}_{\mathbf{e}}}=-\frac{1}{2}\mathsf{G}^{ij}\partial_j\log\Big(-\frac{1}{2}e^a\partial_a K\Big)=-(\partial^2 K)^{ij}\frac{e^a\partial_j\partial_a K}{e^a\partial_aK}=\frac{M_{\rm Pl,4}^2}{2\mathcal{T}_{\mathbf{e}}}e^i\,,
\end{equation}		
where use that $\mathsf{G}_{jk}=\frac{1}{2}\partial_{s^j}\partial_{s^k}K$. Thus $\vec{\zeta}_{\mathcal{T}_{\mathbf{e}}}=\frac{M_{\rm Pl,4}^2}{2\mathcal{T}_{\mathbf{e}}} \mathbf{e}=\frac{M_{\rm Pl,4}^2}{2\mathcal{T}_{\mathbf{e}}} \partial_\sigma\mathbf{s}$, as we wanted to show.
\end{proof}

\begin{fct}\label{fact.length}
	Along a given EFT string flow, the length of the associated $\zeta$ vector is given by
	\begin{equation}\label{eq.length app}
		|\vec{\zeta}_{\mathbf{e}}|_{\rm flow}=\frac{1}{\sqrt{2{\rm deg}_{\sigma}P\big(s(\sigma)\big)|_{\rm flow}}}\,.
	\end{equation}
	An immediate consequence is that EFT strings corresponding with fundamental strings are in one to one correspondence with saxions with degree 1 in $P(s)$. Furthermore, $\vec{\zeta}_{\mathbf{e}}$ is perpendicular to $\Delta_I$, the simplex generated by the $\zeta$-vectors of the elementary EFT string of the directions under which $\mathbf{e}$ is charged.
\end{fct}
\begin{proof}
	From \eqref{eq. tension app} and the fact that $P(s)$ is a homogeneous polynomial, we have that for a given flow associated to $\mathbf{e}$
	\begin{equation}
	\mathcal{T}_{\mathbf{e}}(\sigma)=M_{\rm Pl,4}^2\frac{{\rm deg}_{\sigma}P\big(s(\sigma)\big)|_{\rm flow}}{2\sigma}\to 0\,,\qquad \mathcal{T}_{e^{i_0}}(\sigma)=M_{\rm Pl,4}^2\frac{{\rm deg}_{\sigma}P\big(s^{i_0}(\sigma)\big)|_{\rm flow}}{2\sigma}\,.
	\end{equation}
Note that the elementary strings only become light along the flow if $e^{i_0}\neq 0$, but if they do, they do so at the same rate as $\mathcal{T}_{\mathbf{e}}$ in terms of $\sigma$. This means that along the such trajectory (which as we saw in the previous result, is parallel to $\mathbf{e}\propto \vec{\zeta}_{\mathcal{T}_{\mathbf{e}}}$),
\begin{equation}\label{eq. prod app}
	\vec{\zeta}_{\mathcal{T}_{\mathbf{e}}}\cdot\vec{\zeta}_{\mathcal{T}_{e^{i_0}}}=\left\{\begin{array}{ll}
	|\vec{\zeta}_{\mathcal{T}_{\mathbf{e}}}|^2&\text{if }e^{i_0}\neq 0\\
	0&\text{if }e^{i_0}=0
	\end{array}\right.\;,
\end{equation}
which means that $\vec{\zeta}_{\mathcal{T}_{\mathbf{e}}}\perp\Delta_I$ and ${\rm dist}(\vec 0,\Delta_I)=|\vec{\zeta}_{\mathbf{e}}|_{\rm flow}$. Now, evaluating explicitly,
\begin{equation}
	|\vec{\zeta}_{\mathbf{e}}|_{\rm flow}^2=\left.M_{\rm Pl,4}^4e^ae^b\mathsf{G}^{ij}\frac{\partial_i\partial_a K\partial_j\partial_b K}{4\mathcal{T}_{\mathbf{e}}}\right|_{\rm flow}=\left.M_{\rm Pl,4}^4\frac{\partial_\sigma^2 K}{2\mathcal{T}^2}\right|_{\rm flow}=\frac{1}{{2{\rm deg}_{\sigma}P\big(s(\sigma)\big)|_{\rm flow}}}\,,
\end{equation}
	as we were aiming to show.
\end{proof}
An immediate corollary is that the EFT string with the shortest norm along its flow is that charged under all saxionic directions, for which $|\vec{\zeta}_{\mathcal{T}}|_{\rm flow, min}=\frac{1}{\sqrt{{\rm deg}P(s)}}$. Additionally, for ${\rm deg}P(s)>1$, it is clear that one cannot find BPS fundamental strings in directions along the bulk of the saxionic cone $\Delta$.
\begin{fct}\label{fct.eq}
	The following equality holds:
	\begin{equation}
		\left\{\vec{x}\cdot\vec{\zeta}_{\mathcal{T}_{e^i}}=w|\vec{\zeta}_{\mathcal{T}_{e^i}}|^2\,,\;\text{with }\,w\in\mathbb{Z}_{\geq 0}\;\forall i=1,\dots,n\right\}=\left\{\sum_{i=1}^nw_i\vec{\zeta}_{\mathcal{T}_{e^i}}\,|\,w_i\in\mathbb{Z}_{\geq 0}\;\forall i=1,\dots,n\right\}\,,
	\end{equation}
	where $\vec{\zeta}_{\mathcal{T}_{e^i}}$ are the $\zeta$-vectors associated to the elementary EFT strings.
\end{fct}
\begin{proof}
	The $\supseteq$ inclusion is straightforward from \eqref{eq. prod app}. As for the $\subseteq$ one, take an element $\vec{x}$ in the left set. Since from \eqref{eq. prod app} the set of $\zeta$-vectors associated to elementary EFT strings define an orthogonal basis of $T_p\Delta$, any expansion of $\vec{x}=\sum_{i=1}^nw_i\vec{\zeta}_{\mathcal{T}_{e^i}}$ in such basis is straightforward in terms of non-negative integers, due to the uniqueness of the components $w_i$ and the definition of $\vec{x}$ as an element of the left set.
\end{proof}
\
\subsection{Bounds on the scaling weight\label{fact.bound}}
We will now describe how we can get bounds on the scaling weight $w_i$, given the degree $p_i$ of the polynomial along a given saxionic direction. From \eqref{eq.length app}, we have that $|\vec{\zeta}_{\mathcal{T}_{e^i}}|=\frac{1}{\sqrt{2p_i}}$. Then, from the Integral Scaling Relation \eqref{eq.intscaling}, we have that
\begin{equation}
w_i=\frac{\vec{\zeta}_*\cdot \vec{\zeta}_{\mathcal{T}_{e^i}}}{|\vec{\zeta}_{\mathcal{T}_{e^i}}|^2}\leq \frac{|\vec{\zeta}_*|}{|\vec{\zeta}_{\mathcal{T}_{e^i}}|}\,.
\end{equation}
Now, using that the length of the light towers is bounded from above by that of the KK-1, $|\vec{\zeta}_*|\leq|\vec{\zeta}_{\rm KK,1}|=\sqrt{\frac{d-1}{d-2}}=\sqrt{\frac{3}{2}}$, we arrive to
\begin{equation}
w_i\leq\lfloor\sqrt{3p_i}\rfloor
\end{equation}

We could try to find a similar argument for the lower bound on $w_i$. Assuming that the $\zeta$-vector of the lightest tower is aligned with that of the EFT string along the flow of the later (this is always the case for observed constructions), we have that, since $|\vec{\zeta}_*|\geq|\vec{\zeta}_{\rm osc}|=\frac{1}{\sqrt{d-2}}$, $w_i\geq\lceil\sqrt{p_i}\rceil$. This way, we arrive to
\begin{equation}\label{eq.bounds}
\boxed{
\lceil\sqrt{p_i}\rceil\leq w_i\leq\lfloor\sqrt{3p_i}\rfloor}\,,\quad\text{this is}\quad
\begin{array}{|c|c|c|c|c|c|c|c|c|c|c|c|} \hline p_i &\; 1\; &\; 2 \;& 3 & 4 &\; 5 \;& 6 & 7&8&9&10 &\dots \\ \hline w_i & 1 & 2 & 2,3 & 2,3 & 3 & 3,4 & 3,4 &3,4&3,5&4,5&\dots\\ \hline \end{array}\,,
\end{equation}
and thus the scaling weight grows like $\mathcal{O}(\sqrt{p_i})$. Imposing as a bottom-up constraint the $w\leq 3$ bound on the scaling weight, we recover that $\deg P(s)\leq 9$, which is not saturated by the known examples (the largest known value is given by $\deg P(s)=7$ for M-theory compactified on $G_2$-manifolds, see e.g., \cite{Beasley:2002db}).

We can compare these bounds, particularly the upper one in \eqref{eq.bounds}, with the K\"ahler potentials obtained in top-down constructions, see \cite{Lanza:2021udy, Grieco:2025bjy} as well as Appendix \ref{app.topdown}. Having that the K\"ahler potential has respectively degree 4 and 7 for CY and $G_2$ compactifications, we obtain\footnote{Note that one cannot say much about the lower bound, since usually there is always at least one saxion (usually the one associated to the dilaton) of degree 1 for multi-moduli settings.} 
	 \begin{equation}
	 	w\leq \left\{\begin{array}{ll}
	 	3&\text{for CY 3-folds}\\
	 	4&\text{for }G_2\text{ manifolds}
	 	\end{array}\right.\;.
	 \end{equation}
	This automatically translates to $w\leq 3$ for CY 3-fold compactifications. The same bound, however, for M-theory on $G_2$ manifolds is not immediate. To obtain it, note that in such case $K\sim -3\log V_X$ \cite{Beasley:2002db}, with $V_X$ the volume of the 7-manifold in $M_{\rm Pl,11}$ units, and as such $P(s)=V_X^3$. Now, an EFT string limit charged under all saxions (i.e., with degree 7) in turn is equivalent to an overall homogeneous growth of the whole volume, $V_X=P(s)^{1/3}\to\infty$, and thus a full decompactification of the complete 7-manifold to 11d, with $m_{\rm KK,7}$ being the leading tower in that limit. This means that
	 \begin{equation}
	  w \frac{1}{\sqrt{2\times 7}}\leq |\vec{\zeta}_{\rm KK,7}|=\sqrt{\frac{9}{14}}\;,
	 \end{equation}
	 precisely saturated for $w=3$ and forbidding larger values. One can then conclude that in the top-down cases where one could naively have $w\geq4$ from a bottom-up perspective, this is never realized.\\

\section{Top-down duality frames from String and M-theory compactifications\label{app.topdown}}
In this Appendix we present a quick overview of the global arrangements of duality frames and light towers, together with the EFT strings, for a wide variety of 4d $\mathcal{N}=1$ top-down constructions studied in \cite{Lanza:2021udy,Grieco:2025bjy}, to which we refer for more details and discussion. In such works the Integral Scaling Relation $\vec\zeta_{*} \cdot \vec\zeta_{\cT_{\mathbf{e}}}=w|\vec\zeta_{\cT_{\mathbf{e}}}|^2$ with $w\in\{0,1,2,3\}$ holds for all light towers generating the convex hull of $\zeta$-vectors. We also discuss possible obstructions of some limits within the saxionic cone due to effects not accessible from the EFT data.

\subsection{Heterotic \texorpdfstring{$E_8\times E_8$}{E8xE8} on CY 3-folds}
We first consider compactifications of HE string theory on a CY 3-fold $X$. The 4d saxions $\{s^a\}_{a=1}^{h^{1,1}}$ originate from the expansion of the string frame K\"ahler form $J=s^a\omega_a$ over an integral basis $\{\omega_a\}_{a=1}^{h^{1,1}}$ of $H^{1,1}(X)$, together with the \emph{universal saxion} $s^0=e^{-2\phi_4}=e^{-2\Phi}\mathcal{V}_X$, where $\Phi$, $\phi_4$ and $\mathcal{V}_X=\frac{1}{3!}\kappa_{abc}s^as^bs^c$ are respectively the 10d heterotic dilaton, the 4d dilaton and the volume of $X$ in string units.

The leading form of K\"ahler potential is given by $K=-\log s^0-\log\mathcal{V}_X$, modulo non-perturbative corrections. EFT strings correspond to the F1 string and NS5-branes wrapped on effective divisors $D=e^a[\omega_a]$, with $\{[\omega_a]\}_{a}^{h^{1,1}}$ an integer basis of $H_{1,1}(X)$. 

In Figure \ref{fig.bu-HE} we present the arrangement of duality frames and light towers, together with the $\zeta$-vectors for EFT strings, for different forms of the $\mathcal{V}_X=\frac{1}{3!}\kappa_{abc}s^as^bs^c$ volume.

\begin{figure}[ht]
\begin{center}
\begin{subfigure}[b]{0.5\textwidth}
\centering
     \resizebox{0.95\textwidth}{!}{%
    \includegraphics[scale=1]{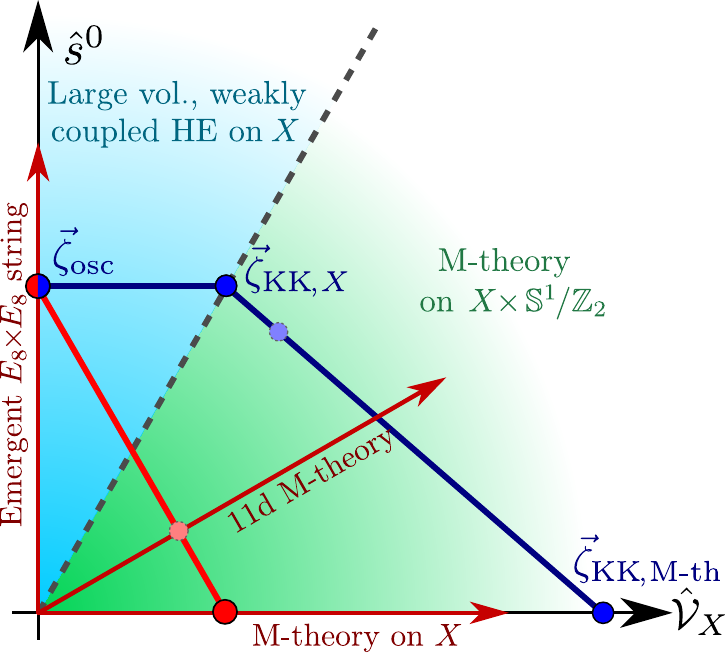}
    }
\caption{Slice $\{s^0,\mathcal{V}_X\}$} \label{fig.bu-HE-1}
\end{subfigure}
\begin{subfigure}[b]{0.49\textwidth}
\centering
     \resizebox{0.95\textwidth}{!}{%
    \includegraphics[scale=1]{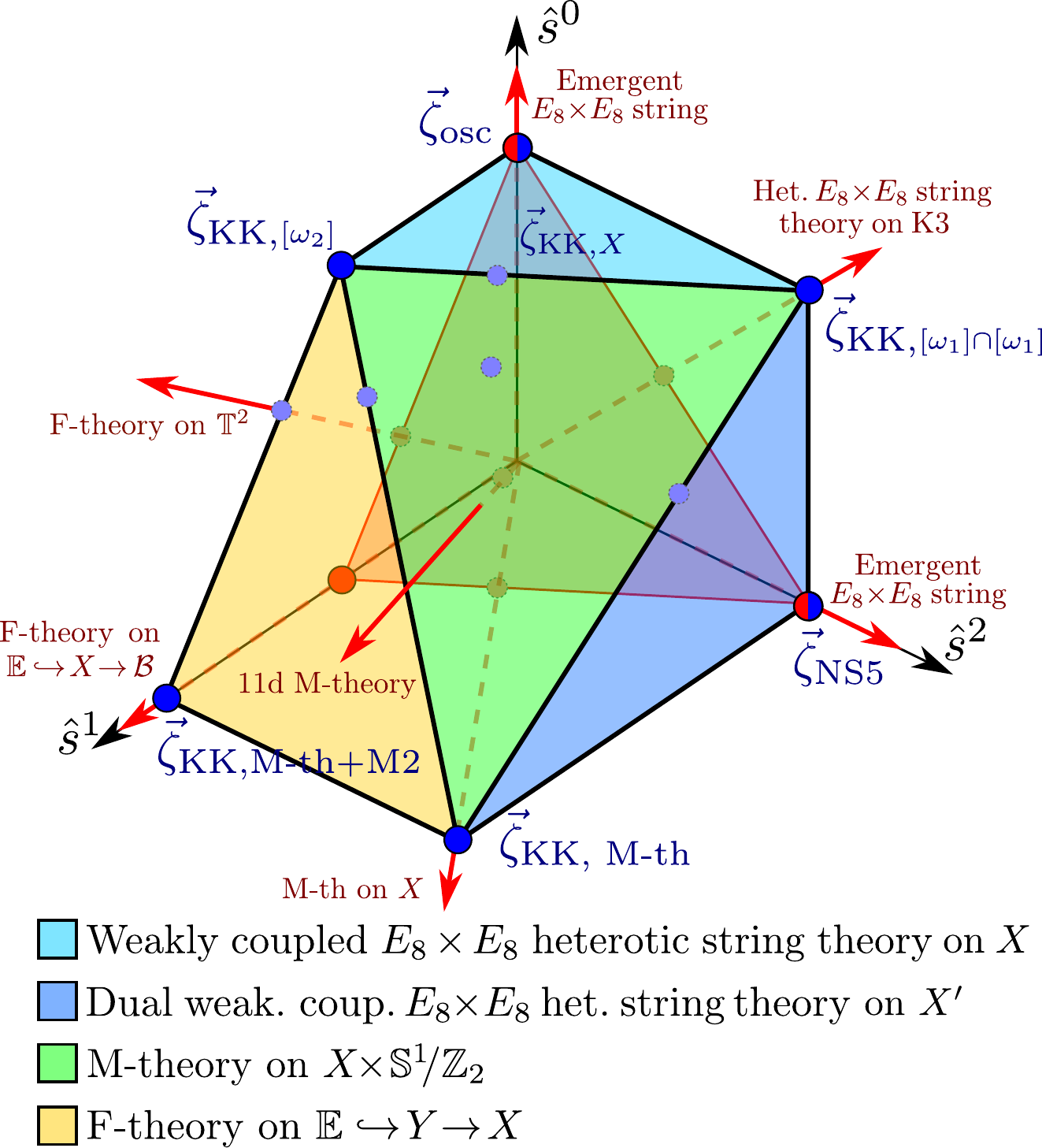}
    }
\caption{Slice $\{s^0,s^1,s^2\}$ for $\mathcal{V}_X\sim (s^1)^2s^2$.} \label{fig.bu-HE-2}
\end{subfigure}
\begin{subfigure}[b]{0.49\textwidth}
\centering
     \resizebox{0.95\textwidth}{!}{%
    \includegraphics[scale=1]{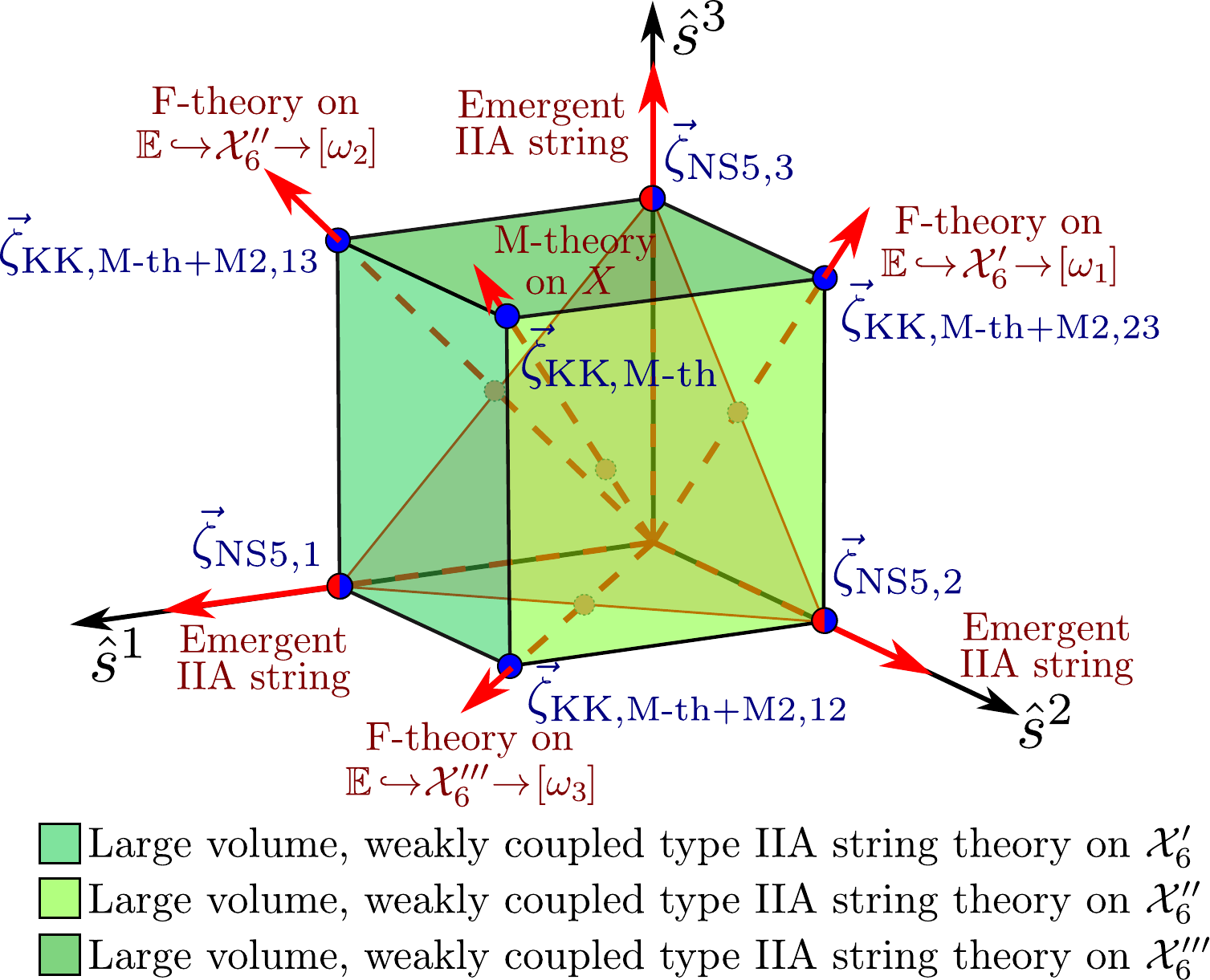}
    }
\caption{Slice $\{s^1,s^2,s^3\}$ for $\mathcal{V}_X\sim s^1s^2s^3$.} \label{fig.bu-HE-3}
\end{subfigure}
\begin{subfigure}[b]{0.49\textwidth}
\centering
     \resizebox{0.95\textwidth}{!}{%
    \includegraphics[scale=1]{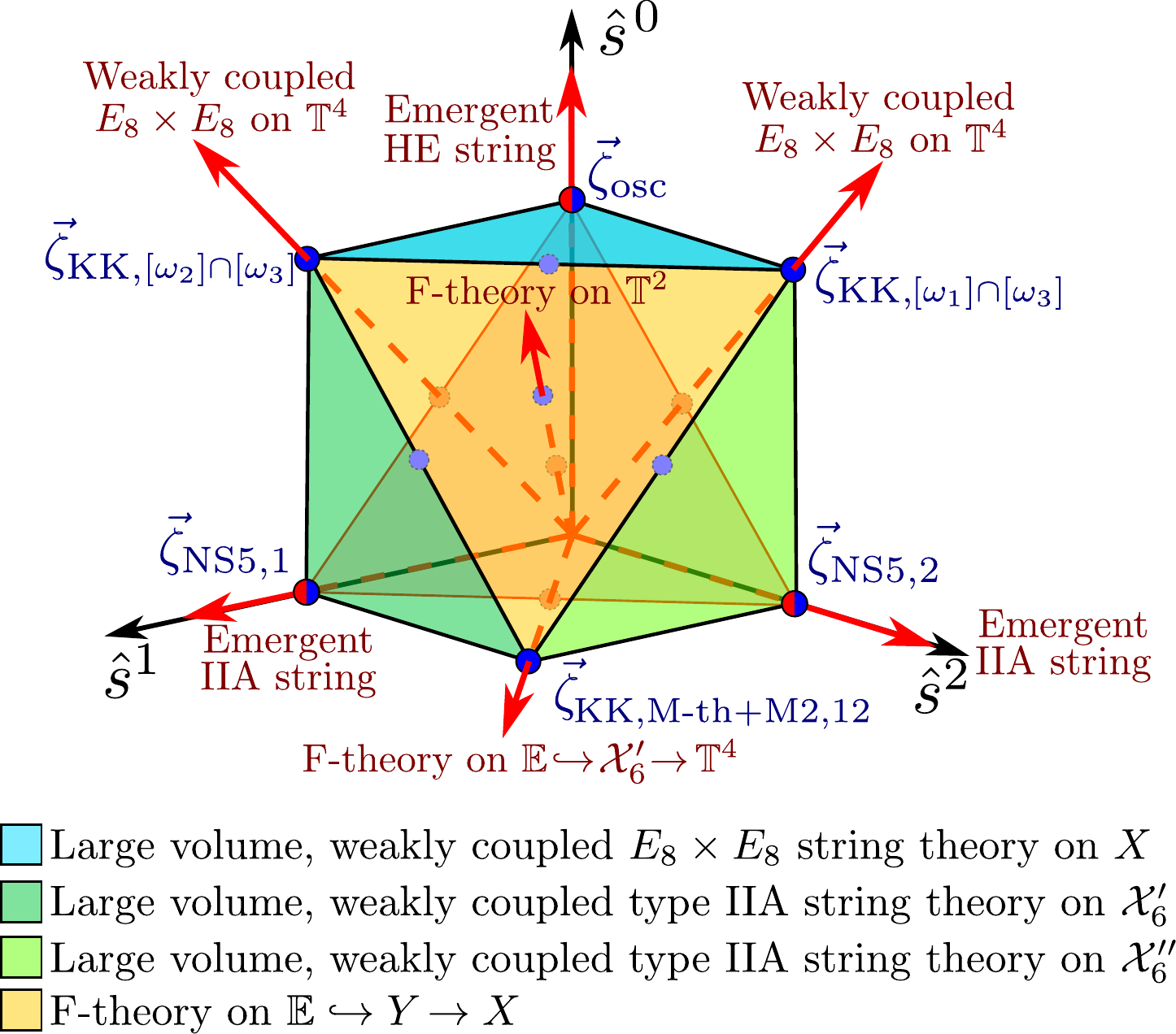}
    }
\caption{Slice $\{s^0,s^1,s^2\}$ for $\mathcal{V}_X\sim s^1s^2s^3$.} \label{fig.bu-HE-4}
\end{subfigure}
\caption{Arrangement of duality frames, light towers (in blue) and EFT strings (in a red hyperplane, with the EFT string flow directions and the associated dual theory at infinite distance labeled) for different volume leading forms of $\mathcal{V}_X$, in (locally) flat coordinates, for HO on CY 3-fold compactifications. For simplicity we do not label the $\zeta$-vectors of the EFT strings. Bounded states of light towers and EFT strings are depicted in a lighter shade. See \cite[Section 3]{Grieco:2025bjy} for more details.
\label{fig.bu-HE}}
\end{center}
\end{figure}

Note that, as explained in \cite{Witten:1996mz,Cvetic:2024wsj}, the large volume limit $\mathcal{V}_X\gg s^0$ (for which $\Phi\to\infty$ and we would expect decompactification of the M-theory Hor\v{a}va-Witten interval) might be obstructed by a diverging of the gauge coupling and warp factor at one of the $E_8$-boundaries, at least in the standard embedding of the gauge group. While such singularity can be resolved by the inclusion of non-perturbative effects, see \cite{Cvetic:2025nfx}, it is not expected that the resulting theory corresponds to a perturbative 4d Minkowski vacuum. Note that these limits might not obstructed if one takes a different embedding of the $E_8\times E_8$ gauge group, see e.g., \cite{Reig:2025gch}, in which case the M-theory description could still be valid.

\subsection{Type IIA on CY orientifolds}
Next we take type IIA string theory compactified on a CY orientifold $X/\mathfrak{i}$ with O6 involution $\mathfrak{i}:X\to X$ acting on the K\"ahler and $(3,0)$-forms as $\mathfrak{i}^*J=-J$ and $\mathfrak{i}^*\Omega=\bar\Omega$ \cite{Grimm:2004ua}. We focus on the K\"ahler sector spanned by the saxions $\{s^a\}_{a=1}^{b_2(X)^+}$ measuring the string-frame volume of even 2-cycles of $H_2(X)^+$ (so that $\mathcal{V}_X=\frac{1}{3!}\kappa_{abc}s^as^bs^c$) and the Hitchin function $\mathcal{H}=\frac{i}{8}\int_X e^{-\Phi}\Omega\wedge\bar\Omega$, playing the role of the universal saxion $s^0$. Working in the large-volume/small coupling perturbative regime (n=thus neglecting internal warping and backreacting) the K\"ahler potential is given by $K=-\log\mathcal{H}-\log\mathcal{V}_X$ at leading order, again modulo non-pertubative contributions.

The setting is similar to the previously discussed heterotic $E_8\times E_8$ string on a CY 3-fold, with EFT strings corresponding to the fundamental type IIA string and NS5-branes wrapping \emph{even} (nef and effective) divisors $D^+e^a D^+_a$, see \cite{Katz:2020ewz}.\footnote{One also finds EFT strings associated to the complex structure sector coming from D4-branes wrapping calibrated SLag 3-cycles, see \cite{Lanza:2021udy}, which were not considered in the analysis from \cite{Grieco:2025bjy}.} Additionally, limits for which $\mathcal{V}_X\gg s^0$ would again result in the 10d type IIA dilaton growing, thus decompactifying the M-theory circle. This way, as argued in \cite{Grieco:2025bjy}, the arrangement of duality frames, light towers and EFT strings is analogous to that of heterotic $E_8\times E_8$ string on a CY 3-fold, see Figure \ref{fig.bu-HE}, simply replacing the Ho\v{r}ava-Witten $\mathbb{S}^1/\mathbb{Z}_2$ by $\mathbb{S}^1$. 

Now, as discussed in \cite{Kaufmann:2026tsy}, there can be also generic obstructions to having infinite distance limits (i.e., a perturbative description) $\mathcal{V}_X\gg s^0$, following from the argued absence of an infinite tower of KK modes associated to the $\mathbb{S}^1$ fiber in the M-theory uplift to a $G_2$ manifold \cite{joyce1996a,joyce1996b,joyce2004constructing,Kachru:2001je},\footnote{The $\mathbb{S}^1$ circle is not a direct product, but rather is backreacted by the presence of D6-branes (where the $\mathbb{S}^1$ shrinks to zero size in a Taub-NUT geometry \cite{Townsend:1995kk}) or O6-planes (where the full 7-dimensional compact geometry is  locally given by a Atiyah-Hitchin manifold \cite{Seiberg:1996nz}).} as well as from possible inconsistencies in the worldsheet theory of the EFT string associated to such limit with the orientifold projection. Though these obstructed limits in the 4d $\mathcal{N}=1$ theory would be equivalent to those in the $E_8\times E_8$ construction with standard embedding, this would not be a problem in the 4d $\mathcal{N}=2$ case (compactifying type IIA on $X$ without further orientifolding), where the strong coupling limit is simply M-theory on the product $\mathbb{S}^1\times X$. In this case, the arrangement of duality frames would be a direct translation of those in Figure \ref{fig.bu-HE} upon $\mathbb{S}^1/\mathbb{Z}_2$ by $\mathbb{S}^1$ replacement.

\subsection{Heterotic \texorpdfstring{$SO(32)$}{SO(32)}/Type I on CY 3-folds}
Considering HO string theory on a CY 3-fold $X$ results in a similar 4d ingredients as in the $E_8\times E_8$ counterpart, with again EFT strings coming from the heterotic F1-string and NS5-branes wrapping effective divisors of $X$. One big difference is, however, the strong coupling limit, where in this case is the S-dual Type I string theory on the same 3-fold $X$ rather than decompactification to M-theory. The new theory has a non-BPS fundamental string (which does not correspond to an EFT string).\footnote{From the type I perspective, EFT strings correspond to the D1-brane and D5-branes wrapping divisors, in exact analogy with the heterotic F1-string and the wrapped NS5-branes.} Now, both in the heterotic and type I frames one could perform even-numbers of T-dualities on curves, divisors or the full $X$ to access additional saxionic directions, though generally this will not be allowed unless for special fibration points with SUSY enhancement, with said infinite distance limits being obstructed.

In Figure \ref{fig.bu-HO} we depict the arrangement of duality frames in the Heterotic $SO(32)$/Type I compactification on $X$, with obstructed directions colored in gray.

\begin{figure}[ht]
\begin{center}
\begin{subfigure}[b]{0.45\textwidth}
\centering
     \resizebox{0.95\textwidth}{!}{%
    \includegraphics[scale=1]{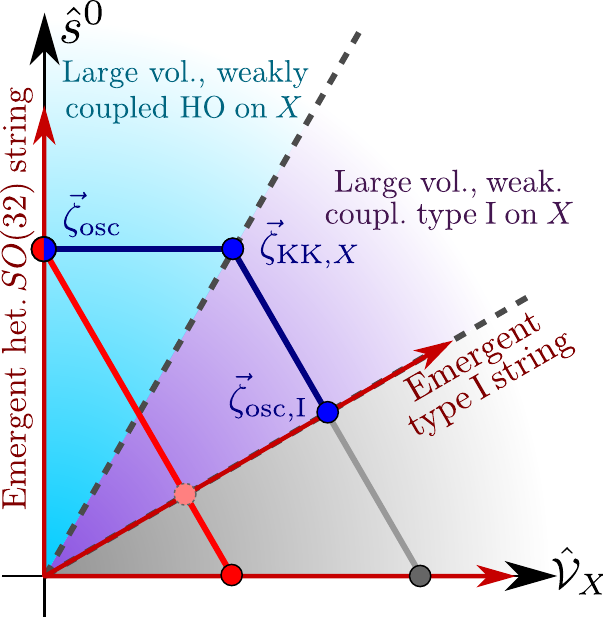}
    }
\caption{Slice $\{s^0,\mathcal{V}_X\}$} \label{fig.bu-HO-1}
\end{subfigure}
\begin{subfigure}[b]{0.54\textwidth}
\centering
     \resizebox{\textwidth}{!}{%
    \includegraphics[scale=1]{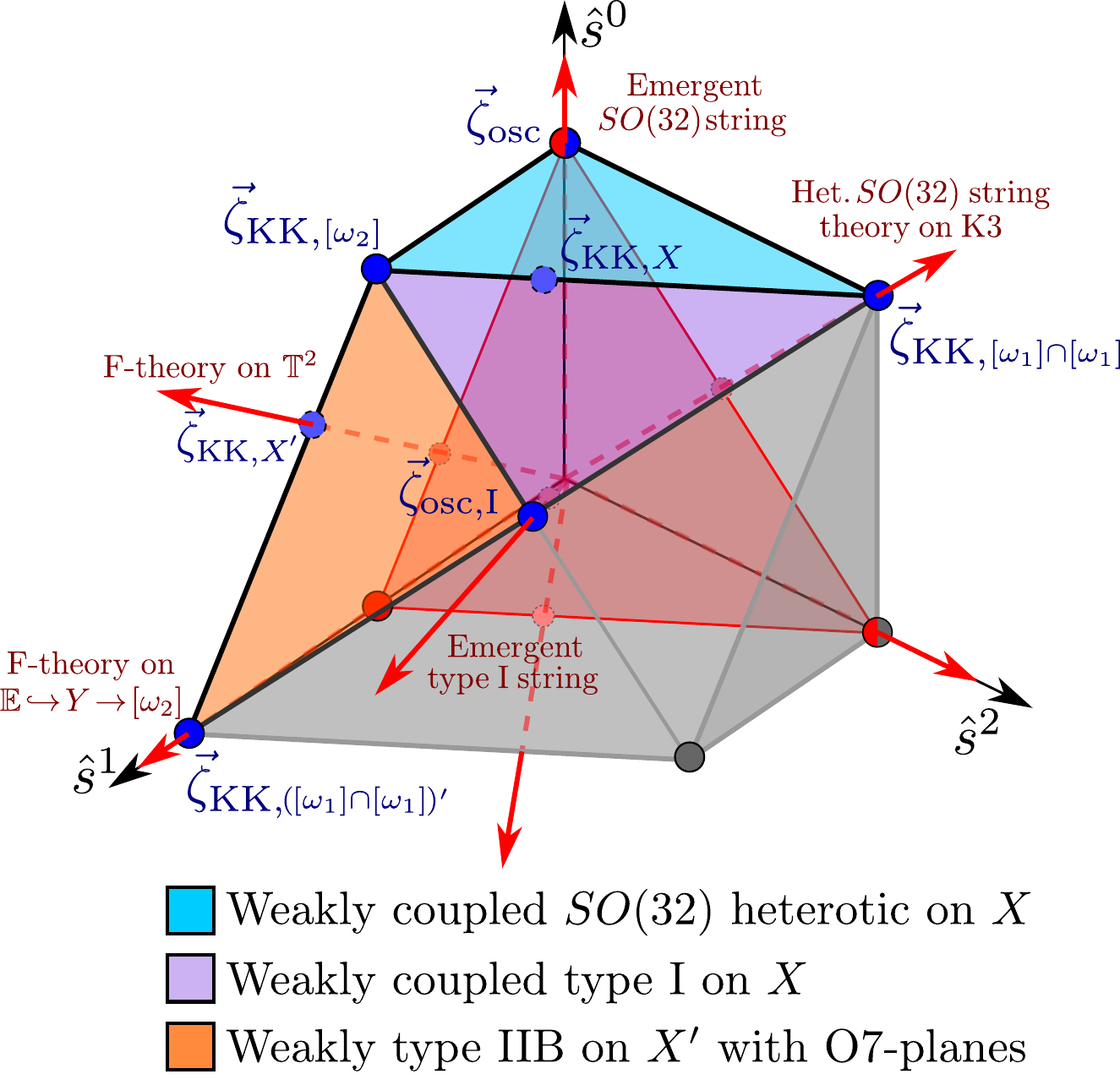}
    }
\caption{Slice $\{s^0,s^1,s^2\}$ for $\mathcal{V}_X\sim (s^1)^2s^2$.} \label{fig.bu-HO-2}
\end{subfigure}

\caption{Arrangement of duality frames, light towers (in blue) and EFT strings (in a red hyperplane, with the EFT string flow directions and the associated dual theory at infinite distance labeled) for different volume leading forms of $\mathcal{V}_X$, in (locally) flat coordinates, for HO on CY 3-fold compactifications. Obstructed regions are depicted in gray. For simplicity we do not label the $\zeta$-vectors of the EFT strings. Bounded states of light towers and EFT strings are depicted in a lighter shade. See \cite[Section 4.2]{Grieco:2025bjy} for more details.
\label{fig.bu-HO}}
\end{center}
\end{figure}

\subsection{Type IIB on CY orientifolds/F-theory on CY 4-folds}
We then turn to type IIB on CY orientifolds, better described by F-theory on elliptically fibered CY 4-folds $\mathbb{E}\hookrightarrow Y\to X$, with the base $X$ a K\"ahler manifold. We can expand the (Einstein frame) K\"ahler form of the base $J=v_a[D^a]$, where $\{D^a\}_{a=1}^{b_4(X)}$ is a base of $H_4(X,\mathbb{Z})$ divisors whose volume in 10d Planck units is measure by the 4d $\mathcal{N}=1$ saxions $\{s^a=\frac{1}{2}\in_{D^a}J\wedge J\}=\frac{1}{2}\kappa^{abc}v_bv_c$. Note the non-trivial relation between the 4d saxions and the K\"ahler moduli $\{v_a\}_{a=1}^{b_4(X)}$. The K\"ahler potential of the K\"ahler sector is given (again, modulo non-perturbative terms) by
\begin{equation}
	K=-2\log\int_X J\wedge J\wedge J=-2\log\big(\underbrace{\kappa^{abc}v_av_bv_c}_{3!V_X}\big)\,,
\end{equation}
of degree $\frac{3}{2}$ on the $\{s^a\}_{a=1}^{b_2(X)}$ saxions , where $V_X$ is the base volume in 10d Planck. We can additionally consider the complex structure moduli of the $Y$ 4-fold, of which we focus on that of the elliptic fiber, identified as the type IIB axio-dilaton $\tau=C_0+ie^{-\Phi}$. Along asymptotic limits in the above saxions, the 4d EFT strings come from D3-branes wrapped along effective curves $C=e^aC_a$ (with $\{C_a\}_{a=1}^{b_2(X)=b_4(X)}$ a dual base to that of the previous divisors) and a D7-brane wrapped on the whole $X$ base. Note that the type IIB F1-string is not an EFT string, as it can break (and thus is non BPS) due to the presence of O7-planes spanning along the 4d spacetime.

Depending on the fibration structure of the base ($\mathbb{T}^2\hookrightarrow X\to D$ or $\mathbb{P}^1\hookrightarrow X\to D$) of the K\"ahler base, the arrangement of duality frames (and the microscopic interpretation of the light towers) differs outside of the weakly-coupled, large volume initial type IIB compactification. We depict this in Figure \ref{fig.bu-IIB}. Note that, while in principle we can perform T-dualities along the directions of the fiber of the $X$ base when these become string-sized, this cannot generally be done along the larger divisors (or the whole base $X$), unless $X$ is a Cartesian product, which would imply some SUSY enhancement. This then prevents us from taking directions in the saxionic cone where ${\rm Im}\tau$ (and as a result the string length) grows faster then the internal volumes of $X$, see Figure \ref{fig.bu-IIB}. 

\begin{figure}[ht]
\begin{center}
\begin{subfigure}[b]{0.51\textwidth}
\centering
     \resizebox{0.95\textwidth}{!}{%
    \includegraphics[scale=1]{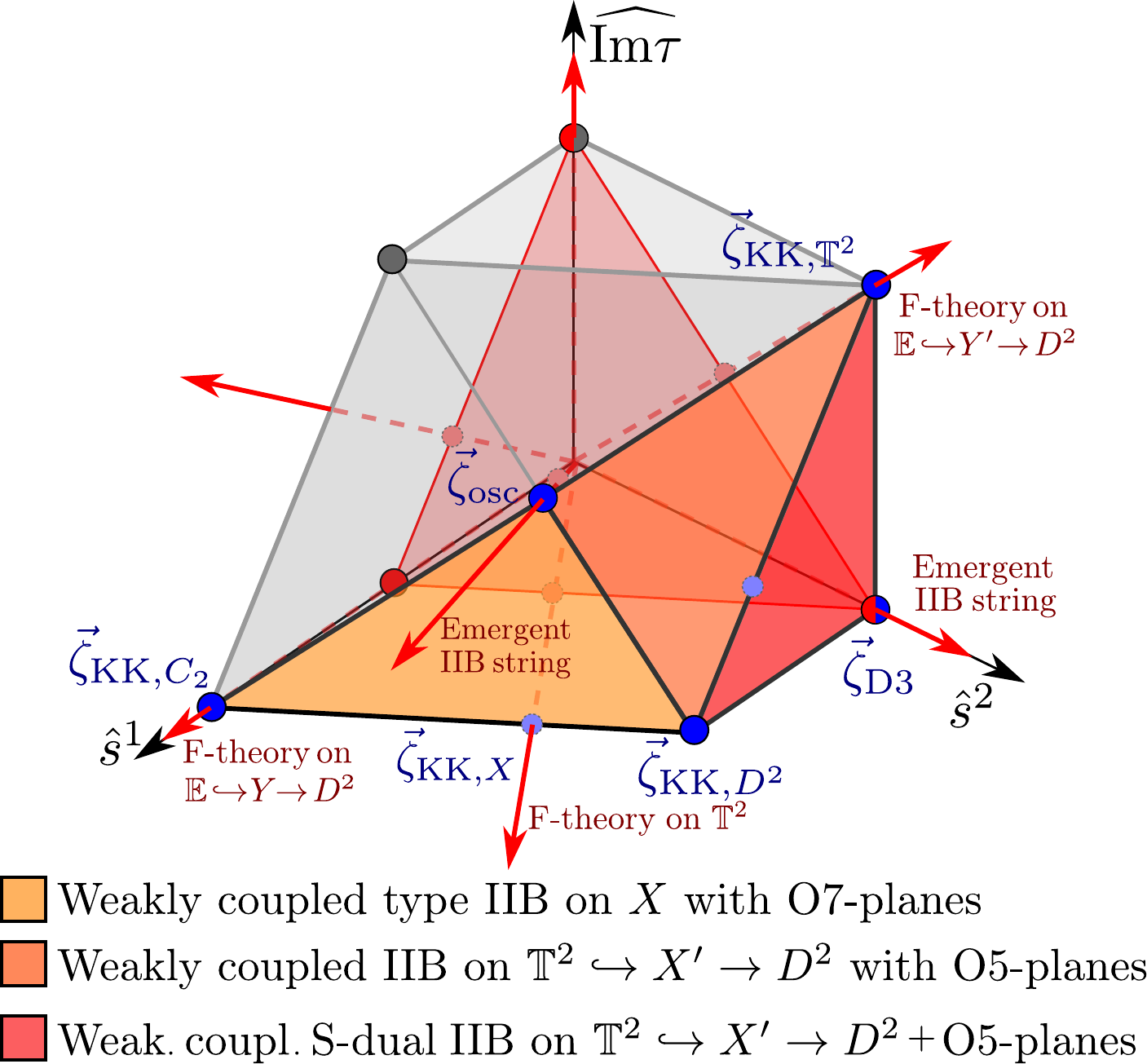}
    }
\caption{Base with $\mathbb{T}^2\hookrightarrow X\to D$ fibration.} \label{fig.bu-IIB-1}
\end{subfigure}
\begin{subfigure}[b]{0.48\textwidth}
\centering
     \resizebox{0.95\textwidth}{!}{%
    \includegraphics[scale=1]{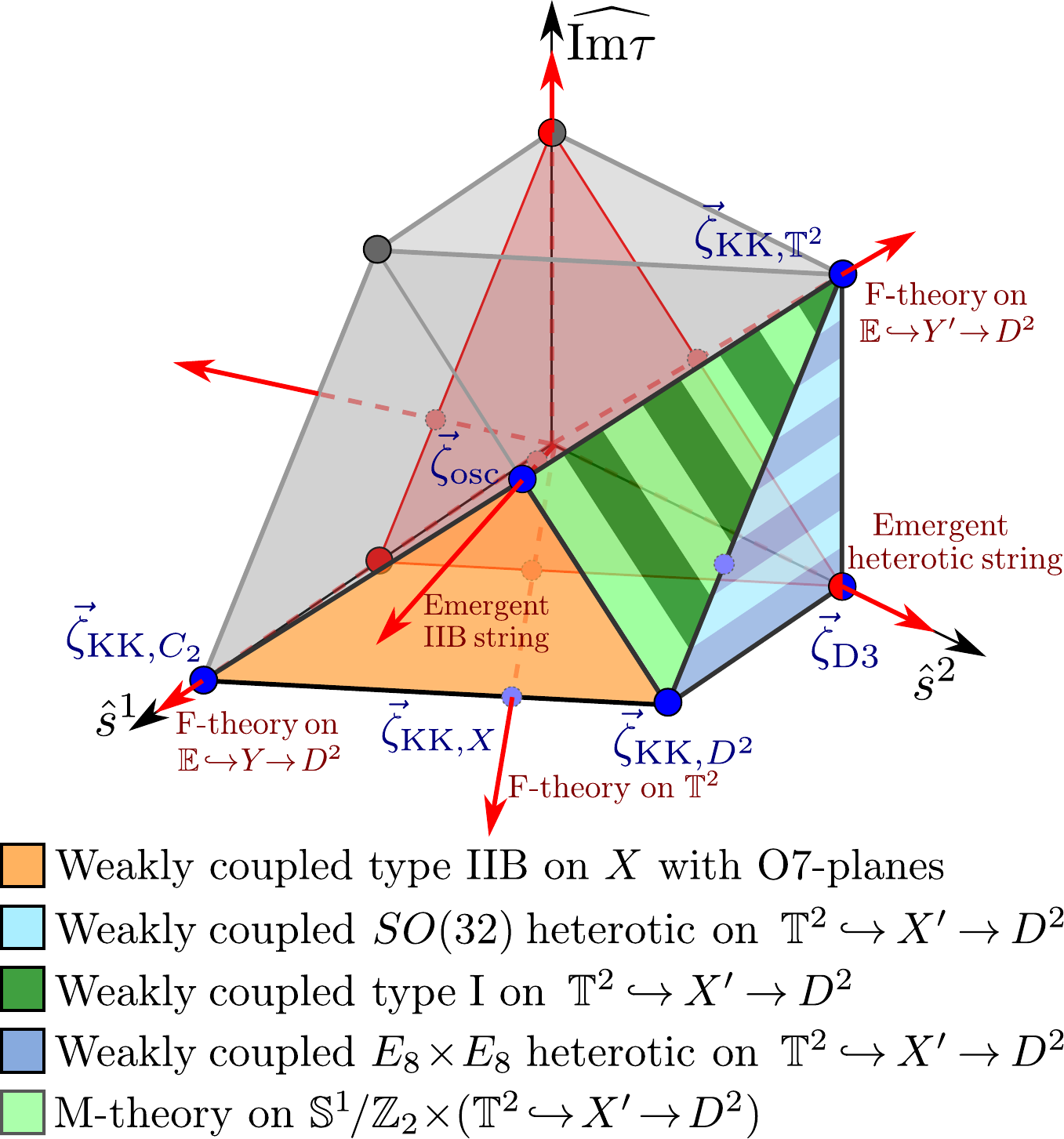}
    }
\caption{Base with $\mathbb{P}^1\hookrightarrow X\to D$ fibration.} \label{fig.bu-IIB-2}
\end{subfigure}

\caption{Arrangement of duality frames, light towers (in blue) and EFT strings (in a red hyperplane, with the EFT string flow directions and the associated dual theory at infinite distance labeled) for different volume leading forms of $\mathcal{V}_X$, in (locally) flat coordinates, for F-theory on a CY 4-fold $\mathbb{E}\hookrightarrow Y\to X$, where $V_X\sim v_1^2v_2\sim s^1\sqrt{s^2}$. Obstructed regions are depicted in gray. For simplicity we do not label the $\zeta$-vectors of the EFT strings. Bounded states of light towers and EFT strings are depicted in a lighter shade. See \cite[Section 4.3]{Grieco:2025bjy} for more details.
\label{fig.bu-IIB}}
\end{center}
\end{figure}

As explored in \cite{Kaufmann:2026fli,Kaufmann:2026mha}, not only is the perturbative description for limits where ${\rm Im}\tau$ grows much faster than the K\"ahler directions of the base, but this generically applies to the saxions associated to the complex structure moduli of the 4-fold $Y$. Pure complex structure limits are expected to be obstructed from quantum effects, in such a way that we need to rescale the K\"ahler moduli (at a rate equal or faster than the complex structure moduli) in order to reach asymptotic regions at perturbative control.

\subsection{M-theory on \texorpdfstring{$G_2$}{G2} Joyce manifolds\label{app.Joyce}}

We finally consider 11d M-theory compactified on $G_2$-manifolds, 7-dimensional Ricci-flat spaces equipped with an associative 3-form $\Phi$, which we can expand in terms of a base $\{[\Sigma_a]\}_{a=1}^{b_3(X_7)}$ for $H^3(X_7)$ as $\Omega=s^a[\Sigma_a]$, with the 4d saxions $s^a$ measuring the volume of the Poincar\'e dual associative 3-cycles $\tilde\Sigma^a$ (with $\tilde\Sigma^a\cdot\Sigma_b=\delta^a_b$) in 11d Planck units. The volume of $X$ in 11d Planck units and 4d $\mathcal{N}=1$  K\"ahler potential (again, modulo non-perturbative contributions) are expressed in terms of the associative 3-form as
\begin{equation}
	V_X=\frac{1}{7}\int_{X}\Phi\wedge\star\Phi\,,\quad K=-\log V_X\,.
\end{equation}
Both the coassociative 4-cycles and the dual associative 3-cycles are calibrated by $\star \Phi$ and $\Phi$, respectively, with EFT strings corresponding to M5-branes wrapped over the coassociative 4-cycles $\Sigma=e^a\Sigma_a$.

Generically, limits in which a single saxion $s^a$ is sent to infinity (i.e., elementary EFT string limits) correspond to asymptotically associative fibrations for which the Joyce manifold takes the form $\mathbb{T}^4\hookrightarrow X\to \mathcal{B}_3$ or ${\rm K3}\hookrightarrow X\to  \mathcal{B}_3$, with a relatively shrinking fiber and a growing base. These limits are well within perturbative control, and respectively correspond to an emergent type II or heterotic string. However, in limits where several saxions co-scale the fibration structures can be way more involved.

Rather than study the full generality of such spaces, whose geometries along asymptotic limits of their moduli spaces is way less understood than for Calabi-Yau manifolds, we take a particular subset consisting in \textbf{Joyce} compact \textbf{manifolds} (of the first kind)\cite{joyce1996a,joyce1996b,joyce2004constructing}: toroidal orbifolds $X=\mathbb{T}^7/(\mathbb{Z}_2\oplus\mathbb{Z}_2\oplus\mathbb{Z}_2)$, whose asymptotic fibrations are well known, see e.g., \cite{Liu:1998tha} or \cite[Appendix A]{Grieco:2025bjy}. For these manifolds $H_\bullet(X;\mathbb{Z})=(\mathbb{Z},0,0,\mathbb{Z}^7,\mathbb{Z}^7,0,0,\mathbb{Z},0,\dots)$, with 3-cycle basis topologically equivalent to tori, $\{\tilde\Sigma_a\simeq\mathbb{T}^3\}_{a=1}^7$ of $H_3(X;\mathbb{Z})$. While after quotienting by $\mathbb{Z}_2\oplus\mathbb{Z}_2\oplus\mathbb{Z}_2$ the manifold ceases to have 1-cycles, we can identify the size of the 3-cycles in terms of the original radii of the $\mathbb{T}^7$.\footnote{The quotient $\mathbb{T}^7/(\mathbb{Z}_2\oplus\mathbb{Z}_2\oplus\mathbb{Z}_2)$ results in singularities consisting in $\sqcup_{i=1}^{12}\mathbb{T}^3$, whose resolution results in the introduction of 36 extra cycles \cite{Liu:1998tha,Lanza:2021udy}, whose volume is controlled by additional \emph{twisted saxions} $\{\tilde{s}^\alpha\}_{\alpha=1}^{36}$. For fixed values their contribution to the overall volume is finite and thus in asymptotic limits of the initial saxions can be safely ignored.} For example, for a particular choice   
\begin{equation}
	\Phi=\eta_{123}+\eta_{145}+\eta_{167}+\eta_{246}-\eta_{257}-\eta_{347}-\eta_{356}\,\quad\text{with}\;\eta_{abc}=R_aR_bR_c\dd y^a\wedge\dd y^b\dd y^c\,,
\end{equation}
where $y^a\sim y^a+1$ and $\{R^a\}_{a=1}^7$ are the initial radii of $\mathbb{T}^7$ in $M_{\rm Pl,11}$ units, we have \cite{Lanza:2021udy,Castellano:2023jjt}
\begin{equation}\label{eq.radii saxion joyce}
	s^a=R_{a_I}R_{a_J}R_{a_K}\qquad\text{with}\quad\begin{array}{c|ccccccc}
	a&1&2&3&4&5&6&7\\\hline
	a_I&1&1&1&2&2&3&3\\
	a_J&2&4&6&4&5&4&5\\
	a_K&3&5&7&6&7&7&6
	\end{array}\;,
\end{equation}
 and modulo corrections from twisted saxions from the resolutions, $V_X=\frac{1}{8}R_1\dots R_7\sim(s^1\dots s^7)^{1/3}$. Now, apart from the coassociative fibrations  $\mathbb{T}^4\hookrightarrow X\to \mathcal{B}_3$ or ${\rm K3}\hookrightarrow X\to  \mathcal{B}_3$, for Joyce manifolds we can have additional asymptotic fibrations for non-elementary saxionic limits, such as \cite{Liu:1998tha}
 \begin{equation}
 	\mathbb{T}^3\hookrightarrow X\to\mathbb{T}^4\;,\quad \mathbb{T}^3\hookrightarrow X\to{\rm K3}\quad\text{or}\quad Y_6\hookrightarrow X\to\mathbb{S}^1/\mathbb{Z}_2\,,
 \end{equation}
 where $Y_6$ its a Calabi-Yau 3-fold which can further present asymptotic fibrations such as
 \begin{equation}
 	\mathbb{T}^3\hookrightarrow Y_6\to Q_3\quad\text{or}\quad {\rm K3}\hookrightarrow Y_6\to\mathbb{T}^2/\mathbb{Z}_2\,,
 \end{equation}
 with $Q_3$ and $\mathbb{T}^2/\mathbb{Z}_2$ orbifolds homeomorphic to $\mathbb{S}^3$ and $\mathbb{S}^2$, respectively. On top of the different fibrations, we can have distinct limits where the relative growth of fiber and base changes, see \cite[Section 5]{Grieco:2025bjy}. This includes limits where the fiber grows relatively to the base, which could in principle result in problematic quantum obstructions such as those of \cite{Witten:1996mz,Cvetic:2024wsj,Kaufmann:2026fli,Kaufmann:2026mha,Kaufmann:2026tsy}. An exhaustive study of the possible obstruction of infinite distance limits in $G_2$-manifolds has not been performed in the literature, and indeed would be required to understand which of these limits correspond to perturbative regimes. In \cite{Grieco:2025bjy} the authors performed a prospective study of the different limits, taking as an assumption that said limit indeed existed. Then the light tower of states would correspond to the new degrees of freedom of the dual theory emerging at infinite distance, thus proposing a set of new duality frames completing the K\"ahler cone. These are depicted in Figures \ref{fig.buG2-2} and \ref{fig.buG2-3}.

\begin{figure}[H]
\begin{center}
\begin{subfigure}[b]{0.32\textwidth}
\captionsetup{width=.95\linewidth}
\center
\includegraphics[width=\textwidth]{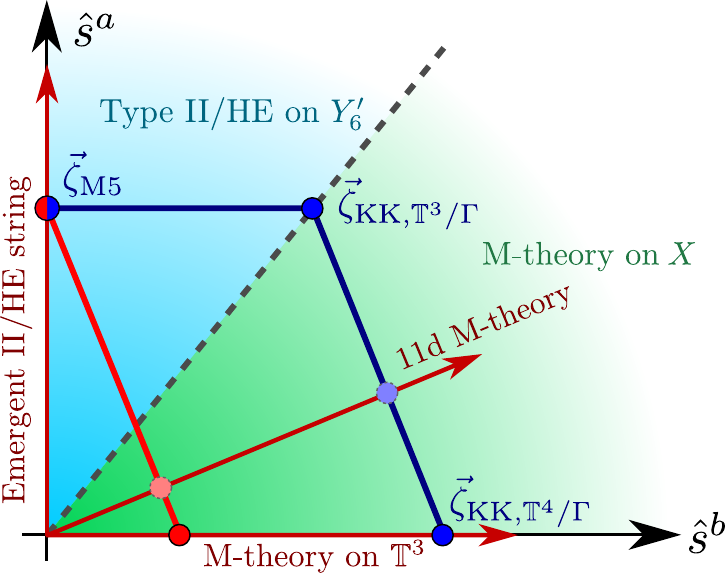}
\caption{\hspace{-0.3em}\scriptsize $(s^a,s^b)\sim(s^{a_1},s^{a_2}\dots s^{a_7})$.}\label{f.SlicesMth-1}
\end{subfigure}
\begin{subfigure}[b]{0.32\textwidth}
\captionsetup{width=.95\linewidth}
\center
\includegraphics[width=\textwidth]{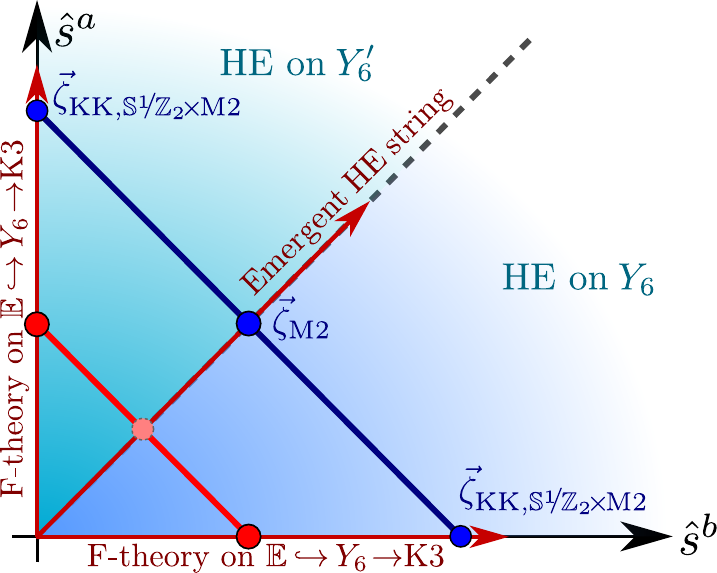}
\caption{\hspace{-0.3em} $(s^a,s^b)\sim(s^{a_1}s^{a_2},s^{a_3}s^{a_4})$.} \label{f.SlicesMth-2}
\end{subfigure}
\hfill
\begin{subfigure}[b]{0.32\textwidth}
\captionsetup{width=.95\linewidth}
\center
\includegraphics[width=\textwidth]{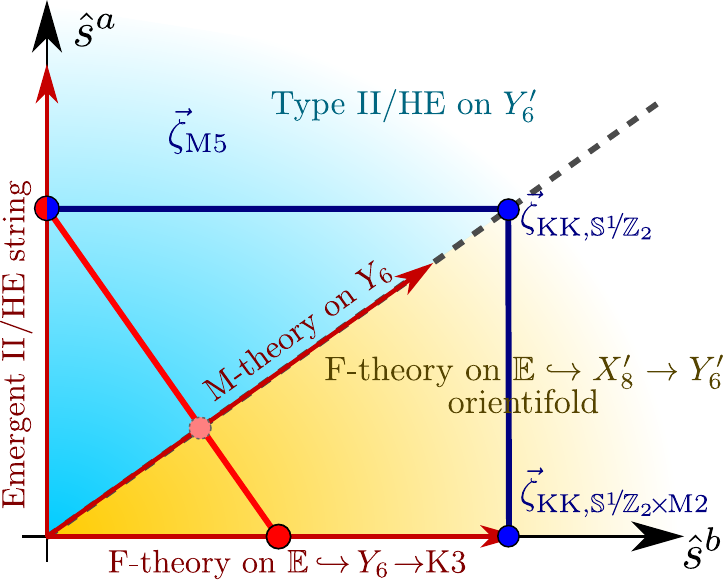}
\caption{\hspace{-0.3em} $(s^a,s^b)\sim(s^{a_1},s^{a_2}s^{a_3})$}\label{f.SlicesMth-3}
\end{subfigure}
\begin{subfigure}[b]{0.32\textwidth}
\captionsetup{width=.95\linewidth}
\center
\includegraphics[width=\textwidth]{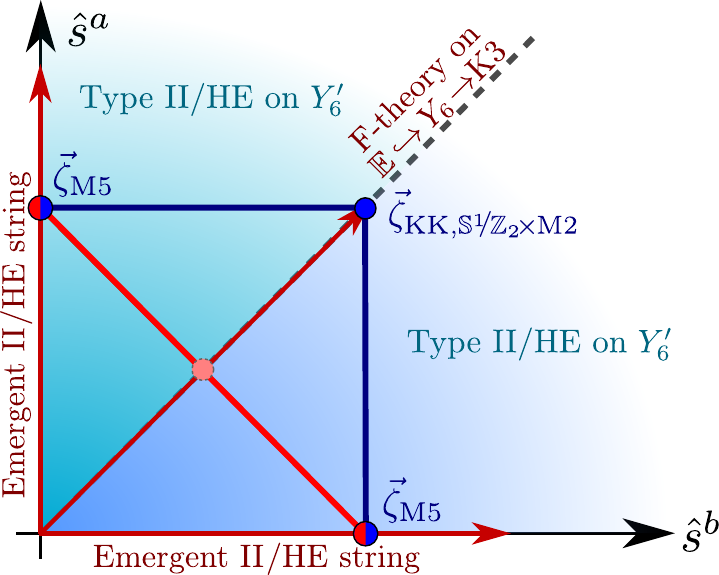}
\caption{\hspace{-0.3em} $(s^a,s^b)\sim(s^{a_1},s^{a_2})$.} \label{f.SlicesMth-4}
\end{subfigure}
\hfill
\begin{subfigure}[b]{0.32\textwidth}
\captionsetup{width=1.1\linewidth}
\center
\includegraphics[width=\textwidth]{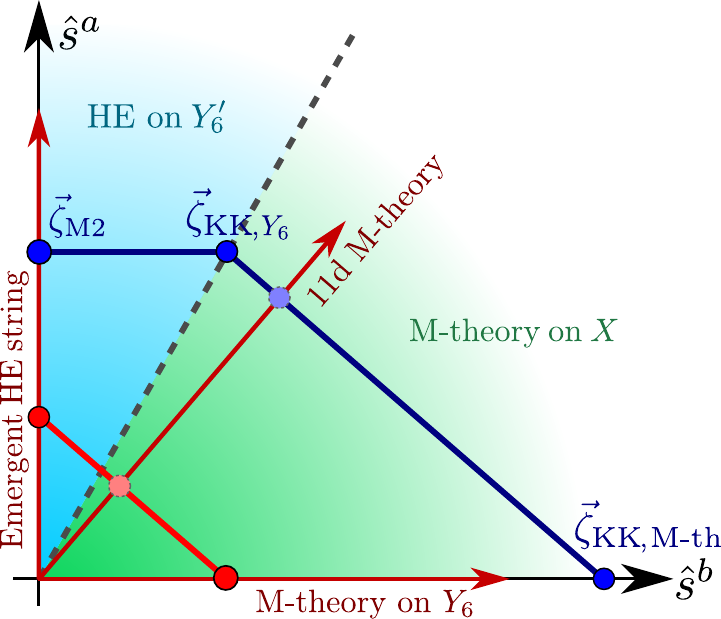}
\caption{\hspace{-0.3em} $(s^a,s^b)\sim(s^{a_1}\dots s^{a_4},s^{a_5}s^{a_6} s^{a_7})$}\label{f.SlicesMth-5}
\end{subfigure}
\begin{subfigure}[b]{0.32\textwidth}
\captionsetup{width=.95\linewidth}
\center
\includegraphics[width=\textwidth]{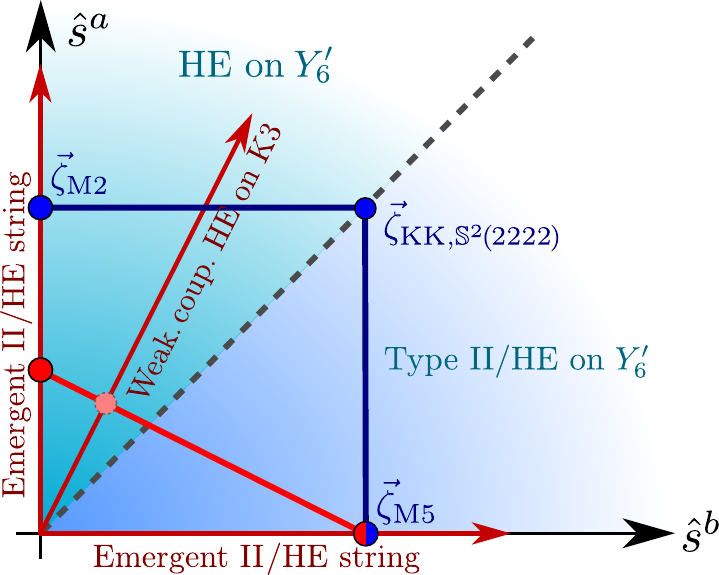}
\caption{\hspace{-0.3em} $(s^a,s^b)\sim (s^{a_1}\dots s^{a_4},s^{a_5})$} \label{f.SlicesMth-6}
\end{subfigure}
\hfill
\begin{subfigure}[b]{0.32\textwidth}
\captionsetup{width=1.1\linewidth}
\center
\includegraphics[width=\textwidth]{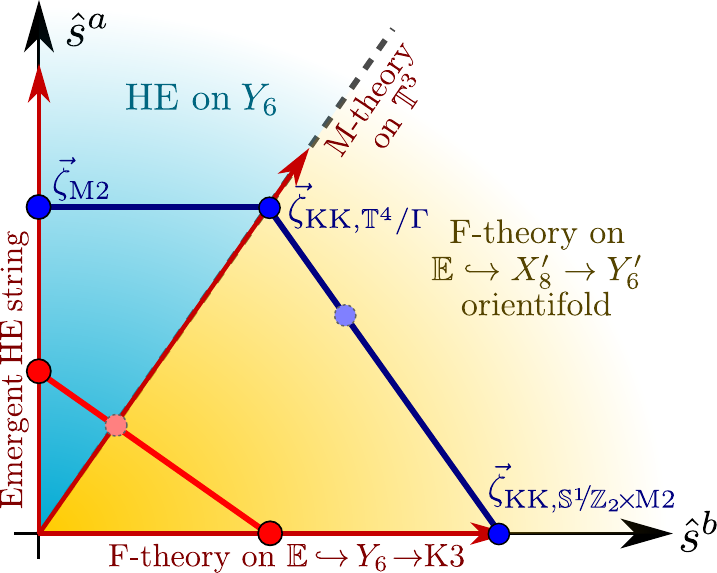}
\caption{\hspace{-0.3em} $(s^a,s^b)\sim (s^{a_1}\dots s^{a_4},s^{a_5}s^{a_6})$.} \label{f.SlicesMth-7}
\end{subfigure}
\begin{subfigure}[b]{0.32\textwidth}
\captionsetup{width=1.1\linewidth}
\center
\includegraphics[width=\textwidth]{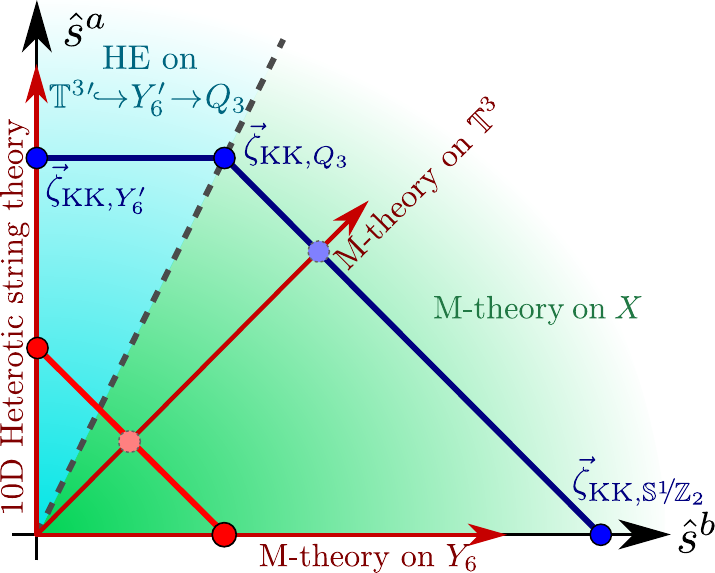}
\caption{\hspace{-0.3em} $(s^a,s^b)\sim (s^{a_1}s^{a_2} s^{a_3},s^{a_4}s^{a_5}s^{a_6})$.} \label{f.SlicesMth-8}
\end{subfigure}
\begin{subfigure}[b]{0.32\textwidth}
\captionsetup{width=1.1\linewidth}
\center
\includegraphics[width=\textwidth]{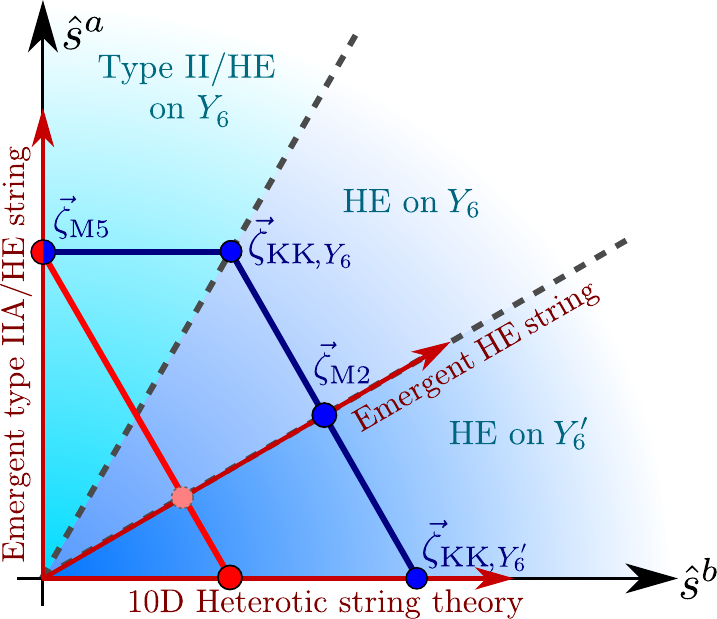}
\caption{\hspace{-0.3em} $(s^a,s^b)\sim (s^{a_1},s^{a_2} s^{a_3}s^{a_4})$.} \label{f.SlicesMth-9}
\end{subfigure}
\hfill
		\caption{Arrangement of duality frames, light towers (in blue) and EFT strings (in a red hyperplane, with the EFT string flow directions and the associated dual theory and infinite distance labeled) for different 2-dimensional slices of the $\{s^1,\dots,s^7\}$ moduli space of M-theory on $\mathbb{T}^7/(\mathbb{Z}_2\oplus\mathbb{Z}_2\oplus\mathbb{Z}_2)$. Here $(s^a,s^b)\sim(s^{a_1}\dots s^{a_k},s^{a_{k+1}}\dots s^{a_{k+l}})$ denotes the slice of moduli space where $k$ and $l$ saxions are sent to infinity homogeneously, with the new saxionic coordinates $s^a$ and $s^b$ controlling the respective growth, while the remaining $7-k-l$ saxions are kept fixed. Bounded states of light towers and EFT strings are depicted in a lighter shade. All quantities are depicted in canonically normalized coordinates. See \cite[Section 6]{Grieco:2025bjy} for more details.}
		 
			\label{fig.buG2-2}
	\end{center}
\end{figure}

 \begin{figure}[H]
\begin{center}
\begin{subfigure}[b]{0.40\textwidth}
\captionsetup{width=\linewidth}
\center
\includegraphics[width=\textwidth]{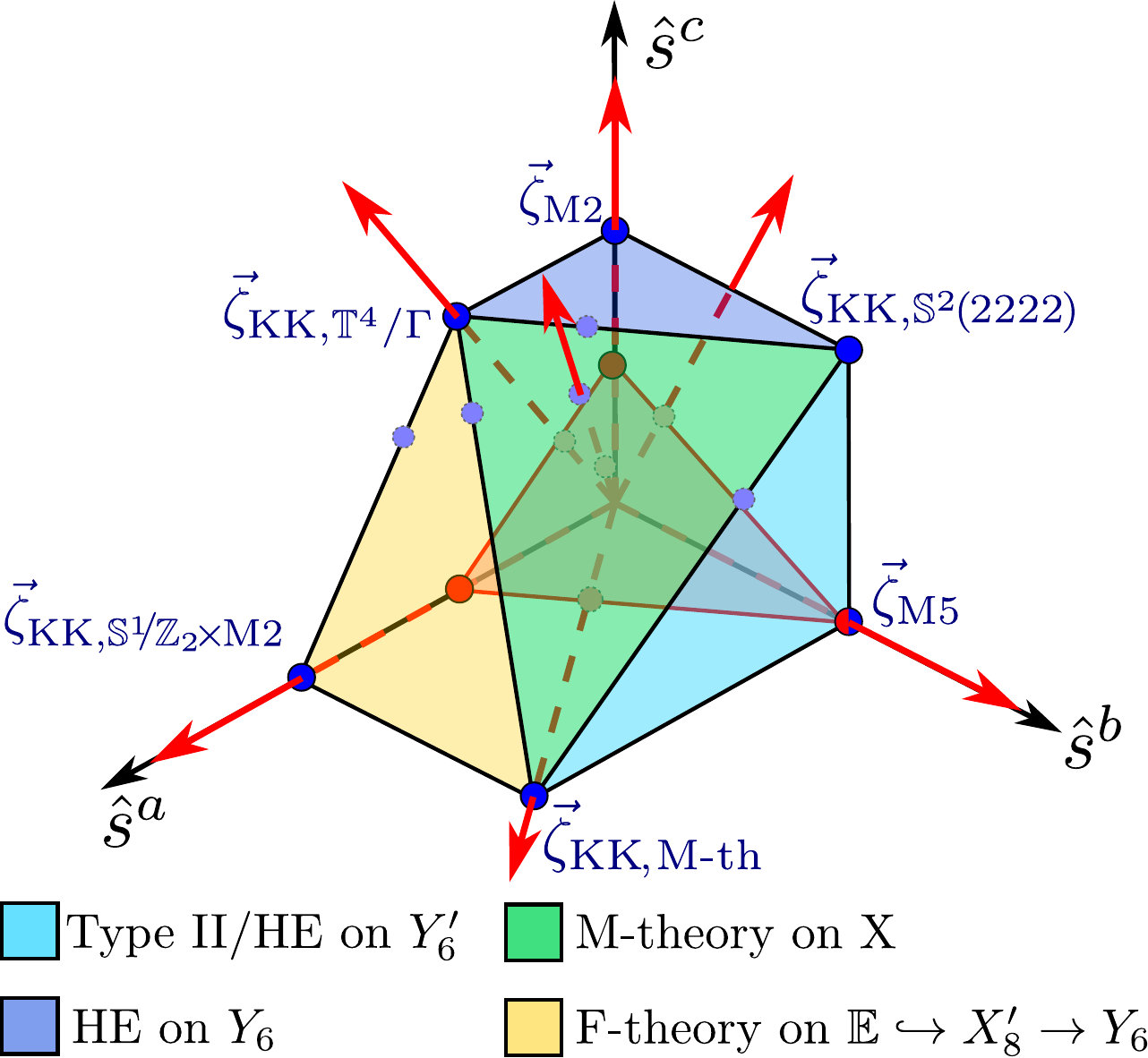}
\caption{\hspace{-0.3em} $(s^a,s^b,s^c)\sim(s^{a_1}s^{a_2},s^{a_3}\dots s^{a_6},s^{a_7})$.} \label{f.slicesMth3-1}
\end{subfigure}
\begin{subfigure}[b]{0.40\textwidth}
\captionsetup{width=\linewidth}
\center
\includegraphics[width=\textwidth]{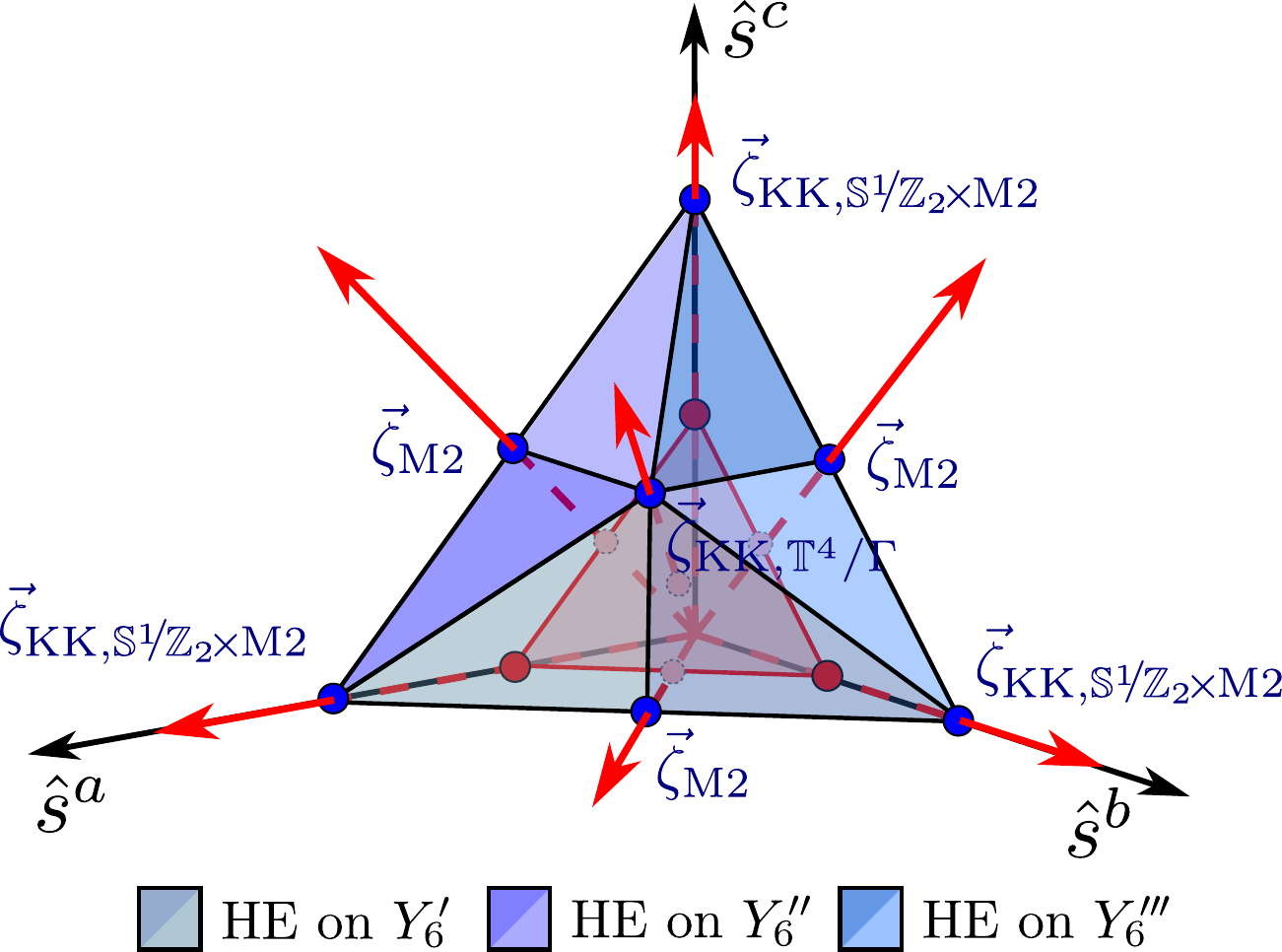}
\caption{\hspace{-0.3em} $(s^a,s^b,s^c)\sim(s^{a_1}s^{a_2},s^{a_3}s^{a_4},s^{a_5}s^{a_6})$.} \label{f.slicesMth3-2}
\end{subfigure}
\begin{subfigure}[b]{0.40\textwidth}
\captionsetup{width=\linewidth}
\center
\includegraphics[width=\textwidth]{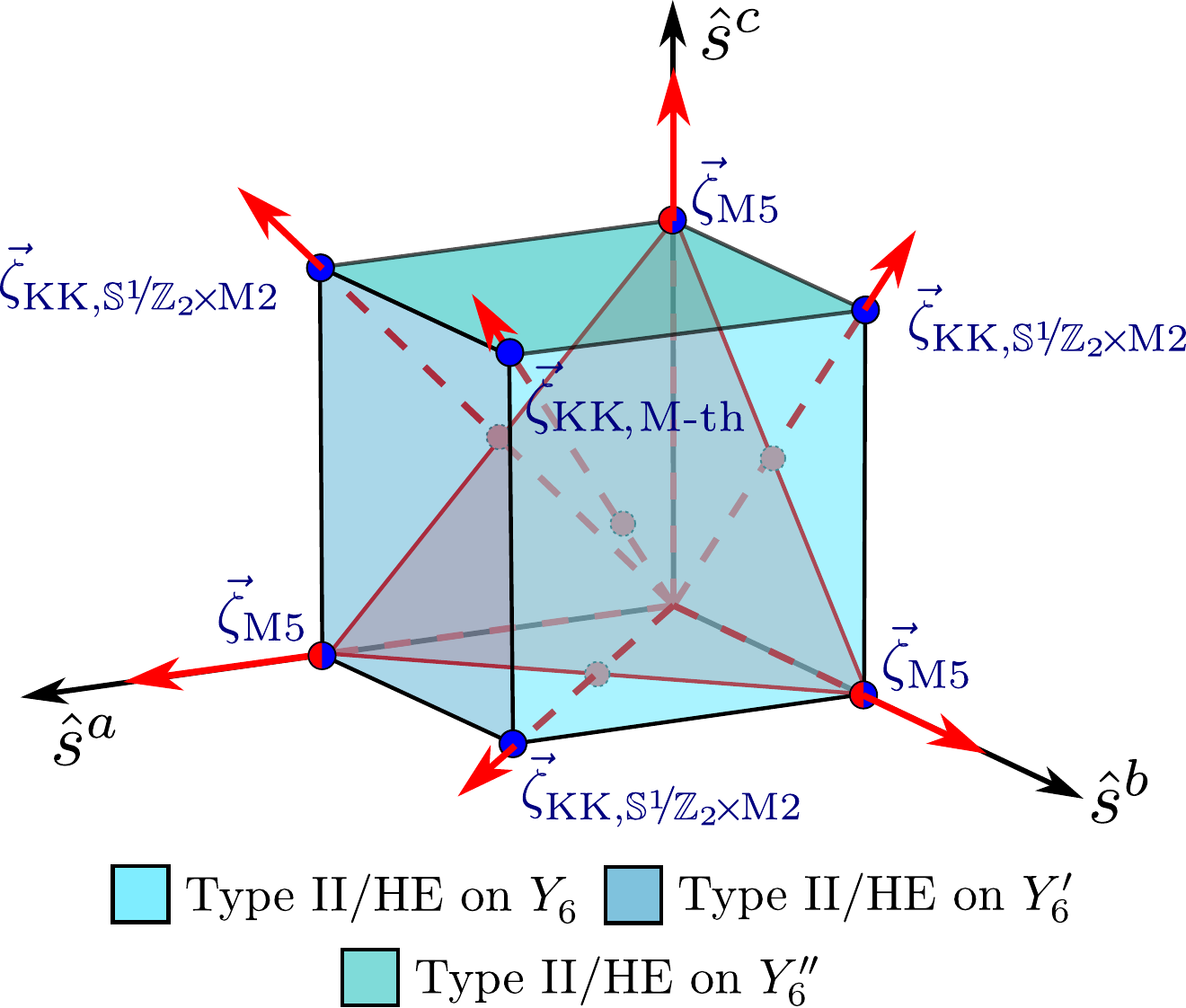}
\caption{\hspace{-0.3em} $(s^a,s^b,s^c)\sim(s^{a_1},s^{a_2},s^{a_3})$.} \label{f.slicesMth3-3}
\end{subfigure}
\begin{subfigure}[b]{0.40\textwidth}
\captionsetup{width=\linewidth}
\center
\includegraphics[width=\textwidth]{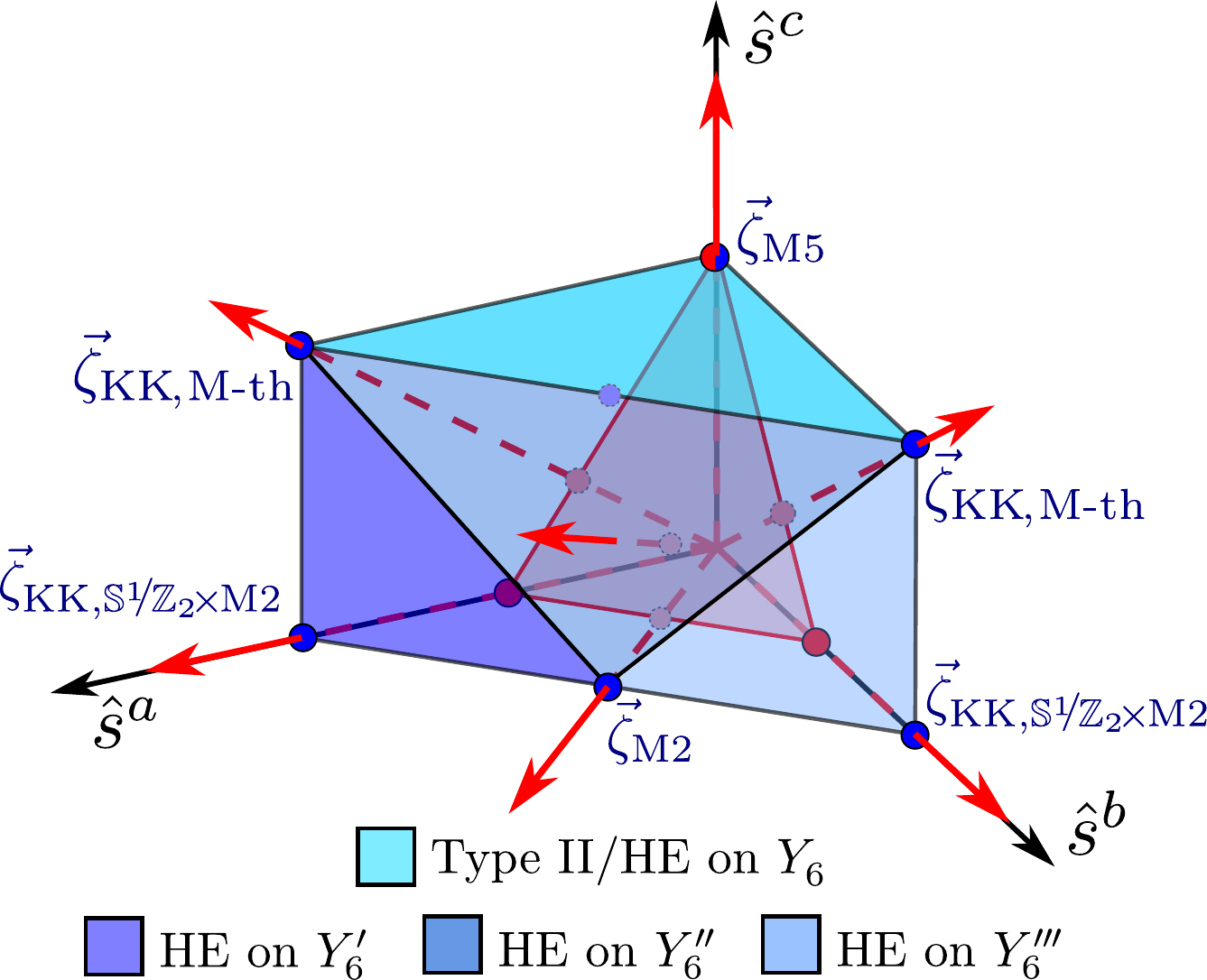}
\caption{\hspace{-0.3em} $(s^a,s^b,s^c)\sim(s^{a_1},s^{a_2}s^{a_3},s^{a_4}s^{a_5})$.} \label{f.slicesMth3-4}
\end{subfigure}
\begin{subfigure}[b]{0.55\textwidth}
\captionsetup{width=\linewidth}
\center
\includegraphics[width=\textwidth]{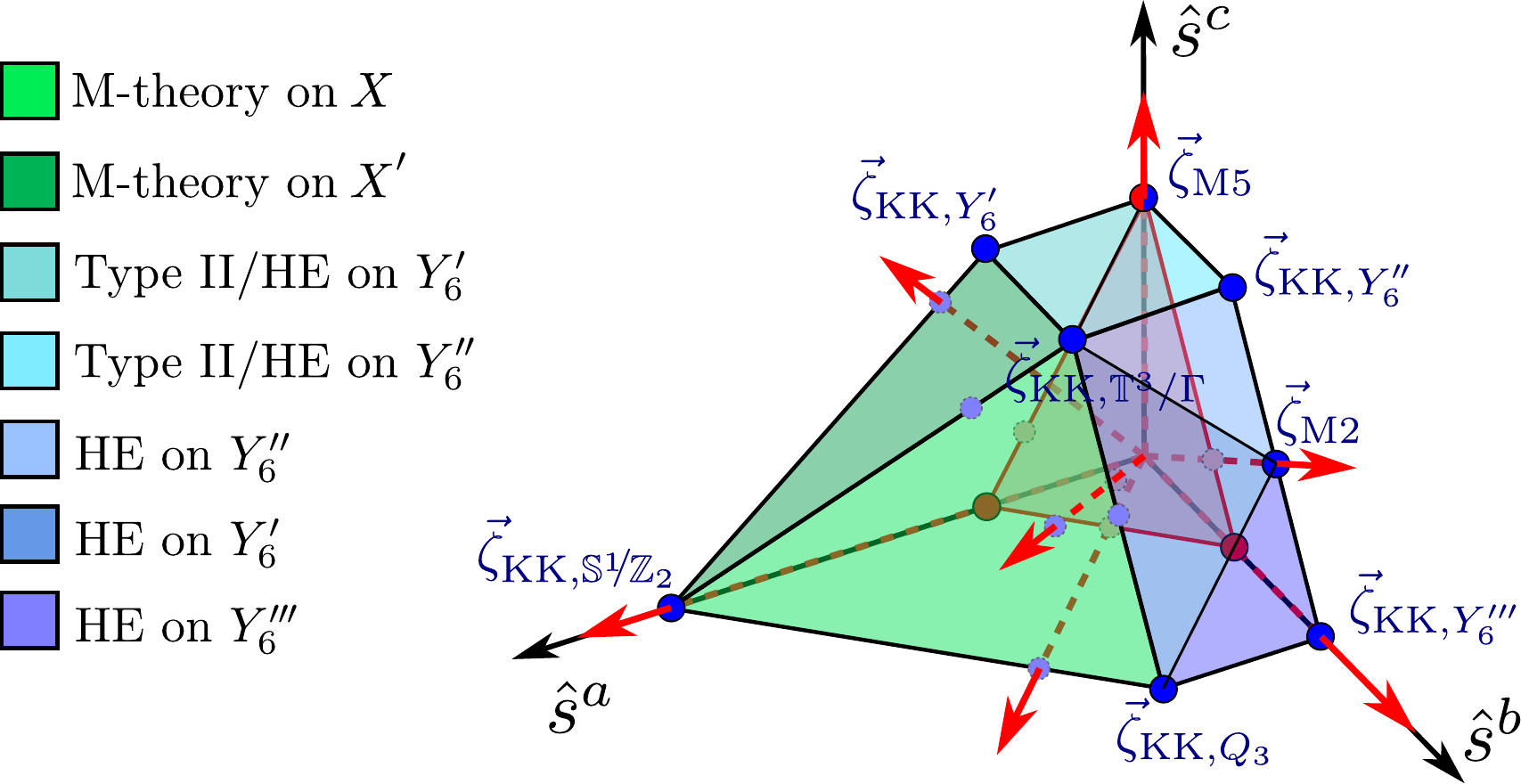}
\caption{\hspace{-0.3em} $(s^a,s^b,s^c)\sim(s^{a_1}s^{a_2}s^{a_3},s^{a_4}s^{a_5}s^{a_6},s^{a_7})$.} \label{f.slicesMth3-5}
\end{subfigure}
		\caption{Arrangement of duality frames, light towers (in blue) and EFT strings (in a red hyperplane, with the EFT string flow directions and the associated dual theory and infinite distance labeled) for different 3-dimensional slices of the $\{s^1,\dots,s^7\}$ moduli space of M-theory on $\mathbb{T}^7/(\mathbb{Z}_2\oplus\mathbb{Z}_2\oplus\mathbb{Z}_2)$. All quantities are depicted in canonically normalized coordinates, with $(s^a,s^b,s^c)\sim(s^{a_1}\dots s^{a_k},s^{a_{k+1}}\dots s^{a_{k+l}},s^{a_{k+l+1}}\dots s^{a_{k+l+h}})$ denoting the $k+l+h$ slice of moduli space with co-scaling saxions as in Figure \ref{fig.buG2-2}. Bounded states of light towers and EFT strings are depicted in a lighter shade. All quantities are depicted in canonically normalized coordinates. See \cite[Section 6]{Grieco:2025bjy} for more details.}
			\label{fig.buG2-3}
	\end{center}
\end{figure}

\bibliographystyle{JHEP}

\bibliography{refs}

\end{document}